\documentclass[twoside,leqno,twocolumn]{article}

\usepackage[letterpaper]{geometry}
\usepackage{ltexpprt}

\usepackage[parfill]{parskip}
\usepackage[numbers]{natbib}
\usepackage[utf8]{inputenc} % allow utf-8 input
\usepackage[T1]{fontenc}    % use 8-bit T1 fonts

\PassOptionsToPackage{hyphens}{url}\usepackage{hyperref}

\usepackage{url}            % simple URL typesetting
\usepackage{booktabs}       % professional-quality tables
\usepackage{nicefrac}       % compact symbols for 1/2, etc.
\usepackage{microtype}      % microtypography
\usepackage{bm,bbm}
\usepackage{graphicx}
\usepackage{amsfonts}
\usepackage{float}
\usepackage{algorithm}
\usepackage{algpseudocode}
\usepackage{mathtools}
\usepackage{xcolor}
\usepackage{xspace}
\usepackage{tikz}
\usetikzlibrary{positioning}
\usetikzlibrary{shapes,arrows}
\usepackage{authblk}

\DeclarePairedDelimiter\norm{\lVert}{\rVert}%
\DeclareMathSymbol{\R}{\mathord}{AMSb}{"52}
\DeclareMathOperator*{\argmax}{arg\,max}

\providecommand{\GS}{\ensuremath{G_\mathcal{S}}\xspace}
\providecommand{\GSPrime}{\ensuremath{G_\mathcal{S'}}\xspace}
\providecommand{\Adj}{\ensuremath{\bm{A}}\xspace}
\providecommand{\AdjS}{\ensuremath{\bm{A}_\mathcal{S}}\xspace}

\providecommand{\TildeAdj}{\ensuremath{\tilde{\bm{A}}}\xspace}
\providecommand{\TildeAdjS}{\ensuremath{\tilde{\bm{A}}_\mathcal{S}}\xspace}

\providecommand{\TildeD}{\ensuremath{\tilde{\bm{D}}}\xspace}
\providecommand{\TildeL}{\ensuremath{\tilde{\bm{L}}}\xspace}
\providecommand{\PerturG}{\ensuremath{\tilde{G}}\xspace}
\providecommand{\PerturGS}{\ensuremath{\tilde{G}_\mathcal{S}}\xspace}
\providecommand{\Perturb}{\ensuremath{\bm{\Delta}}\xspace}
\providecommand{\SET}[1]{\ensuremath{\{ #1 \}}\xspace}

\providecommand{\TildeAdjEigvec}{\ensuremath{\tilde{\bm{v}}_1}\xspace}
\providecommand{\SPrime}{\ensuremath{\mathcal{S'}}\xspace}
\providecommand{\TildeDegDist}{\ensuremath{\tilde{\bm{d}}}\xspace}
\providecommand{\TildeDegDistS}{\ensuremath{\tilde{\bm{d}}}_\mathcal{S}\xspace}
\providecommand{\DegDist}{\ensuremath{\bm{d}}\xspace}

\DeclareMathOperator{\Tr}{Tr}

\providecommand{\ModelName}{POTION }
\providecommand{\AlgoName}{POTION-ALG }

\allowdisplaybreaks

\usepackage[suppress]{color-edits} % use this command to get a submission ready version
\addauthor[Sixie]{sy}{blue} % please add your name and favorite color 
\addauthor[Leo]{leo}{teal}
\addauthor[Tina]{tina}{red}

\begin{document}

\title{\Large \ModelName: Optimizing Graph Structure for Targeted Diffusion}

\author[1]{Sixie Yu \thanks{The first two authors contributed equally to the paper.} \thanks{sixie.yu@wustl.edu.}}
\author[2]{Leo Torres\thanks{leo@leotrs.com}}
\author[3]{Scott Alfeld\thanks{salfeld@amherst.edu}}
\author[2]{Tina Eliassi-Rad \thanks{t.eliassirad@northeastern.edu}}
\author[1]{Yevgeniy Vorobeychik \thanks{yvorobeychik@wustl.edu}}
\affil[1]{Washington University in St. Louis}
\affil[2]{Northeastern University}
\affil[3]{Amherst College}

\date{}

\maketitle

% Copyright Statement
% When submitting your final paper to a SIAM proceedings, it is requested that you include 
% the appropriate copyright in the footer of the paper.  The copyright added should be 
% consistent with the copyright selected on the copyright form submitted with the paper.
% Please note that "20XX" should be changed to the year of the meeting.

% Default Copyright Statement
\fancyfoot[R]{\scriptsize{Copyright \textcopyright\ 2021 by SIAM\\
Unauthorized reproduction of this article is prohibited}}

\pagenumbering{arabic}

\begin{abstract}
The problem of diffusion control on networks has been extensively studied, with applications ranging from marketing to controlling infectious disease. 
However, in many applications, such as cybersecurity, an attacker may want to attack a \emph{targeted} subgraph of a network, while limiting the impact on the rest of the network in order to remain undetected.
We present a model \ModelName in which the principal aim is to optimize graph structure to achieve such targeted attacks.
We propose an algorithm \AlgoName for solving the model at scale, using a gradient-based approach that leverages Rayleigh quotients and pseudospectrum theory.
In addition, we present a condition for certifying that a targeted subgraph is immune to such attacks. 
Finally, we demonstrate the effectiveness of our approach through experiments on real and synthetic networks.
\end{abstract}

\section{Introduction}

Many diverse phenomena that propagate through a network, such as epidemic spread, cascading failures, and \leodelete{networks of} chemical reactions, can be modeled by network diffusion models~\cite{bailey1975mathematical,chakrabarti2008epidemic,motter2002cascade,yang2017small,Vu19}.
The problem of controlling diffusion has, as a result, received \leoreplace{prominent}{much} attention in the literature, with primary focus on two mechanisms for control: the choice of initial nodes to start the spread~\cite{kempe2003maximizing,chen2009efficient,zhang2016data}, and the modification of network structure~\cite{le2015met,tong2012gelling,yu2019removing,torres2020node}.
To date, most work on diffusion control (either promotion or inhibition) has considered diffusion over the entire network.
However, in many problems, the focus is instead on diffusion that is \emph{targeted} to a particular subgraph of the network.
For example, in cybersecurity, diffusion commonly represents malware spread, but malware attacks are often targeted at particular subsets of critical devices~\cite{Haghtalab17}, which should be accounted for \syreplace{in vulnerability analysis}{when modeling attacking behavior}.
Congestion cascades of ground traffic or flight networks are other examples, where the goal of resilience may be to ensure that cascades\leodelete{, should they occur,} concentrate on a subset of high-capacity nodes that can handle them, limiting the impact on the rest of the network~\cite{fleurquin2013systemic,estrada2020hubs}.
In another domain, medical treatments for certain diseases such as cancer may leverage a molecular signaling network, with the goal of targeting just the pathogenic portion of it, while limiting the deleterious effects on the rest~\cite{Wang19}.

We study the problem of targeted diffusion in which an attacker\footnote{The attacker is the agent who initiates diffusion.} can modify the graph structure $G=(\mathcal{V}, \mathcal{E})$ to achieve two goals: 1) maximize the diffusion spread to a target subgraph $\GS$, and 2) minimize the impact on the remaining graph $G \setminus \GS$.
We capture the first goal by maximizing a utility function that incorporates spectral information of the adjacency matrix of $G$, specifically its largest (in \syreplace{module}{magnitude}) eigenvalue, eigenvector centrality, and the normalized cut of the target subgraph.
The second goal is achieved by limiting the modifications made outside of the target subgraph.
% The second goal is achieved explicitly by limiting the modifications made outside of the target subgraph.
% and implicitly through the eigenvector centrality term.
We present a scalable algorithmic framework for solving this problem. Our framework leverages a combination of gradient ascent with the use of Rayleigh quotients and pseudospectrum theory, which yields differentiable approximations of our objective and allows us to avoid projection steps that would otherwise be costly and imprecise.
Moreover, we derive a condition that enables us to certify if a network is robust against a broad class of targeted diffusion attacks.
% attacks.
Finally, we demonstrate the effectiveness of our approach through extensive experiments.

In summary, our contributions are:
\begin{enumerate}
    \item We propose \ModelName (oPtimizing graph structures fOr Targeted diffusION): a model for targeted diffusion attack by optimizing graph structures.
    \item We present \AlgoName: an efficient algorithm to optimize \ModelName by leveraging Rayleigh quotient and pseudospectrum theory.
    \item We describe a condition for certifying that a targeted subgraph is immune to such attacks.
    \item We demonstrate the effectiveness and efficiency of \ModelName and \AlgoName on synthetic and real-world networks; and against baseline and competing methods.\footnote{The code to replicate the experimental results is at \url{https://github.com/marsplus/POTION}.}
\end{enumerate}

% \textcolor{red}{Tina wrote: In data mining papers, the Introduction Section commonly ends with a list of contributions and the outline of the paper. Also we NEED a name for our model and a name for our proposed algorithm/approach.}

%For example, consider the graph $G$ where each node is an airport and each (weighted) edge represents the flow of passengers from one airport to another.
%In this setting, spreading dynamics have been used to study the spread of flight delays as a cascading failure diffusing through this system \cite{fleurquin2013systemic,estrada2020hubs}.
%Our algorithm can be applied in this setting to discover the optimal structure of the airport network that guarantees that the cascading failures caused by flight delays are routed to a specific subset of airports rather than the whole network.
%For example, the targeted subset of airports could be chosen to be those airports that are better equipped to handle flight delays.
%Another example concerns a cybsersecurity setting, where $G$ represents the internal network of computers inside a company or agency, and the edges between them record the flow of information (e.g. network packets, or email communications).
%In this case, we imagine a malicious agent who is adversarially trying to infect with malware a particular subset $\mathcal{S}$ of the network without being detected.

\section{Related Work} \label{sec:related}
Various dynamical processes can be modeled as diffusion dynamics on networks, including the spread of infectious diseases~\cite{bailey1975mathematical,chakrabarti2008epidemic}, cascading failures in infrastructure networks~\cite{motter2002cascade,yang2017small}, and information spread (e.g., rumors, fake news) on social networks~\cite{leskovec2007patterns,leskovec2009meme}.
One line of research assesses the impact of cascading failures.
\citet{yang2017small} simulated cascading failures to quantify the vulnerability of the power grid in North America.
\citet{fleurquin2013systemic} studied the impact of flight delays as a cascading failure diffusing through the network.
\citet{motter2002cascade} investigated the cascading failures on a network due to the malfunction of a single node. 
%These studies share the conclusion that complex networks are vulnerable in the sense  that the malfunction of a small portion of the network can initiate large-scale cascading failures.
%This points to the importance of studying \emph{targeted} diffusion for the purpose of evaluating the vulnerability of the targeted sub-network as well as investigating the conditions under which a dynamical process can be confined within the targeted subgraph.
Another line of research concerns diffusion control, for example, 
selecting a set of nodes such that if the diffusion originated from them, it reaches as many nodes as possible~\cite{kempe2003maximizing,chen2009efficient,zhang2016data}; or modifying network structures to increase or limit some diffusion~\cite{tong2012gelling,preciado2013optimal}.
However, these lines of research do not differentiate between targeted and non-targeted nodes.
\citet{ho2015control} studied targeted diffusion controlled by changing nodal status. 
\syreplace{We focus on optimizing underlying network structures.}{We focus on the problem where an attacker manipulates underlying network structures in order to achieve targeted diffusion.} 

Another relevant research thread is network design, which is the problem of modifying network structure to induce certain desirable outcomes. Some prior work
\cite{tong2012gelling,yu2019removing,tong2010vulnerability} considered the containment of spreading dynamics by adding or removing nodes or edges from the network, while others \cite{van2011decreasing,saha2015approximation,le2015met,ChenTPTEFC16,torres2020node} considered limiting the spread of infectious disease by minimizing the largest eigenvalue of the network. 
\citet{kempe2020inducing} studied modifying network structure to induce certain outcomes from a game-theoretic perspective, but they did not consider diffusion dynamics.
Others have studied the problem of manipulating node centrality measures (e.g., eigenvector or PageRank centrality) \cite{avrachenkov2006effect,amelkin2019fighting} or node similarity measures (e.g., Katz similarity) \cite{zhou2019attacking} through edge perturbation.
All of these prior efforts %are optimizing certain objective by modifying network structure but they
focus on the impact either at the network level or at the node-level properties, while our focus is on the impact of diffusion dynamics on a targeted subgraph of the network.

\section{\ModelName: Proposed Model}\label{sec:model}

We  present a model for targeted diffusion through graph structure optimization. We refer to the agent who initiates diffusion as \emph{the attacker}. 
We use cybersecurity as a running example. Here the attacker \sydelete{(e.g., an attacker or an IT professional performing vulnerability assessment)} initiates the diffusion (e.g., the spread of malware) on a network of computers.
We define the impact of the diffusion as the number of infected nodes (e.g., compromised with malware). 
The attacker has two objectives: 1) she wishes to maximize the impact of the diffusion on a targeted set of nodes (e.g., computing nodes with access to critical assets), and 2) to limit the impact on non-targeted nodes to ensure stealth~\cite{Haghtalab17}.

\syreplace{Let $G=(\mathcal{V}, \mathcal{E})$ be a connected, undirected, possibly weighted, graph with no self-loops.}
{Let $G=(\mathcal{V}, \mathcal{E})$ be a connected, weighted or unweighted,  undirected graph with no self-loops.}
Let $n=|\mathcal{V}|$ be the number of nodes in $G$ and $\Adj$ be its adjacency matrix.
Throughout this paper, the eigenvalues of $\Adj$ are ranked in descending order $\lambda_1(\Adj) \ge \cdots \ge \lambda_n(\Adj)$.
Suppose the attacker targets a subgraph $\GS$ where $\mathcal{S} \subseteq \mathcal{V}$ is the node set of $\GS$.
Let $\mathcal{S'}= \mathcal{V} \setminus \mathcal{S}$, and its induced subgraph $\GSPrime$. Throughout the paper we assume $\GS$ is connected, and denote its
 adjacency matrix by $\AdjS$.
To achieve her objectives, the attacker modifies the structure of $G$.
The modified graph and targeted subgraph are represented by $\PerturG$ and $\PerturGS$, respectively.
Formally, the attacker's action is to add a perturbation 
$\bm{\Delta} \in \mathbb{R}^{n \times n}$ to $\Adj$, which results in the perturbed adjacency matrix $\TildeAdj = \Adj + \bm{\Delta}$.
The adjacency matrix of $\PerturGS$ is denoted by $\TildeAdjS$.

\subsection{Diffusion Dynamics} 
The status of a node is modeled by the well-known SIS (Susceptible-Infected-Susceptible) diffusion dynamics, where it alternates between ``infected'' and ``susceptible''.
\footnote{Due to brevity, a discussion on generalization of our approach to other diffusion dynamics is at \url{https://arxiv.org/abs/2008.05589}.} 
Due to the malware spread by infected neighbors, a susceptible node becomes infected with probability $\beta$.
An infected node becomes susceptible again (e.g., malware is removed) with probability $\delta$.
Following~\citeauthor{chakrabarti2008epidemic}\cite{chakrabarti2008epidemic}, this process is modeled by a nonlinear dynamical system.
Let $\pi_i$ be the probability of node $i$ becoming infected (e.g., compromised with malware) in the steady state of this dynamical system, with $\bm{\pi}$ the vector of these probabilities.
A key result in~\cite{chakrabarti2008epidemic} is that when $\lambda_1(\Adj) < \delta / \beta$ the system converges to the steady state $\bm{\pi}=\bm{0}$, which implies that the diffusion process quickly dies out.  However, when $\lambda_1(\Adj) \ge \delta / \beta$ the system converges to another steady state $\bm{\pi} \ne \bm{0}$.
We leverage this connection between graph structure, dynamical model of epidemic spread, and the epidemic threshold, in constructing our \emph{threat model}, as discussed next.

\subsection{Threat Model} 
\label{S:threatmodel}
\noindent{\bf Maximizing the Impact on $\GS$: }
To maximize the impact of diffusion on $\GS$, the attacker has two goals: 1) ensure that epidemics starting in $\GS$ spread rather than die out, and 2) ensure that epidemics starting outside $\GS$ are likely to reach it.
We capture the first goal by maximizing the largest (in modulus) eigenvalue of $\GS$, $\lambda_1(\TildeAdjS)$, which corresponds to the epidemic threshold of the targeted subgraph.\footnote{If $\GS$ is not connected, we may replace $\lambda_1(\TildeAdjS)$ by the largest eigenvalue of the largest connected component of $\GS$.} The second goal is captured 
by maximizing the \emph{normalized cut} of $\GS$, $\phi(\mathcal{S})$, where $\mathcal{S}$ is the set of nodes in $\GS$ and $\SPrime$ are the nodes in the remaining graph.
The normalized cut is formally defined as follows:
\vspace{-0.1in}
\begin{equation}\label{eq:norm-cut}
    \small
    \phi(\mathcal{S}) = \text{cut}(\mathcal{S}, \SPrime)\left( \frac{1}{\text{vol}(\mathcal{S})} + \frac{1}{\text{vol}(\SPrime)} \right),
\end{equation} 
where $\text{cut}(\mathcal{S}, \SPrime)$ is the sum of the weights on the edges across $\mathcal{S}$ and $\SPrime$ (unit weights for unweighted graphs), and $\text{vol}(\mathcal{S})$ (resp. $\text{vol}(\SPrime)$) is the sum of degrees of the nodes in $\mathcal{S}$ (resp. $\SPrime$).
The formal rationale for using the normalized cut is based on Meila et al.~\cite{meila2001learning}, which showed that increasing $\phi(\mathcal{S})$ increases the probability that a random walker transitions from $\SPrime$ to $\mathcal{S}$, if we assume that $\GS$ is smaller than $\GSPrime$.

\noindent{\bf Limiting the Impact on $\GSPrime$: }
%\label{sec:limit_impact}
Another important objective of the attacker is to limit the impact on $\GSPrime$, the non-targeted part of the graph.
%, defined as 
We capture this goal in two different ways. First, by limiting the likelihood of the epidemic spreading to $\GSPrime$, which we define as minimizing the impact $I(\GSPrime)= \sum_{i \in \SPrime}^{}{\pi_i}$. Second, by limiting the impact on the spectrum of $\Adj$.

We now demonstrate that minimizing $I(\GSPrime)$ is approximately equivalent to minimizing the eigenvector centrality of $\SPrime$. Let $\bm{P}^t$ be the global configuration of the graph at time step $t$, where $P^t_i$ is the probability that node $i$ is infected (e.g., compromised with malware).
Following~\citeauthor{van2008virus}~\cite{van2008virus} \leodelete{(see Section IV), for an arbitrary node $i$}, ignoring higher-order terms \leodelete{involving $P^{t}_i$} and taking the time step to be infinitesimally small,
the dynamics of $P^{t}_i$ is modeled as the following:
\vspace{-0.1in}
\begin{equation}\label{eq:trans_P_i}
    \small
    \frac{d P^{t}_i}{d t} = \sum_{j \in \mathcal{V}}^{}{\beta \tilde{A}_{ij} P^{t}_j} - \delta P^{t}_i.
\end{equation}

Here, we can think of the two terms on the right side as two competing forces. 
The first term is the force contributed by the infected neighbors of node $i$ (which increases $P^t_i$), while the second term is the force due to $i$'s self recovery (which decreases $P^t_i$).
Rewriting in matrix notation yields:
% \vspace{-0.1in}
\begin{equation}\label{eq:proxy_chain}
\small
\frac{d \bm{P}^t}{d t} = \big[ \beta \TildeAdj - \delta \bm{I} \big] \bm{P}^t,
\end{equation}

which gives a linear approximation to the non-linear dynamical system proposed in~\cite{chakrabarti2008epidemic}.
The steady state $\bm{\pi}$ must satisfy $\big[ \beta \TildeAdj - \delta \bm{I} \big]\bm{\pi} = \bm{0}$,  which is equivalent to $\TildeAdj \bm{\pi} = (\delta / \beta) \bm{\pi}$.
Suppose $\lambda_1(\TildeAdj)=\delta / \beta$,  and $\bm{\pi}$ is the corresponding eigenvector. 
Let $\TildeAdjEigvec$ be the unit eigenvector associated with $\lambda_1(\TildeAdj)$.
Let $\sigma(\mathcal{S})=\sum_{j \in \mathcal{S}}^{}{\tilde{v}_1[j]}$ be the eigenvector centrality of $\mathcal{S}$.
Noting that $\bm{\pi}$ may differ from $\TildeAdjEigvec$ by up to a multiplicative constant $c$,
the impact on $\GSPrime$ can be approximated as:
% \vspace{-0.1in}
\begin{equation}\label{eq:eig_cent}
\small
I(\GSPrime) = \sum_{j \in \SPrime}^{}{\pi_j} \approx c \sum_{j \in \SPrime}^{}{\tilde{v}_1[j]} = c\big(1 - \sigma(\mathcal{S}) \big),
\end{equation}

where the last equality is because $\mathcal{S}$ and $\SPrime$ are disjoint and $\TildeAdjEigvec$ is an unit vector.
Thus, minimizing the impact on $\GSPrime$ is approximately equivalent to maximizing the eigenvector centrality of $\mathcal{S}$.

Recall that to have an epidemic spread, one needs $\lambda_1(\TildeAdj) \ge \delta / \beta$. Here, we assumed $\lambda_1(\TildeAdj) = \delta / \beta$. In Section~\ref{sec:exp}, we demonstrate that our analysis yields an approach that is effective even when this assumption fails to hold (i.e., when $\lambda_1(\TildeAdj) > \delta / \beta$).

Now we focus on limiting the impact on the spectrum of $\Adj$.
\footnote{In  cybersecurity, there are natural interpretations of an attack's stealth. For further details, see the extended version at \url{https://arxiv.org/abs/2008.05589}.} 
Let $\epsilon > 0$ be the attacker's budget.
Formally, this notion is captured through the following constraints:
    \vspace{-0.2in}
    \begin{equation}\label{eq:budget}
    \small
    | \lambda_i (\tilde{\Adj} ) - \lambda_i ( \Adj )  |  \le \epsilon , \, i=1,\ldots,n.
    \end{equation}

In summary, the principal aims to (\emph{i}) maximize the impact on $\GS$ through maximizing $\lambda_1(\TildeAdjS)$ while (\emph{ii}) limiting the impact on $\GSPrime$ by maximizing the eigenvector centrality $\sigma(\mathcal{S})$, and satisfying Eq.~\eqref{eq:budget}. Formally, the principal aims to solve the following optimization problem:
\begin{equation}\label{eq:model}
\small
\begin{aligned}
& \max_{\tilde{\Adj}} & & \alpha_1 \lambda_1(\tilde{\Adj}_\mathcal{S}) + \alpha_2 \sigma(\mathcal{S})  + \alpha_3 \phi(\mathcal{S})  \\
&s.t.     &    & \tilde{\Adj} \in  \mathcal{P} = \Bigg\{ \tilde{\Adj}\, \Bigg| \, 
\begin{aligned}
%& \norm{\tilde{\Adj}  - \Adj }_2 \le \epsilon  \\
&   | \lambda_i (\tilde{\Adj} ) - \lambda_i ( \Adj )  |  \le \epsilon , \, i=1,\ldots,n,\\
& \TildeAdj = \TildeAdj^\top, \TildeAdj_{ii}=0,\,\forall i=1,\ldots,n
\end{aligned}
 \Bigg\},
\end{aligned}
\end{equation}
\vspace{-0.2in}

where the relative importance of the terms is balanced by the nonnegative constants $\alpha_1, \alpha_2$, $\alpha_3$, and the restrictions $\TildeAdj = \TildeAdj^\top$ and $\TildeAdj_{ii}=0,\,\forall i=1,\ldots,n$ ensure that $\TildeAdj$ is a valid adjacency matrix.

\section{\AlgoName: Proposed Algorithm}\label{sec:algo}
To solve the optimization problem in Eq.~\eqref{eq:model}, a natural approach would be to use a form of projected gradient ascent.
There are, however, two major hurdles to this basic approach: 1) the objective function involves terms that do not have an explicit functional representation in the decision variables, and 2) the projection step is quite expensive, as it involves projecting into a spectral norm ball, which entails an expensive SVD operation~\cite{lefkimmiatis2013hessian}.
We address these challenges in Algorithm~\ref{algo:grad_ascent}, which is our gradient-based solution to the attacker's optimization problem as described in Eq.~\eqref{eq:model}.

    \begin{algorithm}[h]
    \small
    \caption{\syreplace{Gradient Ascent Algorithm}{\AlgoName} }\label{algo:grad_ascent}
    \begin{algorithmic}[1]
    \State   \textbf{Input}: $\Adj, \epsilon, \{ \eta_i \}_{i=1}$ \Comment{$\{ \eta_i \}_{i=1}$ is a schedule of step sizes} 
    \State   \textbf{Initialize}: $i =1, \TildeAdj_1 = \Adj, B_1=0$ \Comment{$B_i$: the amount of budget used just before step $i$} 
            \While{True} 
                \State Set $\bm{\Delta}_i$ to the gradient of $\alpha_1 \lambda_1(\tilde{\Adj}_\mathcal{S}) + \alpha_2 \sigma(\mathcal{S}) + \alpha_3 \phi(\mathcal{S})$ w.r.t. to ${\TildeAdj_i}$ \label{grad-step} %\Comment{$\bm{\Delta}_i$ is symmetric}
                \State Set the diagonal entries of $\bm{\Delta}_i$ to zeros
                \If{$\norm{\bm{\Delta}_i} = \bm{0}$} \Comment{a local optimum is found}
                  \State return $\TildeAdj_i$ 
                \EndIf
                \If{$B_i + || \eta_i \bm{\Delta}_i ||_2 \le \epsilon$}\label{look-ahead} \Comment{one-step look ahead}
                  \State $\TildeAdj_{i+1} = \TildeAdj_{i} + \eta_i \bm{\Delta}_i$, $B_{i+1} = B_{i} + \norm{\eta_i \bm{\Delta}_i}_2$\label{used_budget}, $i = i + 1$
         \Else
            \State return $\TildeAdj_i$
                \EndIf
            \EndWhile
    \end{algorithmic}
    \end{algorithm}

\leoreplace{The first}{A} key step of Algorithm~\ref{algo:grad_ascent} is line~\ref{grad-step}, where we compute the gradient of the attacker's utility function with respect to $\TildeAdj$.
This gradient involves terms that do not have an explicit functional form in terms of the decision variable, and we deal with each of these in turn.

First, consider the gradient of the normalized cut $\phi(\mathcal{S})$ w.r.t. $\TildeAdj$.
Let $\bm{x}_\mathcal{S}$ be the characteristic vector of $\mathcal{S}$, that is $x_{\mathcal{S}}[i]=1$ iff $i \in \mathcal{S}$.
Let $\tilde{\bm{D}}$ be the diagonal degree matrix $\tilde{\bm{D}}_{ii} = \sum_j \TildeAdj_{ij}$, and let $\TildeL=\TildeAdj - \TildeD$ be the Laplacian matrix.
Using $\tilde{\bm{D}}$ and $\TildeL$ to express $\text{vol}(\mathcal{S})$ and 
$\text{cut}(\mathcal{S}, \SPrime)$, respectively, we have:

\vspace{-0.2in}
\begin{small}
\begin{equation}\label{eq:norm_cut_1}
    \small
    \phi(\mathcal{S}) = \bm{x}^\top_\mathcal{S} \TildeL \bm{x}_\mathcal{S} \Bigg( \frac{1}{\bm{x}^\top_\mathcal{S} \TildeD \bm{x}_\mathcal{S}} + \frac{1}{\bm{x}^\top_\SPrime \TildeD \bm{x}_\SPrime} \Bigg).
\end{equation}
\end{small}

\leoreplace{It is clear that}{Clearly,} Eq.~\eqref{eq:norm_cut_1} is a differentiable function of $\TildeAdj$.
Computing its gradient $\nabla_{\TildeAdj} \phi(\mathcal{S})$ can then be handled by automatic differentiation tools such as PyTorch~\cite{paszke2017automatic}.

Next, we compute the gradient of $\lambda_1(\TildeAdjS)$  w.r.t. $\TildeAdj$.
A standard way to compute $\lambda_1(\TildeAdjS)$ is by using SVD.
However, this is both prohibitively expensive ($O(n^3)$), and does not provide us with the necessary gradient information. 
Instead, we use the power method~\cite{golub1996matrix} to compute $\lambda_1(\TildeAdjS)$. 
Let $\bm{v}_\mathcal{S}$ be the eigenvector associated with the largest eigenvalue $\lambda_1(\TildeAdjS)$.
Using Rayleigh quotients~\cite{trefethen1997numerical}, we can compute $\lambda_1(\TildeAdjS)$ as follows:
% \vspace{-0.1in}
\begin{small}
\begin{subequations}
    \begin{align}
        & \bm{v}_\mathcal{S}  = \argmax_{\norm{\bm{x}}_2=1} \bm{x}^\top \TildeAdjS \bm{x} \label{eq:rayleigh}\\
        & \lambda_1(\TildeAdjS)  = \bm{v}_\mathcal{S}^\top \TildeAdjS \bm{v}_\mathcal{S}.
    \end{align}
\end{subequations}
\end{small}

Thus, when $\bm{v}_\mathcal{S}$ is known, the computation of $\lambda_1(\TildeAdjS)$ reduces to matrix multiplications.
In addition, $\TildeAdjS$ is usually sparse, so we can  leverage sparse matrix multiplication to speed up the computation.

The remaining challenge is that $\bm{v}_\mathcal{S}$ is an optimal solution of an optimization problem, and we need an explicit derivative of it.
%$\bm{v}_\mathcal{S}$ w.r.t. $\TildeAdjS$ due to the argmax operator in Eq.~\eqref{eq:rayleigh}.
Fortunately, our problem has a special structure that we exploit to obtain an approximation of the derivative of $\bm{v}_\mathcal{S}$. 
From our experiments we find that $\PerturGS$ is nearly always  connected.
%\footnote{$\GS$ is disconnected in some rare cases, e.g., $\alpha_3$ dominates $\alpha_1$ and $\alpha_2$.}.
This means that the largest eigenvalue of $\TildeAdjS$ is simple.
In addition, due to the Perron–Frobenius theorem, the absolute value of the largest eigenvalue is strictly greater than the absolute values of others, i.e., $|\lambda_1(\TildeAdjS)| > |\lambda_k(\TildeAdjS)|$ for all $k\ne 1$.
Under these conditions, we can use the power method to estimate $\bm{v}_\mathcal{S}$ by repeating the formula: $\tilde{\bm{v}}_\mathcal{S}^{(t+1)} = \TildeAdjS \tilde{\bm{v}}_\mathcal{S}^{(t)} / \norm{\TildeAdjS \tilde{\bm{v}}_\mathcal{S}^{(t)}}_2$.
The $\ell_2$-norm distance between $\tilde{\bm{v}}_\mathcal{S}^{k}$ and $\bm{v}_\mathcal{S}$ decreases in a rate $O(\rho^k)$~\cite{golub1996matrix}, where $\rho < 1$.
In our experiments we found $k=50$ is enough to give a high-quality estimation for a graph with $986$ nodes.
Intuitively, we are using a sequence of differentiable operations to approximate the argmax operation.
Therefore the computation of $\nabla_{\TildeAdj} \lambda_1(\TildeAdjS)$ can be handled by PyTorch. \sydelete{ (see Figure~\ref{fig:lambda1_comp}). }

We use the same machinery  to compute $\nabla_{\TildeAdj} \sigma(\mathcal{S})$.
First, we write  $\sigma(\mathcal{S})$ in matrix notation:
\vspace{-0.1in}
\begin{equation}
    \small
    \sigma(\mathcal{S}) = \bm{v}^\top \bm{x}_\mathcal{S},
\end{equation}
where $\bm{v}$ is the unit eigenvector associated with $\lambda_1(\TildeAdj)$.
Then we apply the power method to compute $\bm{v}$.
Finally, $\sigma(\mathcal{S})$ is just a linear function of $\bm{v}$.
All of these operations %from $\TildeAdj$ to $\sigma(\mathcal{S})$
are differentiable, and the computation of $\nabla_{\TildeAdj} \sigma(\mathcal{S})$ is handled by PyTorch. 
%This ends our discussion of line~\ref{grad-step}.

We next address the challenge imposed by the constraints \eqref{eq:budget}, which can result in a computationally challenging projection step which can also significantly harm solution quality.
%First, note that even checking whether the principal is within budget would require computing every eigenvalue of $\Adj$ and $\TildeAdj$, which takes $O(n^3)$, by using singular value decomposition (SVD).
%Worse, in an iterative algorithm such as gradient descent, the SVD would be run every iteration.
%Finally, we need to have a way to map solutions at the end of each iteration to this feasible space, for example, through a projection step.
%However, projecting into a spectral norm ball also involves an expensive SVD operation~\cite{lefkimmiatis2013hessian}.
We address this challenge as follows.
Given a real symmetric matrix $\bm{X}$, let $\norm{\bm{X}}_2$ denote its spectral norm.
To satisfy Eq.~\eqref{eq:budget}, we use the following result from pseudospectrum theory (see \cite{trefethen2005spectra}, Theorem~2.2):
\vspace{-0.1in}
\begin{equation}\label{eq:budget-pseudo}
    \small
    \big| \lambda_i(\tilde{\Adj}) - \lambda_i(\Adj)  \big|  \le \epsilon , i=1,\ldots,n \iff \norm{\tilde{\Adj}  - \Adj }_2 \le \epsilon %\iff \norm{ \Perturb }_2 \le \epsilon.
\end{equation}

Since $\Perturb = \tilde{\Adj} - \Adj$ is real and symmetric, we have $\norm{\Perturb}_2 = \max \SET{|\lambda_1(\Perturb)|, |\lambda_n(\Perturb)|}$ \leoreplace{.
In addition, the eigen-decomposition of $\Perturb$ indicates that}{and} $-\lambda_n(\Perturb) = \lambda_1(-\Perturb)$, which leads to:

\vspace{-0.1in}
\begin{small}
\begin{equation}
\text{$\TildeAdj$ satisfies Eq.~\eqref{eq:budget}} \iff  \max\{ |\lambda_1(\bm{\Delta})|, |\lambda_1(-\bm{\Delta})| \} \le \epsilon.
\end{equation}
\end{small}

This equivalence allows the attacker to check whether she is within budget simply by evaluating $\max\{ |\lambda_1(\bm{\Delta})|, |\lambda_1(-\bm{\Delta})| \}$, i.e., computing the largest eigenvalue of a real symmetric matrix, which can be computed efficiently using, e.g., the power method~\cite{golub1996matrix}.

Our algorithm leverages this connection as follows.
Line~\ref{look-ahead} in Algorithm~\ref{algo:grad_ascent} is a one step look-ahead, which ensures that the perturbation $\bm{\Delta}_i$ is only added to $\TildeAdj_i$ when there is enough budget.
Recall from Section~\ref{S:threatmodel} that $\norm{\bm{\Delta}_i}_2 = \max\{ |\lambda_1(\bm{\Delta}_i)|, |\lambda_1(-\bm{\Delta}_i)| \}$.
Thus this step requires us to compute $\lambda_1(\bm{\Delta}_i)$ and $\lambda_1(-\bm{\Delta}_i)$, using again the power method. 
Line~\ref{used_budget} tracks the amount of budget used so far. 
We now show that the output of Algorithm~\ref{algo:grad_ascent} always returns a feasible solution.
%solves problem~\eqref{eq:model}.
Suppose Algorithm~\ref{algo:grad_ascent} terminates after $k > 1$ iterations.
This means $B_k + \norm{\eta_k \Perturb_k}_2 > \epsilon$ and $B_k \le \epsilon$.
In other words $B_k = \sum_{i=1}^{k-1}{\norm{\eta_i \Perturb_i}_2} \le \epsilon$.
Note that the total amount of perturbation added to $\Adj$ is $\Perturb = \sum_{i=1}^{k-1}{\eta_i \Perturb_i}$.
The triangle inequality implies $\norm{\Perturb}_2 \le \epsilon$.
%Consequently, the combination of lines~\ref{look-ahead} and \ref{used_budget}  ensure that the optimizer $\TildeAdj_i$ always lies in the feasible region of the principal's optimization problem.
%Thus, we do not need to worry about the feasibility issue. 
%An alternative method to solve our model is projection gradient descent, however, it requires repeatedly projecting onto the feasible region. 
%Observe that the feasible region is a spectral norm ball.
%Projecting the optimizer $\TildeAdj_i$ onto a spectral norm ball amounts to solving an SVD problem~\cite{lefkimmiatis2013hessian}, which is too expensive.
%In addition, the projection step does not guarantee the projected matrix is symmetric. 

For each iteration of Algorithm~\ref{algo:grad_ascent}, 
\syreplace{the most expensive part comes from the power method and matrix multiplications.}{the most computationally expensive components are the power method and matrix multiplication.}
Let $m$ be the number of nonzeros in $\TildeAdj_i$; if the graph is unweighted then $m$ is the number of edges at this iteration. 
By leveraging the sparseness exhibited in $\TildeAdj_i$, the power method runs in $O(m)$ and the matrix multiplications cost $O(mn)$.
Thus, the time complexity of each iteration is $O(mn)$, \syedit{which significantly improves the $O(n^3)$ time complexity of SVD that would otherwise be needed.}

Recall that our model for targeted diffusion is applicable to both weighted and unweighted graphs.
For weighted graphs, the attacker modifies the weights on existing edges.
%(i.e., increasing the number of phishing emails along certain links). 
For unweighted graphs, the attacker adds new edges 
or deletes existing edges from the graph.
The main difference between the two settings is that the latter needs a rounding heuristic to convert a matrix with fractional entries to a binary adjacency matrix. 
We discuss this heuristic below.

After running Algorithm~\ref{algo:grad_ascent}, we obtain a perturbed  matrix $\TildeAdj$ with fractional entries. 
For unweighted graphs, a rounding heuristic is needed to convert $\TildeAdj$ to a valid adjacency matrix.
Let $\mathcal{D}=\{ (i, j) | \Tilde{A}_{i,j} \ne A_{i,j} \}$ be the set of candidate edges that will be added or deleted from $G$.
For each edge $(i, j) \in \mathcal{D}$ define the score $s_{(i, j)} = |\Tilde{A}_{i,j} - A_{i,j}|$.
Intuitively, $s_{(i, j)}$ indicates the impact that adding or deleting the edge has on the principal's utility.
Next, we iteratively modify $G$, by adding or deleting edges in $\mathcal{D}$, starting with the one with the largest $s_{(i, j)}$.
The modification process stops when the budget is exhausted, which results in the desired binary adjacency matrix. 
For weighted graphs, %, in most applications the weights are integers.
let $C=\max_{i, j} A_{ij}$ and normalize each entry by $C$, that is $A_{ij} / C$.
We run Algorithm~\ref{algo:grad_ascent} on the normalized adjacency matrix, which results in $\TildeAdj$.
The desired adjacency matrix is obtained by multiplying each $\tilde{A}_{ij}$ by $C$, $C \tilde{A}_{ij}$.
If integer weights are desired (\syedit{e.g., the number of packages transmitted between two computers}), a final rounding step is applied.
Our experimental results show that 
the rounding heuristic is effective in practice.

\section{Certified Robustness}\label{sec:min_budget}
%Having developed a framework for network structure optimization for targeted diffusion, we now ask the following question: 

This section addresses the following question: what are the limits on the attacker's ability to successfully accomplish her attack?
More precisely, we now seek to identify necessary conditions on the attack budget $\epsilon$ so the attack  succeeds; conversely, we can view a given graph to be \emph{certified} to be robust to attacks that use a smaller budget than the one required.

Let \texttt{TargetDiff}($\mathcal{S}, G, \epsilon$) be an instance of the targeted diffusion problem with target subset $\mathcal{S}$, underlying graph $G$ and budget $\epsilon$. \leoreplace{An instance \texttt{TargetDiff}($\mathcal{S}, G, \epsilon$) can be successfully attacked if the attacker is able to modify $G$ into $\PerturG$ within budget $\epsilon$ such that $I(\PerturGS) > I(\GS)$.}{The attacker is successful on an instance \texttt{TargetDiff}($\mathcal{S}, G, \epsilon$) if she is able to modify $G$ into $\PerturG$ within budget $\epsilon$ such that $I(\PerturGS) > I(\GS)$.} We now derive a necessary condition for successful \leodelete{targeted} attacks, in the form of a lower bound on $\epsilon$.

To derive the necessary condition on $\epsilon$, we use our experimental observation that in successful attacks the degrees of nodes in the targeted subgraph $\GS$ always increase.
This is intuitive: a denser subgraph $\GS$ will tend to increase the propensity of the diffusion (e.g., of malware) to spread within it, which is one of our explicit objectives.
Let $d_i$ (resp. $\tilde{d}_i$) be the degree of node $i$ before (resp. after) graph modification.
We assume if an attack is successful, the degrees of nodes in $\GS$ are increased, i.e., $\tilde{d}_i \ge d_i$ for $i \in \mathcal{S}$.

Now, observe that computing the exact value of $I(\GS)$ is intractable, since the exact computation of $\pi_i$ is prohibitive (see, e.g., \cite{van2008virus}, Section IV.B).
\citet{van2008virus} proposed a simple yet effective estimator for $\pi_i$ to be $1 - \delta / (\beta d_i)$.
The estimator works in the regime $\delta / \beta \le d_{min}$, where $d_{min}$ is the minimum degree of $G$.
Consequently, an estimator for $I(\GS)$ is $\hat{I}(\GS) = \sum_{i \in \mathcal{S}}^{}{1 - \delta / (\beta d_i)}$.
We focus on the setting where the estimation error is bounded by a small number, i.e., $| \hat{I}(\GS) - I(\GS) | \le \tau$.  
Note that $\tau$ can be estimated from historical diffusion data.
The formal statement of the necessary condition is in Theorem~\ref{th:cert_robust}. 
\footnote{Due to brevity the proof is in Appendix~\ref{app:cert_robust} at \url{https://arxiv.org/abs/2008.05589}.}

\begin{theorem}\label{th:cert_robust}
% \small
Given an instance \textup{\texttt{TargetDiff}}($\mathcal{S}, G, \epsilon$), $I(\GS)$ is estimated by $\hat{I}(\GS) = \sum_{i \in \mathcal{S}}^{}{1 - \delta / (\beta d_i)}$.
Suppose we have an upper bound $| \hat{I}(\GS) - I(\GS) | \le \tau$, 
the degrees of nodes in $\mathcal{S}$ are increased, i.e., $\tilde{d}_i \ge d_i$ for $i \in \mathcal{S}$, and $\delta / \beta \le d_{min}$. 
In order to have $I(\PerturGS) - I(\GS) > 2\tau$, the budget $\epsilon$ must satisfy:
    % \vspace{-0.1in}
    \begin{small}
    \begin{equation}\label{eq:cert_robust}
         \epsilon \ge \sqrt{\frac{|\mathcal{S}|}{n}} \left( \frac{\sum_{i \in \mathcal{S}}^{}{d^2_i}}{|S|} - \frac{\left(\sum_{i \in \mathcal{S}}^{}{d_i} \right)^2}{|S|^2} \right)^{1/2}.
    \end{equation}
    \end{small}
\end{theorem}
% \vspace{-0.1in}

The quantity inside the square root is always nonnegative due to Jensen's inequality.
The lower bound involves only structural properties of the graph (node degrees and the size of $\mathcal{S}$) and thus can be easily computed given an arbitrary graph.
As mentioned above, we can view this lower bound as a robustness certificate, or guarantee for the given graph.
It guarantees, in particular, that when the budget is below the lower bound, the 
total probability of ``infection'' (e.g., malware infection) in $\GS$ cannot be increased by more than $2\tau$.
In the special case of perfect estimation ($\tau = 0$), it implies impossibility of increasing the susceptibility of $\GS$ to targeted diffusion.

The proof of Theorem~\ref{th:cert_robust} does not depend on the specific objective function proposed in this paper. 
Consequently, the certificate is not specific to our particular objective function.
Further, the lower bound is independent of the values of $\delta$ and $\beta$, as long as $\delta / \beta \le d_{min}$.
\leodelete{In other words, it applies to both highly infectious (small $\delta / \beta$) and slow-spreading (large $\delta / \beta$) diffusion.}

We briefly discuss the settings  where the robustness guarantee is most applicable.
First, the estimation for the infected ratio on $\GS$ is accurate, i.e.,  $| \hat{I}(\GS) - I(\GS) | \le \tau$ and $\tau$ is small.
According to~\citet{van2008virus}, this usually happens on graphs with small degree variation. 
Another setting is where the degrees of nodes in $\GS$ increase as a result of the attack, which is both natural and empirically founded, as we mentioned earlier.
We provide experimental results on synthetic networks to verify the robustness guarantee in Section~\ref{sec:exp}.

\section{Experiments and Discussion}\label{sec:exp}
This section presents experimental results on three  real-world datasets: an email network, an airport network, and a brain network. 
\footnote{
Due to brevity, additional results on real and synthetic networks are at \url{https://arxiv.org/abs/2008.05589}.
}

\syreplace{
For each network we run four experiments which differ in the hyper-parameters $(\alpha_1, \alpha_2, \alpha_3)$, corresponding to the four columns of Figure~\ref{fig:all}.
The first uses $\alpha_1=\alpha_2=\alpha_3=1/3$, which encodes that the attacker's objectives are equally important.
This is to show the overall effectiveness of our approach.
The other three experiments are designed to show the effectiveness of each term in the principal's objective function.
For example, in the second experiment, we use hyper-parameters $\alpha_1=1/3, \alpha_2=0$, and $\alpha_3=1/3$ (the hyper-parameters do not need to sum to one).
}{For each network we run \ModelName with hyper-parameters $\alpha_1=\alpha_2=\alpha_3=1/3$, which encodes that the attacker's objectives are equally important.
}
To study how the attacker's effectiveness changes with respect to her budget, we set $\epsilon = \gamma \lambda_1(\Adj)$ and vary $\gamma$ from $10\%$ to $50\%$.
A single initially infected node is selected uniformly at random.

Recall we use $G$ and $\PerturG$ to denote the original and the modified graphs, respectively.
We simulate the spreading dynamics 2000 times on both $G$ and $\PerturG$.
For unweighted graphs the recovery rate $\delta$ and transmission rate $\beta$ are set to $0.24$, $0.06$, resp.; for weighted graphs we set $\delta=0.24$ and $\beta=0.2$.
The spreading dynamics converges exponentially fast to the steady state:
empirically, we found 30 time steps to be enough to reach the steady state in most cases.
\sydelete{\footnote{\leocomment{if we need space, I think we can get rid of this footnote} The SIS model is similar to an irreducible ($G$ and $\PerturG$ are connected) and aperiodic (each infected node recovers with probability $\delta$) Markov chain representing a lazy random walk on the network.}}
When the simulation finishes, we extract the number of nodes that are ``infected''.
We use $I_{\text{original}}$  and $I_{\text{modified}}$  to represent the fractions of infected nodes on $G$ and $\PerturG$, resp.

We use two other algorithms as baselines for comparison, which we call \texttt{deg} and \texttt{gel}. The two algorithms work by alternating between modifying $\GS$ and modifying $\GSPrime$ until the budget is spent. When modifying $\GS$, \texttt{deg} chooses the edge $(i,j)$ with the maximum value of $d_i + d_j$. Algorithm \texttt{gel} is based on \cite{tong2012gel}, and chooses the edge $(i,j)$ with maximum \emph{eigenscore}, defined as $u[i] v[j]$, where $u,v$ are the left and right principal eigenvectors of $\Adj$, respectively. This edge is chosen from among those edges that are absent (if the graph is unweighted) or present (if the graph is weighted). When modifying $\GSPrime$, these baselines choose an existing edge of to remove (if unweighted) or decrease its weight (if weighted) with the highest value of $d_i + d_j$ or $u[i] v[j]$, respectively.

\noindent{\bf Unweighted Graphs: } 
We consider (the largest connected component of) an email network~\cite{leskovec2007graph} that has 986 nodes.
An edge $(i, j)$ indicates that there were email exchanges between nodes $i$ and $j$.
This data set contains ground-truth labels to indicate which community a node belongs to.
We \sydelete{uniformly at random} pick a community with 15 nodes as $\mathcal{S}$; \syedit{the results for communities with other sizes are similar.}

\sydelete{The experimental results for the email network are at the top row of Figure~\ref{fig:comparison}.}
The overall effectiveness of our approach is shown in Figure~\ref{fig:comparison}, top.
The difference $I_{\text{modified}} - I_{\text{original}}$ of the impact on the modified and original graphs is shown for $\GS$ (red line) and $\GSPrime$ (purple line), respectively.
As $\gamma$ gets larger, the impact on $\GS$ increases, while the impact on $\GSPrime$ is under control, which demonstrates that the proposed approach is highly effective at both increasing the impact of diffusion on the targeted subgraph, and at the same time preventing the impact on the remaining graph.
% \syedit{Additional experimental results to examine the effectiveness of each term in the objective function of \ModelName are provided in Appendix~\ref{sec:real_exp}. }

%\begin{figure}[t]
%\def\FigSize{1.12in}
%\small
%\centering
%\setlength{\tabcolsep}{0.1pt}
%\begin{tabular}{ccc}
%\includegraphics[width=\FigSize]{figure/Email_numExp_1_unweighted_overall.pdf} &
%\includegraphics[width=\FigSize]{figure/Airport_numExp_1_weighted_overall.pdf} &
%\includegraphics[width=\FigSize]{figure/Brain_numExp_1_weighted_overall.pdf}
%\end{tabular}
%\includegraphics[width=\columnwidth]{figure/fig_effectiveness.pdf}
%\caption{
%\leoreplace{Experiments showing the model's effectiveness}{\ModelName effectively achieves targeted diffusion in $G_S$}.
%\textbf{Left}: the email network; \textbf{Middle}: the airport network; \textbf{Right}: the brain network.
%}
%\label{fig:all}
%\end{figure}

\noindent{\bf Weighted Graphs:}  We consider an airport network and a brain network.
The airport network~\cite{airport2010data} was collected from the website of Bureau of Transportation Statistics of the U.S., where the nodes represent all of the 1572 airports in the U.S. and the weights on edges encode the number of passengers traveled between two airports in 2010. 
We scaled the weights on the airport network to $[0, 1]$.
The targeted set $\mathcal{S}$ was chosen by first sampling a node $i$ uniformly at random, and then setting $\mathcal{S}$ to be $i$ and all its neighbors.
We report experimental results for an $\mathcal{S}$ with 60 nodes.
The brain network~\cite{crossley2013cognitive} consists of 638 nodes where each node corresponds to a region in human brain.
An edge between nodes $i$ and $j$ indicates that the two regions have co-activated on some tasks.
The weight on the edge quantifies the strength of the co-activation estimated by the Jaccard index.
The weights on edges lie in $[0, 1]$.
The 638 regions are categorized into four areas: default mode, visual, fronto-parietal, and central.
Each area is responsible for some functionality of human.
We select 100 nodes from the central area as the targeted set $\mathcal{S}$.
The results for the airport (resp. brain) network are at the center (resp. right) column of Figure~\ref{fig:comparison}.
The overall trend is similar to that of the email network. 
% \syedit{Additional results to corroborate the effectiveness of each term in the objective function of \ModelName are provided in Appendix~\ref{sec:real_exp}. }

\noindent{\bf Comparison against Baselines:} The comparisons against the baselines are shown in Figure~\ref{fig:comparison}, middle and bottom rows.
The moddle row shows the infectious ratios within the targeted subgraphs.
It is clear that our algorithm is more effective at increasing the infectious ratios than the baselines. 
The bottom row shows the infectious ratios within the non-targeted subgraphs.
The magnitudes of the differences are negligible, although in some cases our algorithm is significantly better than the baselines (e.g., on airport network when $\gamma=0.4$).

\noindent{\bf Verify the Certified Robustness:} We run experiments on synthetic networks to verify the certified robustness; the synthetic networks include Barab\'asi-Albert (BA)~\cite{barabasi1999emergence}, Watts-Strogatz~\cite{watts1998collective}, and Block Two-level Erd\H{o}s-R\'enyi (BTER) networks~\cite{seshadhri2012community}.
We use the same experimental setup as described above. 
Fig.~\ref{fig:robustness} shows the difference of infectious ratios on the modified and original graphs (within targeted subgraphs), as a function of the attacker's budget $\epsilon$.
The vertical dashed lines are the lower bounds on the budget computed using Eq.~\eqref{eq:cert_robust}.
Note that when the budget is less than the lower bound, the differences are close to zero, which means that the network is robust against targeted diffusion.

\begin{figure}[t]
\def\FigSize{1.1in}
\centering
\setlength{\tabcolsep}{0.1pt}
%\begin{tabular}{ccc}
%\includegraphics[width=\FigSize]{figure/BA_numExp_30_unweighted_robustness.pdf} & \includegraphics[width=\FigSize]{figure/Small-World_numExp_30_unweighted_robustness.pdf} & \includegraphics[width=\FigSize]{figure/BTER_numExp_30_unweighted_robustness.pdf}
%\end{tabular}
\includegraphics[width=0.75\columnwidth]{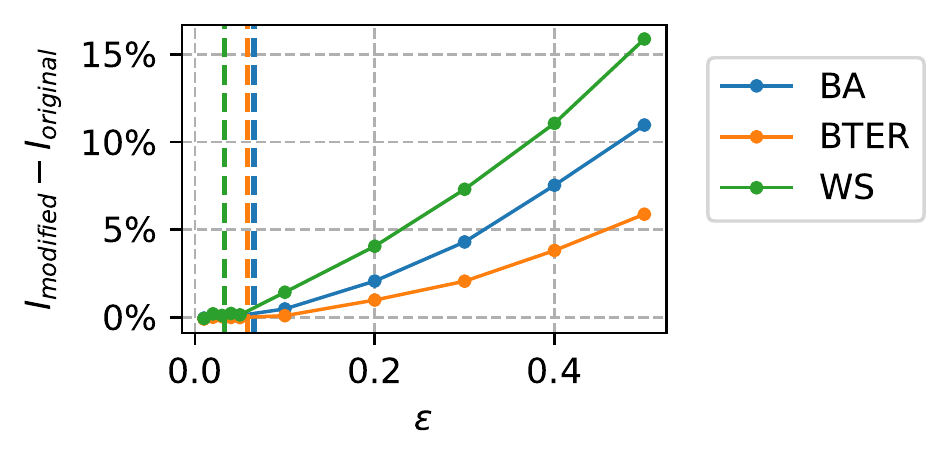} 
\caption{Certified robustness results. Dashed lines mark the lower bounds from Eq.~\eqref{eq:cert_robust}.
Solid lines represent infectious ratios within targeted subgraphs.
}
\label{fig:robustness}
\end{figure}

\begin{figure}[t]
\def\FigSize{1.1in}
\centering
\setlength{\tabcolsep}{0.1pt}
%\begin{tabular}{ccc}
%\includegraphics[width=\FigSize]{figure/Email_targeted.pdf} & \includegraphics[width=\FigSize]{figure/Airport_targeted.pdf} & \includegraphics[width=\FigSize]{figure/Brain_targeted.pdf} \\
%\includegraphics[width=\FigSize]{figure/Email_untargeted.pdf} & \includegraphics[width=\FigSize]{figure/Airport_untargeted.pdf} & \includegraphics[width=\FigSize]{figure/Brain_untargeted.pdf} 
%\end{tabular}
\includegraphics[width=\columnwidth]{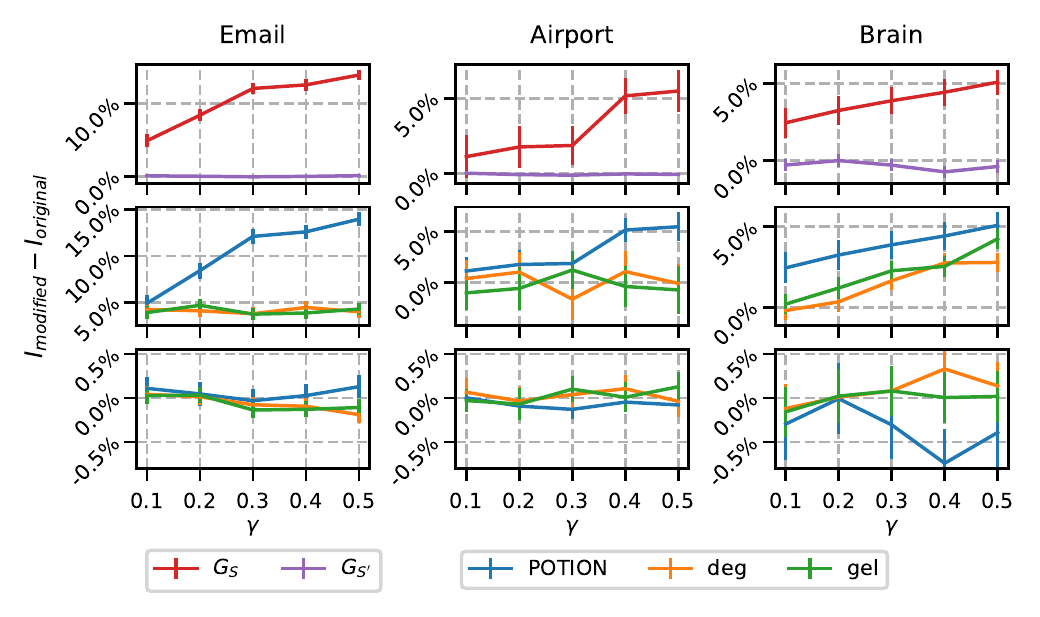} 
\caption{
\small
\ModelName effectively achieves targeted diffusion (\textbf{top}) in $G_S$ (red line) without affecting $G_{S'}$ (purple); higher is better.
Comparison against deg and gel baselines in $G_S$ (\textbf{middle}; higher is better) and $G_{S'}$ (\textbf{bottom}; lower is better).
}
\label{fig:comparison}
\end{figure}

\noindent{\bf Running Time:}
The running time of Algorithm~\ref{algo:grad_ascent} on the three real-world networks is showed in Figure~\ref{fig:runningTime}.
Each point in the figure is the average running time over 10 trials. 
Intuitively, as the budget $\gamma$ increases the attacker needs to search a larger space, therefore the running time increases.
The numbers of nodes and edges of the three networks are in Table~\ref{tab:real-stat}.

\begin{figure}[ht]
\def\FigSize{1.8in}
\centering
\small
\setlength{\tabcolsep}{0.1pt}
\begin{tabular}{c}
\includegraphics[width=\FigSize]{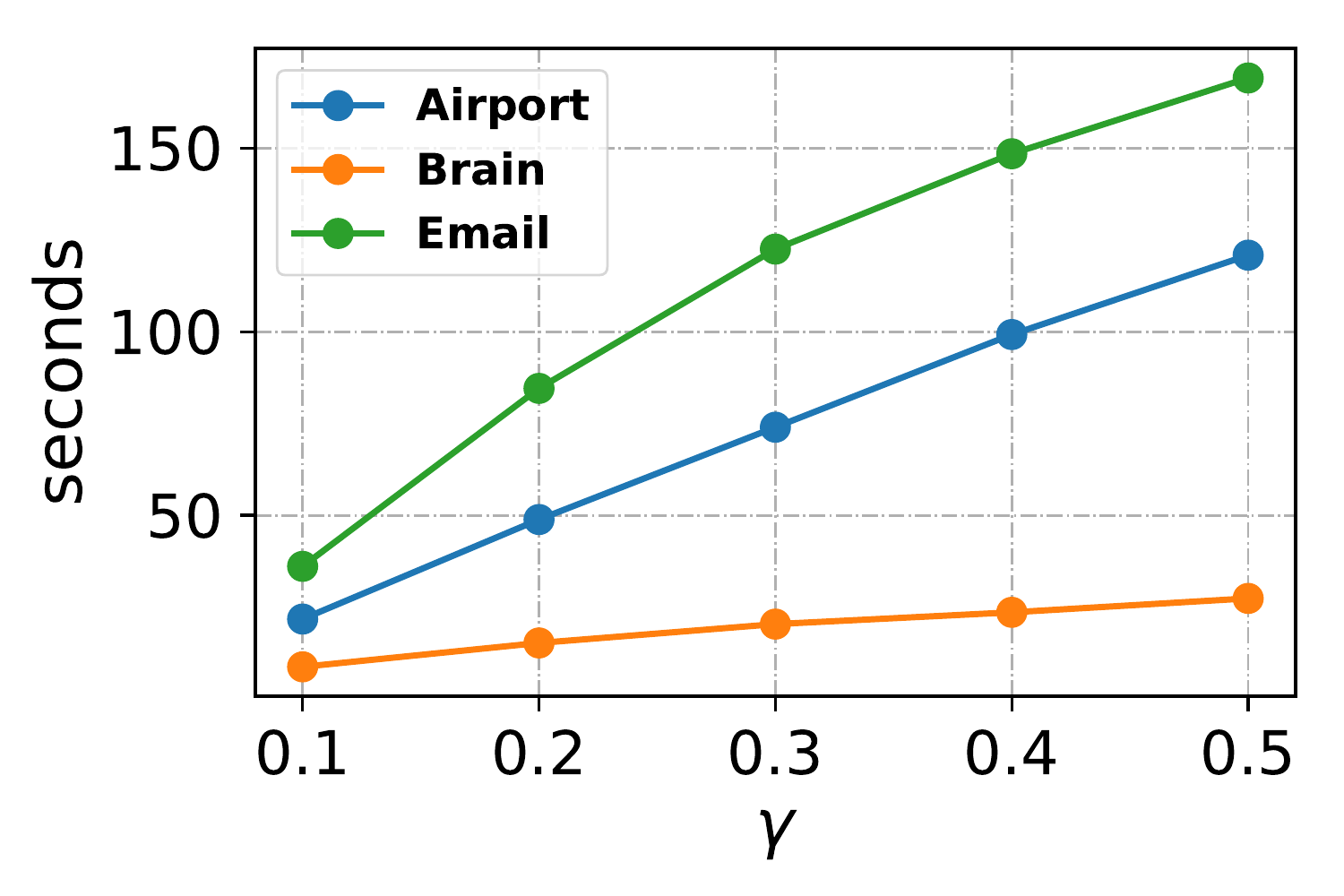}
\end{tabular}
\caption{Running time on the real-world networks.}
\label{fig:runningTime}
\end{figure}

\begin{small}
\begin{table}[h]
\centering
\begin{tabular}{@{}cccc@{}}
\toprule
        & Email & Airport & Brain \\ \midrule
\#nodes & 986   & 1572    & 638   \\
\#edges & 16064 & 17214   & 18625 \\ \bottomrule
\end{tabular}
\caption{Statistics of the real-world networks.}
\label{tab:real-stat}
\end{table}
\end{small}

\section{Conclusion}

Diffusion control on network has attracted much attention, however, most studies focus on diffusion over the entire network. 
We address the problem of \emph{targeted} diffusion attack on networks.
We present a combination of modeling and algorithmic advances to systematically address this problem.
On the modeling side, we present a novel model called \ModelName  that optimizes graph structure to affect such targeted diffusion attacks, which preserves structural properties of the graph. 
On the algorithmic side, we design an efficient algorithm named \AlgoName by leveraging Rayleigh quotients and pseudospectrum theory, which is scalable to real-world graphs. 
We also derive a condition to certify whether a network is robust against a broad class of targeted diffusion. 
Our experiments on both synthetic and real-world networks show that the model is highly effective in implementing the targeted diffusion attack.

% In this paper we propose the model POTION for targeted diffusion attack by optimizing graph structures.
% The attacker's targets---maximizing the impact on a targeted subgraph while limiting the impact on the rest of the graph---are encoded as the objective function of the model.
% An algorithm called POTION-ALG is proposed to solve the model, which leverages Rayleigh quotients and pseudospectrum theory.
% We also present a condition for certifying that a targeted subgraph is robust against such attack. 

% The effectiveness and efficiency of POTION and POTION-ALG are corroborated by experiments on both synthetic and real-world networks.
% In addition, we demonstrate the advantages of POTION by comparing against two baselines. 
% A running time analysis for POTION-ALG is provided

\section*{Acknowledgement}
SY and  YV were partially supported by the National Science Foundation (grants IIS-1903207 and IIS-1910392) and Army Research Office (grants W911NF1810208 and W911NF1910241). LT and TER were supported in part by the National Science Foundation (IIS-1741197) and by the Combat Capabilities Development Command Army Research Laboratory (under Cooperative Agreement Number W911NF-13-2-0045).
The authors would like to thank the anonymous reviewers and Chloe Wohlgemuth for their helpful comments.

% \section*{Broader Impact}
% \input{section/impact}

\bibliographystyle{plainnat}
\bibliography{main_abbrev}

\begin{thebibliography}{49}
\providecommand{\natexlab}[1]{#1}
\providecommand{\url}[1]{\texttt{#1}}
\expandafter\ifx\csname urlstyle\endcsname\relax
  \providecommand{\doi}[1]{doi: #1}\else
  \providecommand{\doi}{doi: \begingroup \urlstyle{rm}\Url}\fi

\bibitem[Amelkin and Singh(2019)]{amelkin2019fighting}
Victor Amelkin and Ambuj~K. Singh.
\newblock Fighting opinion control in social networks via link recommendation.
\newblock In \emph{{KDD}}, pages 677--685. {ACM}, 2019.

\bibitem[Avrachenkov and Litvak(2006)]{avrachenkov2006effect}
Konstantin Avrachenkov and Nelly Litvak.
\newblock The effect of new links on google pagerank.
\newblock \emph{Stochastic Models}, 22\penalty0 (2):\penalty0 319--331, 2006.

\bibitem[Backstrom and Leskovec(2011)]{backstrom2011supervised}
Lars Backstrom and Jure Leskovec.
\newblock Supervised random walks: predicting and recommending links in social
  networks.
\newblock In \emph{{KDD}}, pages 635--644, 2011.

\bibitem[Bailey et~al.(1975)]{bailey1975mathematical}
Norman~TJ Bailey et~al.
\newblock \emph{The mathematical theory of infectious diseases and its
  applications}.
\newblock Charles Griffin \& Company Ltd, 1975.

\bibitem[Barab{\'a}si and Albert(1999)]{barabasi1999emergence}
Albert-L{\'a}szl{\'o} Barab{\'a}si and R{\'e}ka Albert.
\newblock Emergence of scaling in random networks.
\newblock \emph{Science}, 286\penalty0 (5439):\penalty0 509--512, 1999.

\bibitem[Chakrabarti et~al.(2008)Chakrabarti, Wang, Wang, Leskovec, and
  Faloutsos]{chakrabarti2008epidemic}
Deepayan Chakrabarti, Yang Wang, Chenxi Wang, Jure Leskovec, and Christos
  Faloutsos.
\newblock Epidemic thresholds in real networks.
\newblock \emph{{ACM} Trans. Inf. Syst. Secur.}, 10\penalty0 (4):\penalty0
  1:1--1:26, 2008.

\bibitem[Chen et~al.(2016)Chen, Tong, Prakash, Tsourakakis, Eliassi{-}Rad,
  Faloutsos, and Chau]{ChenTPTEFC16}
Chen Chen, Hanghang Tong, B.~Aditya Prakash, Charalampos~E. Tsourakakis, Tina
  Eliassi{-}Rad, Christos Faloutsos, and Duen~Horng Chau.
\newblock Node immunization on large graphs: Theory and algorithms.
\newblock \emph{{TKDE}}, 28\penalty0 (1):\penalty0 113--126, 2016.

\bibitem[Chen et~al.(2009)Chen, Wang, and Yang]{chen2009efficient}
Wei Chen, Yajun Wang, and Siyu Yang.
\newblock Efficient influence maximization in social networks.
\newblock In \emph{{KDD}}, pages 199--208. {ACM}, 2009.

\bibitem[Crossley et~al.(2013)Crossley, Mechelli, V{\'e}rtes, Winton-Brown,
  Patel, Ginestet, McGuire, and Bullmore]{crossley2013cognitive}
Nicolas~A Crossley, Andrea Mechelli, Petra~E V{\'e}rtes, Toby~T Winton-Brown,
  Ameera~X Patel, Cedric~E Ginestet, Philip McGuire, and Edward~T Bullmore.
\newblock Cognitive relevance of the community structure of the human brain
  functional coactivation network.
\newblock \emph{{PNAS}}, 110\penalty0 (28):\penalty0 11583--11588, 2013.

\bibitem[Estrada(2020)]{estrada2020hubs}
Ernesto Estrada.
\newblock ‘{H}ubs-repelling’ {L}aplacian and related diffusion on
  graphs/networks.
\newblock \emph{Linear Algebra Appl.}, 2020.

\bibitem[Fleurquin et~al.(2013)Fleurquin, Ramasco, and
  Eguiluz]{fleurquin2013systemic}
Pablo Fleurquin, Jos{\'e}~J Ramasco, and Victor~M Eguiluz.
\newblock Systemic delay propagation in the {US} airport network.
\newblock \emph{Scientific Reports}, 3:\penalty0 1159, 2013.

\bibitem[Golub and Van~Loan(1996)]{golub1996matrix}
Gene Golub and Charles Van~Loan.
\newblock \emph{Matrix computations}.
\newblock Johns Hopkins Studies in Mathematical Sciences, 1996.

\bibitem[Haghtalab et~al.(2017)Haghtalab, Laszka, Procaccia, Vorobeychik, and
  Koutsoukos]{Haghtalab17}
Nika Haghtalab, Aron Laszka, Ariel~D. Procaccia, Yevgeniy Vorobeychik, and
  Xenofon Koutsoukos.
\newblock Monitoring stealthy diffusions.
\newblock \emph{Knowledge and Information Systems}, 2017.

\bibitem[Ho et~al.(2015)Ho, Kochenderfer, Mehta, and Caceres]{ho2015control}
Christopher Ho, Mykel~J. Kochenderfer, Vineet Mehta, and Rajmonda~S. Caceres.
\newblock Control of epidemics on graphs.
\newblock In \emph{{CDC}}, pages 4202--4207. {IEEE}, 2015.

\bibitem[Kempe et~al.(2003)Kempe, Kleinberg, and Tardos]{kempe2003maximizing}
David Kempe, Jon~M. Kleinberg, and {\'{E}}va Tardos.
\newblock Maximizing the spread of influence through a social network.
\newblock In \emph{{KDD}}, pages 137--146. {ACM}, 2003.

\bibitem[Kempe et~al.(2020)Kempe, Yu, and Vorobeychik]{kempe2020inducing}
David Kempe, Sixie Yu, and Yevgeniy Vorobeychik.
\newblock Inducing equilibria in networked public goods games through network
  structure modification.
\newblock In \emph{{AAMAS}}, pages 611--619, 2020.

\bibitem[Le et~al.(2015)Le, Eliassi{-}Rad, and Tong]{le2015met}
Long~T. Le, Tina Eliassi{-}Rad, and Hanghang Tong.
\newblock {MET:} {A} fast algorithm for minimizing propagation in large graphs
  with small eigen-gaps.
\newblock In \emph{{SDM}}, pages 694--702. {SIAM}, 2015.

\bibitem[Lefkimmiatis et~al.(2013)Lefkimmiatis, Ward, and
  Unser]{lefkimmiatis2013hessian}
Stamatios Lefkimmiatis, John~Paul Ward, and Michael Unser.
\newblock Hessian schatten-norm regularization for linear inverse problems.
\newblock \emph{{IEEE} Trans. Image Process.}, 22\penalty0 (5):\penalty0
  1873--1888, 2013.

\bibitem[Leskovec et~al.(2007{\natexlab{a}})Leskovec, Kleinberg, and
  Faloutsos]{leskovec2007graph}
Jure Leskovec, Jon~M. Kleinberg, and Christos Faloutsos.
\newblock Graph evolution: Densification and shrinking diameters.
\newblock \emph{{ACM} Trans. Knowl. Discov. Data}, 1\penalty0 (1):\penalty0 2,
  2007{\natexlab{a}}.

\bibitem[Leskovec et~al.(2007{\natexlab{b}})Leskovec, McGlohon, Faloutsos,
  Glance, and Hurst]{leskovec2007patterns}
Jure Leskovec, Mary McGlohon, Christos Faloutsos, Natalie~S. Glance, and
  Matthew Hurst.
\newblock Patterns of cascading behavior in large blog graphs.
\newblock In \emph{{SDM}}, pages 551--556. {SIAM}, 2007{\natexlab{b}}.

\bibitem[Leskovec et~al.(2009)Leskovec, Backstrom, and
  Kleinberg]{leskovec2009meme}
Jure Leskovec, Lars Backstrom, and Jon~M. Kleinberg.
\newblock Meme-tracking and the dynamics of the news cycle.
\newblock In \emph{{KDD}}, pages 497--506. {ACM}, 2009.

\bibitem[Meila and Shi(2000)]{meila2001learning}
Marina Meila and Jianbo Shi.
\newblock Learning segmentation by random walks.
\newblock In \emph{{NIPS}}, pages 873--879. {MIT} Press, 2000.

\bibitem[Mieghem et~al.(2009)Mieghem, Omic, and Kooij]{van2008virus}
Piet~Van Mieghem, Jasmina Omic, and Robert~E. Kooij.
\newblock Virus spread in networks.
\newblock \emph{{IEEE/ACM} Trans. Netw.}, 17\penalty0 (1):\penalty0 1--14,
  2009.

\bibitem[Motter and Lai(2002)]{motter2002cascade}
Adilson~E Motter and Ying-Cheng Lai.
\newblock Cascade-based attacks on complex networks.
\newblock \emph{Phys. Rev. E}, 66\penalty0 (6):\penalty0 065102, 2002.

\bibitem[Opsahl(2010)]{airport2010data}
Tore Opsahl.
\newblock {US} airport network traffic data in 2010.
\newblock \url{https://bit.ly/3dlpN3a}, 2010.

\bibitem[Page et~al.(1999)Page, Brin, Motwani, and Winograd]{page1999pagerank}
Lawrence Page, Sergey Brin, Rajeev Motwani, and Terry Winograd.
\newblock The pagerank citation ranking: Bringing order to the web.
\newblock Technical report, Stanford InfoLab, 1999.

\bibitem[Pastor-Satorras et~al.(2015)Pastor-Satorras, Castellano, Van~Mieghem,
  and Vespignani]{pastor2015epidemic}
Romualdo Pastor-Satorras, Claudio Castellano, Piet Van~Mieghem, and Alessandro
  Vespignani.
\newblock Epidemic processes in complex networks.
\newblock \emph{Reviews of modern physics}, 87\penalty0 (3):\penalty0 925,
  2015.

\bibitem[Paszke et~al.(2017)Paszke, Gross, Chintala, Chanan, Yang, DeVito, Lin,
  Desmaison, Antiga, and Lerer]{paszke2017automatic}
Adam Paszke, Sam Gross, Soumith Chintala, Gregory Chanan, Edward Yang, Zachary
  DeVito, Zeming Lin, Alban Desmaison, Luca Antiga, and Adam Lerer.
\newblock Automatic differentiation in pytorch, 2017.

\bibitem[Perozzi et~al.(2014)Perozzi, Al-Rfou, and Skiena]{perozzi2014deepwalk}
Bryan Perozzi, Rami Al-Rfou, and Steven Skiena.
\newblock Deepwalk: Online learning of social representations.
\newblock In \emph{{KDD}}, pages 701--710, 2014.

\bibitem[Prakash et~al.(2012)Prakash, Chakrabarti, Valler, Faloutsos, and
  Faloutsos]{prakash2012threshold}
B.~Aditya Prakash, Deepayan Chakrabarti, Nicholas Valler, Michalis Faloutsos,
  and Christos Faloutsos.
\newblock Threshold conditions for arbitrary cascade models on arbitrary
  networks.
\newblock \emph{Knowl. Inf. Syst.}, 33\penalty0 (3):\penalty0 549--575, 2012.

\bibitem[Preciado et~al.(2013)Preciado, Zargham, Enyioha, Jadbabaie, and
  Pappas]{preciado2013optimal}
Victor~M. Preciado, Michael Zargham, Chinwendu Enyioha, Ali Jadbabaie, and
  George~J. Pappas.
\newblock Optimal vaccine allocation to control epidemic outbreaks in arbitrary
  networks.
\newblock In \emph{{CDC}}, pages 7486--7491. {IEEE}, 2013.

\bibitem[Saha et~al.(2015)Saha, Adiga, Prakash, and
  Vullikanti]{saha2015approximation}
Sudip Saha, Abhijin Adiga, B.~Aditya Prakash, and Anil Kumar~S. Vullikanti.
\newblock Approximation algorithms for reducing the spectral radius to control
  epidemic spread.
\newblock In \emph{{SDM}}, pages 568--576. {SIAM}, 2015.

\bibitem[Seshadhri et~al.(2012)Seshadhri, Kolda, and
  Pinar]{seshadhri2012community}
Comandur Seshadhri, Tamara~G Kolda, and Ali Pinar.
\newblock Community structure and scale-free collections of
  erd{\H{o}}s-r{\'e}nyi graphs.
\newblock \emph{Phys. Rev. E}, 85\penalty0 (5):\penalty0 056109, 2012.

\bibitem[Sun et~al.(2005)Sun, Qu, Chakrabarti, and
  Faloutsos]{sun2005neighborhood}
Jimeng Sun, Huiming Qu, Deepayan Chakrabarti, and Christos Faloutsos.
\newblock Neighborhood formation and anomaly detection in bipartite graphs.
\newblock In \emph{{ICDM}}. IEEE, 2005.

\bibitem[Tong et~al.(2006)Tong, Faloutsos, and Pan]{tong2006fast}
Hanghang Tong, Christos Faloutsos, and Jia-Yu Pan.
\newblock Fast random walk with restart and its applications.
\newblock In \emph{{ICDM}}, pages 613--622. IEEE, 2006.

\bibitem[Tong et~al.(2010)Tong, Prakash, Tsourakakis, Eliassi{-}Rad, Faloutsos,
  and Chau]{tong2010vulnerability}
Hanghang Tong, B.~Aditya Prakash, Charalampos~E. Tsourakakis, Tina
  Eliassi{-}Rad, Christos Faloutsos, and Duen~Horng Chau.
\newblock On the vulnerability of large graphs.
\newblock In \emph{{ICDM}}, pages 1091--1096. {IEEE} Computer Society, 2010.

\bibitem[Tong et~al.(2012{\natexlab{a}})Tong, Prakash, Eliassi{-}Rad,
  Faloutsos, and Faloutsos]{tong2012gel}
Hanghang Tong, B.~Aditya Prakash, Tina Eliassi{-}Rad, Michalis Faloutsos, and
  Christos Faloutsos.
\newblock Gelling, and melting, large graphs by edge manipulation.
\newblock In \emph{{CIKM}}, pages 245--254. {ACM}, 2012{\natexlab{a}}.

\bibitem[Tong et~al.(2012{\natexlab{b}})Tong, Prakash, Eliassi{-}Rad,
  Faloutsos, and Faloutsos]{tong2012gelling}
Hanghang Tong, B.~Aditya Prakash, Tina Eliassi{-}Rad, Michalis Faloutsos, and
  Christos Faloutsos.
\newblock Gelling, and melting, large graphs by edge manipulation.
\newblock In \emph{{CIKM}}, pages 245--254. {ACM}, 2012{\natexlab{b}}.

\bibitem[Torres et~al.(2020)Torres, Chan, Tong, and
  Eliassi-Rad]{torres2020node}
Leo Torres, Kevin~S Chan, Hanghang Tong, and Tina Eliassi-Rad.
\newblock Node immunization with non-backtracking eigenvalues.
\newblock \emph{arXiv preprint arXiv:2002.12309}, 2020.

\bibitem[Trefethen and Bau~III(1997)]{trefethen1997numerical}
Lloyd~N Trefethen and David Bau~III.
\newblock \emph{Numerical linear algebra}, volume~50.
\newblock SIAM, 1997.

\bibitem[Trefethen and Embree(2005)]{trefethen2005spectra}
Lloyd~N Trefethen and Mark Embree.
\newblock \emph{Spectra and pseudospectra: the behavior of nonnormal matrices
  and operators}.
\newblock Princeton University Press, 2005.

\bibitem[Van~Mieghem et~al.(2011)Van~Mieghem, Stevanovi{\'c}, Kuipers, Li, Van
  De~Bovenkamp, Liu, and Wang]{van2011decreasing}
Piet Van~Mieghem, Dragan Stevanovi{\'c}, Fernando Kuipers, Cong Li, Ruud Van
  De~Bovenkamp, Daijie Liu, and Huijuan Wang.
\newblock Decreasing the spectral radius of a graph by link removals.
\newblock \emph{Phys. Rev. E}, 84\penalty0 (1):\penalty0 016101, 2011.

\bibitem[Van~Vu and Hasegawa(2019)]{Vu19}
Tan Van~Vu and Yoshihiko Hasegawa.
\newblock Diffusion-dynamics laws in stochastic reaction networks.
\newblock \emph{Phys. Rev. E}, 99:\penalty0 012416, 2019.

\bibitem[Wang and Deisboeck(2019)]{Wang19}
Zhihui Wang and Thomas~S. Deisboeck.
\newblock Dynamic targeting in cancer treatment.
\newblock \emph{Frontiers in Physiology}, 10:\penalty0 1--9, 2019.

\bibitem[Watts and Strogatz(1998)]{watts1998collective}
Duncan~J Watts and Steven~H Strogatz.
\newblock Collective dynamics of small-world networks.
\newblock \emph{Nature}, 393\penalty0 (6684):\penalty0 440, 1998.

\bibitem[Yang et~al.(2017)Yang, Nishikawa, and Motter]{yang2017small}
Yang Yang, Takashi Nishikawa, and Adilson~E Motter.
\newblock Small vulnerable sets determine large network cascades in power
  grids.
\newblock \emph{Science}, 358\penalty0 (6365):\penalty0 eaan3184, 2017.

\bibitem[Yu and Vorobeychik(2019)]{yu2019removing}
Sixie Yu and Yevgeniy Vorobeychik.
\newblock Removing malicious nodes from networks.
\newblock In \emph{{AAMAS}}, pages 314--322, 2019.

\bibitem[Zhang et~al.(2016)Zhang, Vorobeychik, Letchford, and
  Lakkaraju]{zhang2016data}
Haifeng Zhang, Yevgeniy Vorobeychik, Joshua Letchford, and Kiran Lakkaraju.
\newblock Data-driven agent-based modeling, with application to rooftop solar
  adoption.
\newblock \emph{{JAAMAS}}, 30\penalty0 (6):\penalty0 1023--1049, 2016.

\bibitem[Zhou et~al.(2019)Zhou, Michalak, Waniek, Rahwan, and
  Vorobeychik]{zhou2019attacking}
Kai Zhou, Tomasz~P. Michalak, Marcin Waniek, Talal Rahwan, and Yevgeniy
  Vorobeychik.
\newblock Attacking similarity-based link prediction in social networks.
\newblock In \emph{{AAMAS}}, pages 305--313, 2019.

\end{thebibliography}

\newpage
\section*{Appendix}

\def\thesection{\Alph{section}}
\setcounter{section}{0}

\section{Generalization to Other Diffusion Dynamics}\label{sec:gen}
In this section we discuss generalization of the targeted diffusion model, i.e., Eq.~\eqref{eq:model}, to other common diffusion dynamics.
The fundamental question is: does the heuristic encoded by the model apply to other scenarios with different diffusion dynamics (e.g., SIR or SEIR)?

First,  the feasible region of the  model is independent of the diffusion dynamics, as it is only related to the spectral properties of the underlying graph.
Thus,  the structural (i.e., spectra, degree sequence, and triangle distribution) preserving properties of the diffusion model generalize to other diffusion dynamics.
Next, recall that the objective function of the model is the following
\[
\small
\alpha_1 \lambda_1(\tilde{\Adj}_\mathcal{S}) + \alpha_2 \sigma(\mathcal{S})  + \alpha_3 \phi(\mathcal{S}).
\]
The third term is the normalized cut, which only depends on structural properties of the underlying graph, so it generalizes to any other diffusion dynamics.
The first term also generalizes to many common diffusion dynamics, including SIR and SEIR, as their epidemic thresholds are known to also be determined by the largest eigenvalue of the underlying adjacency matrix ~\cite{prakash2012threshold}.
The only exception is the second term $\sigma(\mathcal{S})$, that is, limiting the impact on non-targeted subset through maximizing the eigencentrality of the targeted subset.
This is because the rationale of maximizing $\sigma(\mathcal{S})$ depends on the steady state of the diffusion dynamics.
Here, the steady state is where in the long run a constant (in average) fraction of infected nodes exist. 
However, both SIR and SEIR have been shown without a steady state, as in the long run all nodes will be in the recovered state (i.e., immune to the diffusion)~\cite{pastor2015epidemic}.

Finally, the certified robustness in Section~\ref{sec:min_budget} generalizes to other diffusion dynamics, as its proof only depends on the spectral properties of the underlying graph.

\section{Degree Sequence and Triangles}\label{app:deg_triangle}

We now show that satisfying the restrictions Eq~\eqref{eq:budget} implies that certain structural properties of the graph will be perturbed by only a small amount. 

Indeed, the principal's action has mild impact on the degree sequence of $G$.
Let $\DegDist=\Adj \bm{1}$ be the vector whose $i$-th entry is the degree of the $i$-th node in the original graph, and similarly, let $\TildeDegDist=\TildeAdj \bm{1}$ be the degree sequence after the perturbation.

\begin{proposition}\label{prop:deg_bound}
\small
The degree sequence of $G$ before and after the perturbation satisfies:
\begin{equation}
\norm{\TildeDegDist - \DegDist}_2 \le \sqrt{n}\epsilon.
\end{equation}
\end{proposition}
\begin{proof}
\small
\begin{equation}
    \begin{aligned}
        \norm{\TildeDegDist - \DegDist}_2 = \norm{\TildeAdj \bm{1} - \Adj \bm{1}}_2 \overset{(a)}{=} \norm{\bm{\Delta} \bm{1}}_2 & = \norm{\bm{\Delta}(1/\sqrt{n})\bm{1} \sqrt{n}}_2 \\
        & \le \sqrt{n} \max_{\norm{x}=1}\norm{\bm{\Delta} \bm{x}}_2 \\
        & \overset{(b)}{=} \sqrt{n} \norm{\bm{\Delta}}_2 \\
        & \le \sqrt{n} \epsilon,
    \end{aligned}   
\end{equation}
where $(a)$ is due to the fact that $\TildeAdj - \Adj = \bm{\Delta}$, and $(b)$ comes from the definition of spectral norm.% of $\bm{\Delta}$, that is $\norm{\bm{\Delta}}_2 := \max_{\norm{x}=1}\norm{\bm{\Delta} \bm{x}}_2$.
\end{proof}

A direct corollary of Proposition~\ref{prop:deg_bound} concerns the average degree of $G$.
\begin{corollary}\label{cor:avg-deg-bound}
\small
The average degree of $G$ after the perturbation is within $\epsilon$ of the average degree before the perturbation:
\begin{equation}
\left| d_{avg}(G; \TildeAdj) - d_{avg}(G; \Adj) \right| \le \epsilon.
\end{equation}
\end{corollary}
\begin{proof}
Note that $d_{avg}(\PerturG) = \frac{\bm{1}^\top \TildeDegDist}{n}$ and $d_{avg}(G) = \frac{\bm{1}^\top \DegDist}{n}$.
Thus we have:
\begin{equation}
\begin{aligned}
\left| \bm{1}^\top \TildeDegDist/n -  \bm{1}^\top \DegDist / n \right| & = (1/n)\left| \bm{1}^\top(\TildeDegDist - \DegDist) \right| \\
& \le (1/n) \norm{\bm{1}}_2 \cdot \norm{\TildeDegDist - \DegDist}_2 \\
& = \epsilon.
\end{aligned}
\end{equation}
\end{proof}

Next, we perform a similar analysis for the number of triangles before and after the perturbation.
\begin{proposition}\label{prop:triangles}
\small
Assume $G$ is unweighted with $m$ edges and $T$ triangles. Suppose the number of triangles after the perturbation is $\tilde{T}$. Then we have
\begin{equation}
    |T - \tilde{T}| \leq \epsilon m,
\end{equation}
where the estimate is correct up to a first order approximation.
\end{proposition}
\begin{proof}
\small
Since $G$ is unweighted, we have $T = \Tr \left({\Adj}^3 \right) / 6$, where $\Tr$ is the trace operator.
%Recall that $\Tr(X) = \sum_i^n \lambda_i(X)$, where $\lambda_i(X)$ are the eigenvalues of the $n \times x$ matrix $X$, counted with multiplicity.
%Recall also that $\lambda_i(X)^k$ is an eigenvalue of $X^k$.
The restrictions Eq.~\eqref{eq:budget} guarantee that we can write $\lambda_i ( \TildeAdj ) = \lambda_i\left( \Adj \right) + \eta_i \epsilon$, where $\eta_i \in [-1 ,1]$. Thus,
\begin{equation}
\begin{aligned}
6 \tilde{T} = \sum_i^n \lambda_i( \TildeAdj^3 ) = \sum_i^n \lambda_i( \TildeAdj )^3 = \sum_i^n \left( \lambda_i( \Adj ) + \eta_i \epsilon \right)^3.
\end{aligned}
\end{equation}
Expanding the cube and neglecting the terms of higher order in $\eta_i$, we have
\begin{equation}
\begin{aligned}
6 \tilde{T} \approx \sum_i^n \lambda_i(\Adj)^3 + 3 \epsilon \sum_i^n \lambda_i(\Adj)^2 \eta_i &= 6T + 3 \epsilon \sum_i^n \lambda_i(\Adj)^2 \eta_i.
\end{aligned}
\end{equation}
And thus
\begin{equation}
\begin{aligned}
2 |\tilde{T} - T| \approx \epsilon \left| \sum_i^n \lambda_i(\Adj)^2 \eta_i \right| & \leq \epsilon \sum_i^n \left|\lambda_i(\Adj)^2 \eta_i \right|  \\
        & \leq  \epsilon \sum_i^n \lambda_i(\Adj)^2 \\
        & = \epsilon \Tr(\Adj^2).
\end{aligned}
\end{equation}
Since $\Adj$ is symmetric and binary, we have $\Tr(\Adj^2) = 2m$, where $m$ is the number of edges in $G$.
\end{proof}

In what follows we present experimental results to show that the spectra and degree sequences do not change a lot due to the targeted diffusion.
The spectra and degree sequences of $G$ and $\PerturG$ are showed in Figure~\ref{fig:spec_deg}, in which the top row presents the spectra with the eigenvalues as ranked in descending order, and the bottom row presents the degree sequences.
The three columns (from left to right) correspond to the email network, the airport network, and the brain network, respectively.
The parameter $\gamma$ is set to $0.5$, the most powerful principal.

\noindent{\bf The Email Network:} From Figure~\ref{fig:spec_deg} (top row), the eigenvalues with large value admit the largest deviation, while the bottom row of that figure shows that the degree sequence is not significantly affected by the targeted diffusion.
In fact, the change to the original degree sequence is mild, and a student's t-test cannot differentiate the modified degree sequence from the original one (p-value=0.081).

\noindent{\bf The Airport and The Brain Networks: }
\syedit{The airport network is directed and each node pair is associated with two edges in opposite directions. 
We convert the network to an undirected one by substituting an undirected edge for the two edges.
The weight of the undirected edge is the sum of the weights on the two edges.  
}
The spectra and degree sequences of $G$ and $\PerturG$ are showed in the last two columns of Figure~\ref{fig:spec_deg}.
As we can see from Figure~\ref{fig:spec_deg} (top row), the graph spectrum is again nearly preserved, except the eigenvalues with small values admit some deviation.
Similarly, the modified degree sequences cannot be differentiated from the original one by student's t-tests (airport: p-value=0.4969, brain: p-value=0.9919).

\begin{figure}[h]
\def\FigSize{1.1in}
\centering
\setlength{\tabcolsep}{0.01pt}
\begin{tabular}{ccc}
\includegraphics[width=\FigSize]{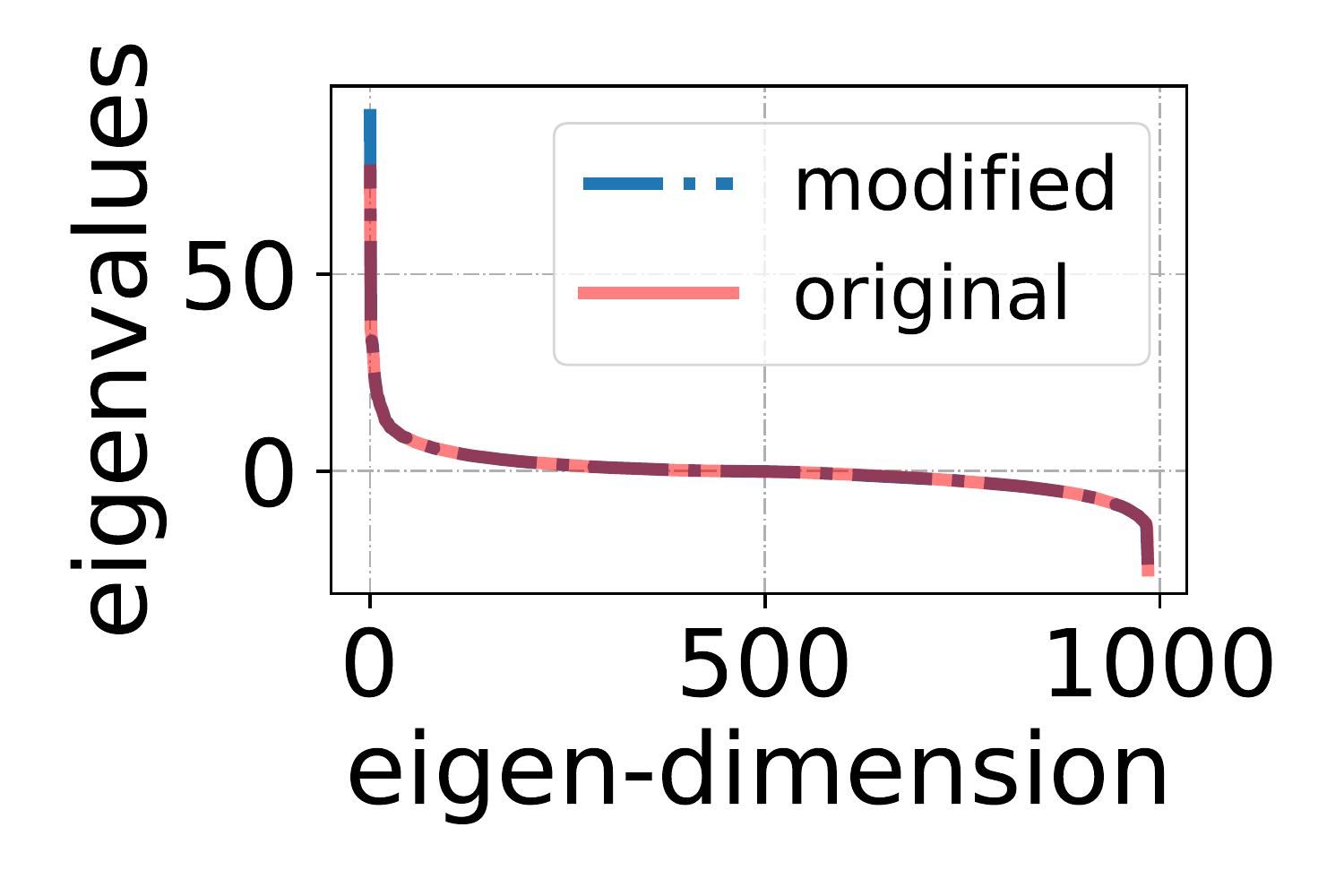} &
\includegraphics[width=\FigSize]{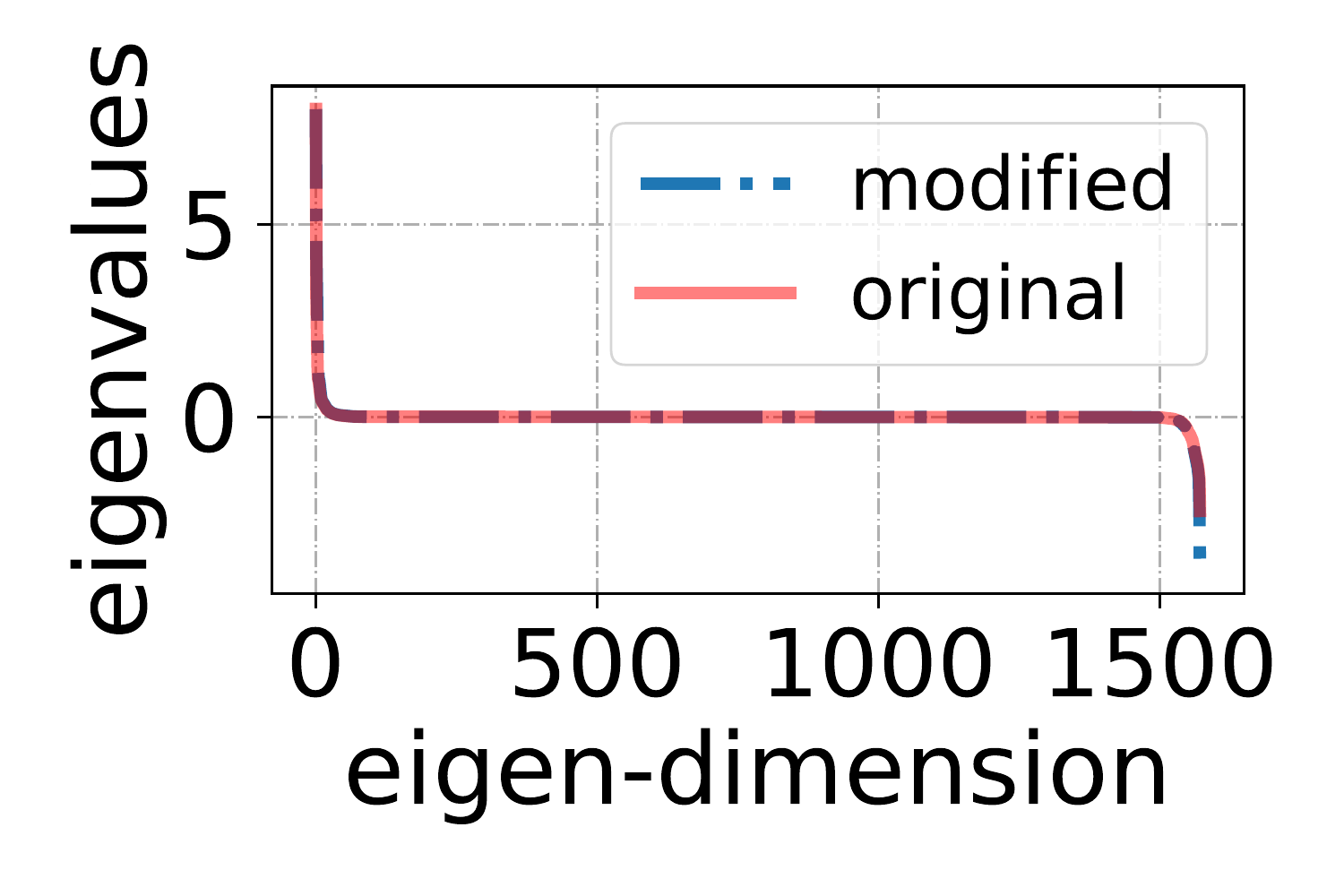} & 
\includegraphics[width=\FigSize]{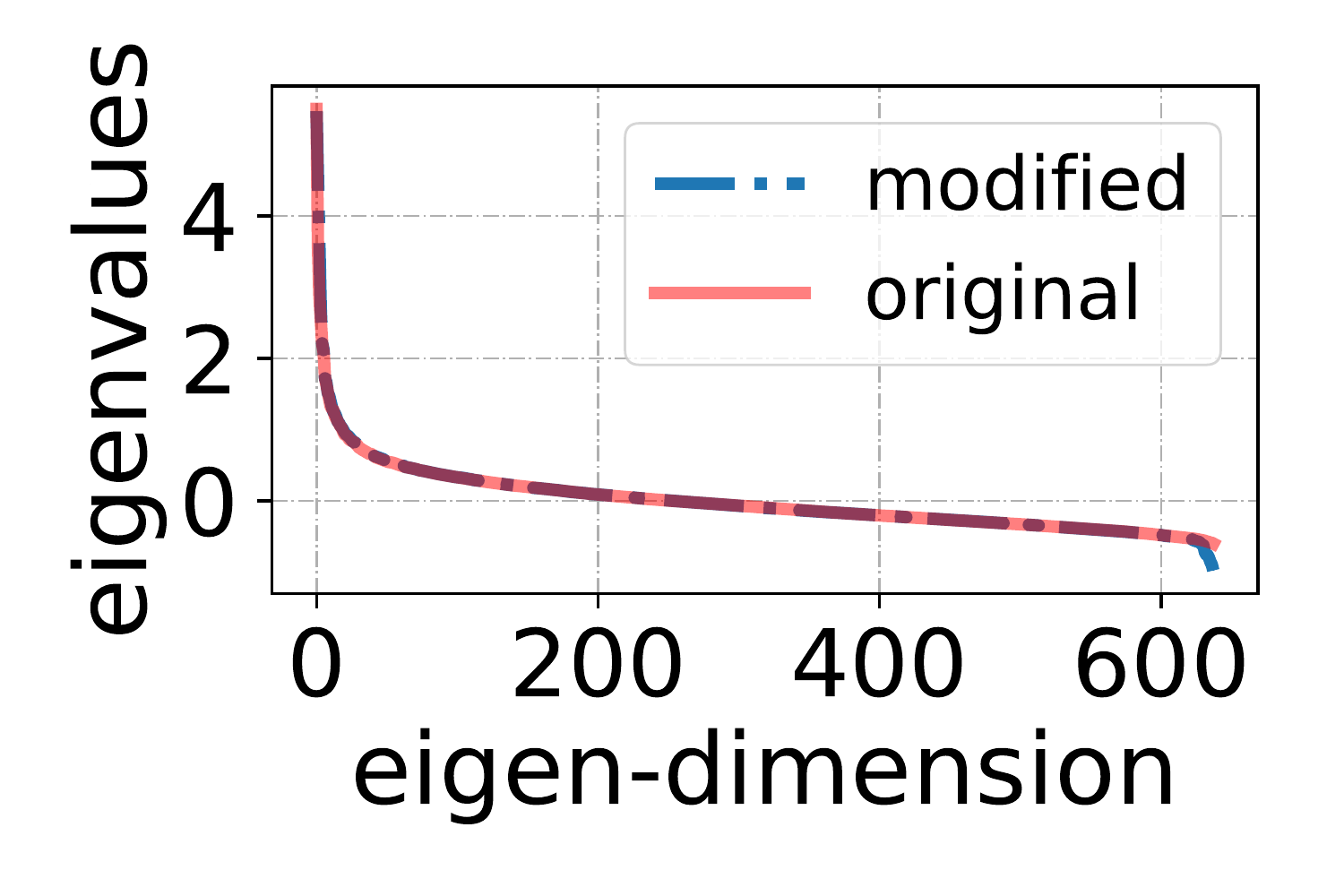} \\
\includegraphics[width=\FigSize]{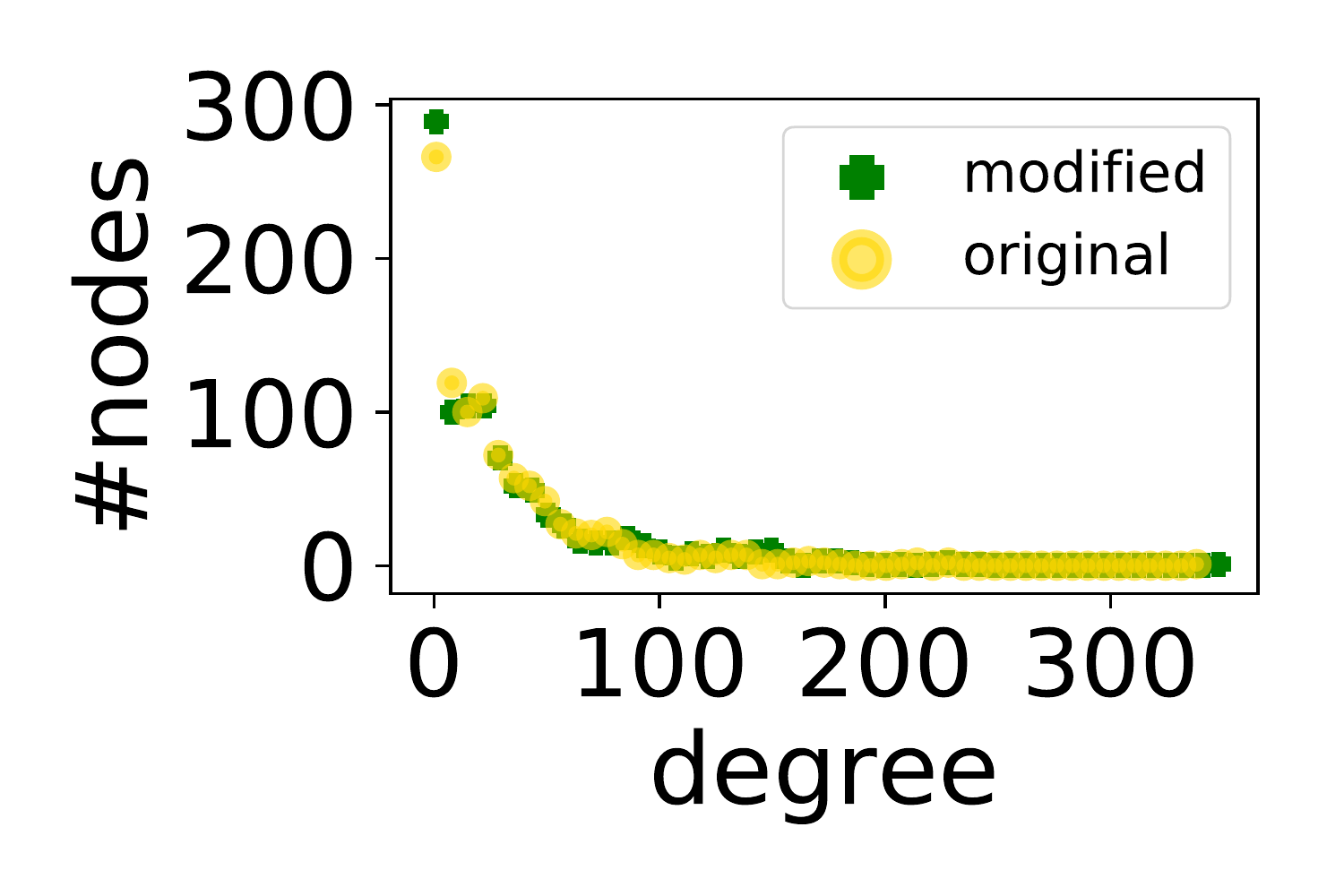}  &
\includegraphics[width=\FigSize]{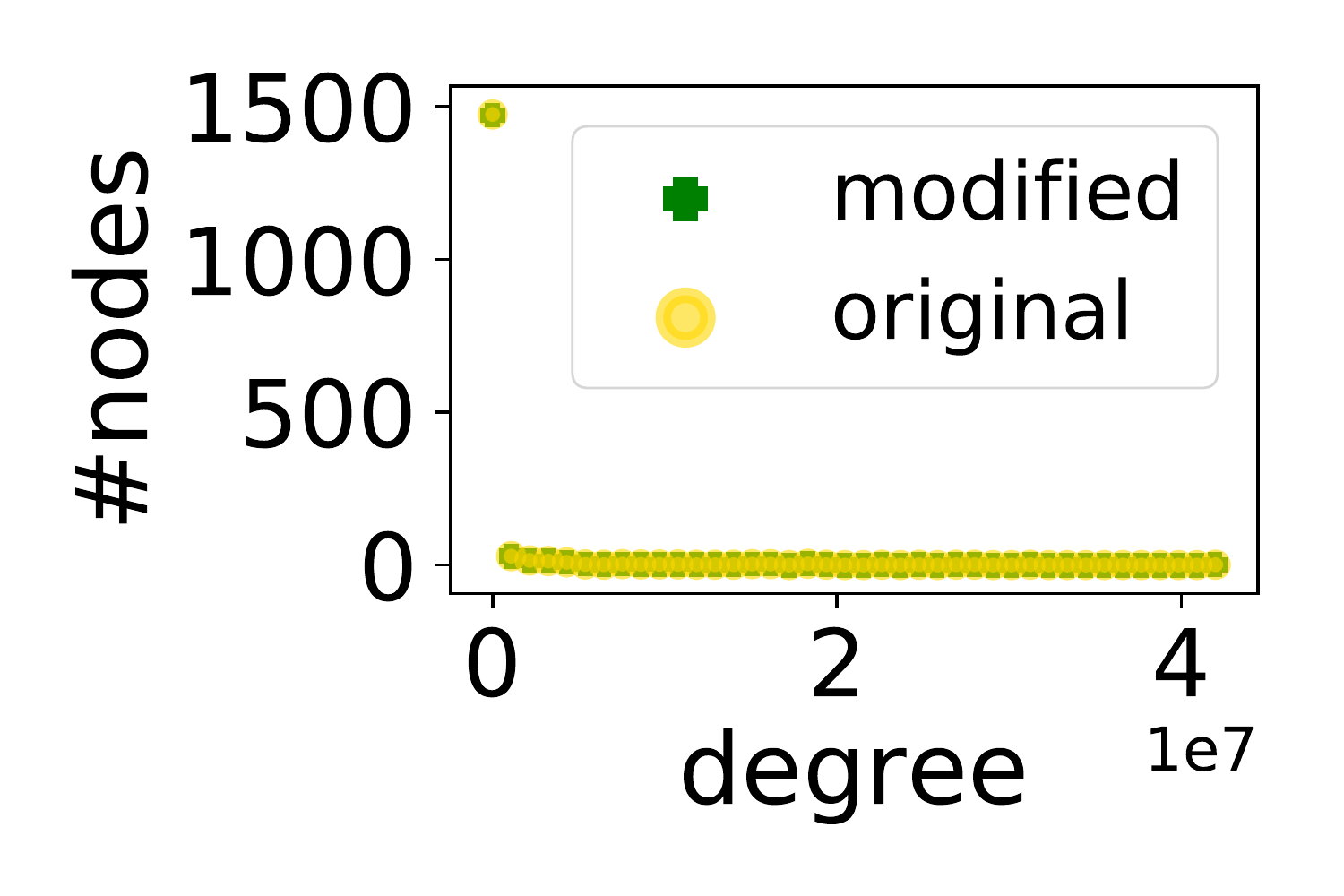} &
\includegraphics[width=\FigSize]{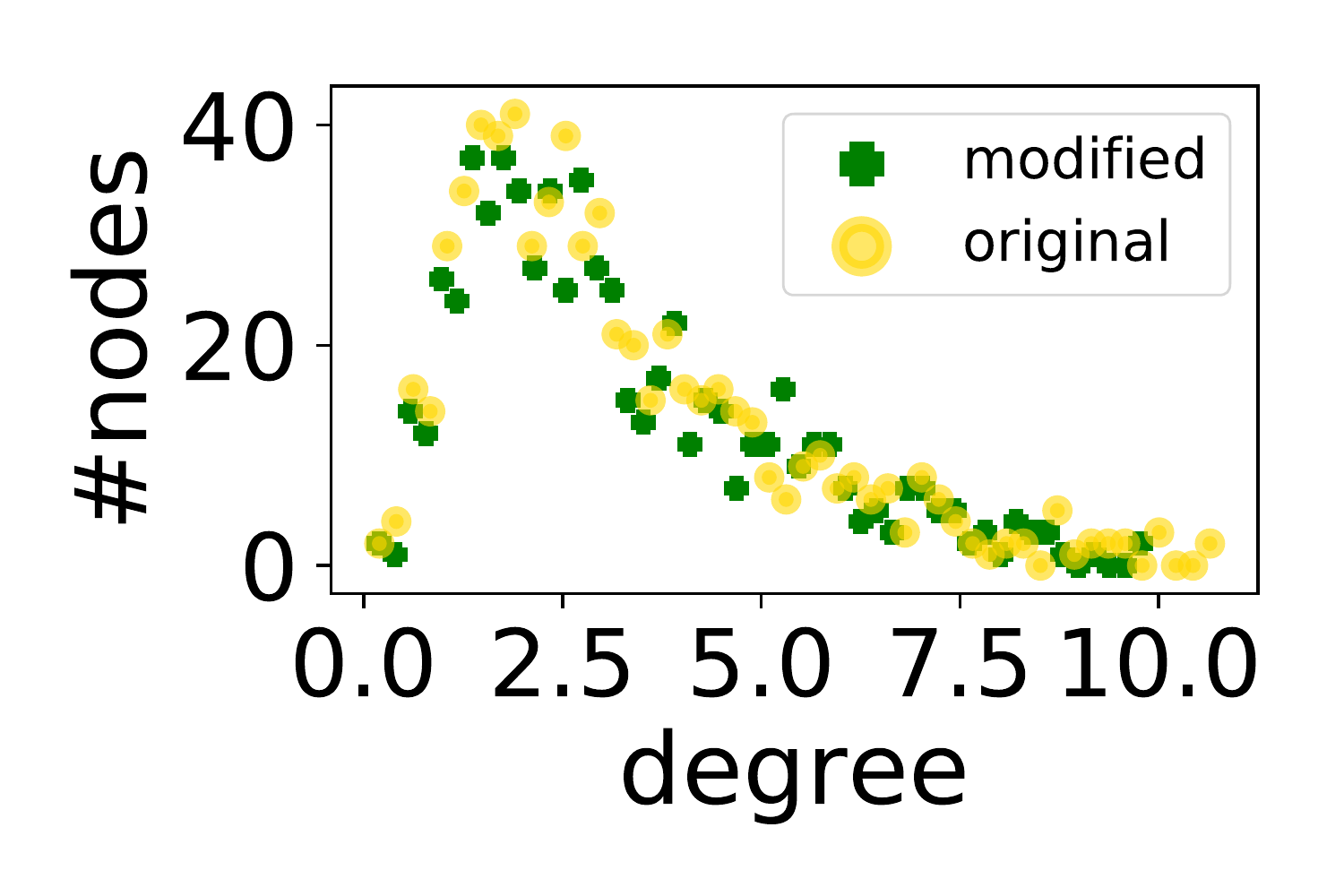} 
\end{tabular}
\vspace{-0.15in}
\caption{Spectra (top) and degree sequences (bottom) for the \textbf{left:} the email network; \textbf{middle:} the airport network; and \textbf{right:} the brain network. The hyper-parameters are set to $\alpha_1=\alpha_2=\alpha_3=1/3$.}
\label{fig:spec_deg}
\end{figure}

\section{Proof of Theorem~\ref{th:cert_robust}}\label{app:cert_robust}
\begin{proof}
\small
From the discussion in the main paper, an instance \texttt{TargetDiff}($\mathcal{S}, G, \epsilon$) can be encoded by the following meta model:
\begin{equation}
\begin{aligned}
& \max_{\TildeAdj} & &  I(\PerturGS) - I(\GS) \\
&s.t.     &    &    \TildeAdj \in \mathcal{P}.
\end{aligned}
\end{equation}
As discussed in~\cite{van2008virus} (see Section IV.B), computing the exact value of $I(\GS)$ is intractable, since the exact computation of $\pi_i$ is challenging.
Our model Eq.~\eqref{eq:model} can be thought of as a tractable proxy to the meta model. 
An estimation to $I(\GS)$ is given in \cite{van2008virus}, i.e., $\hat{I}(\GS) = \sum_{i \in \mathcal{S}}^{}{1 - \delta / (\beta d_i)}$, where $d_i$ is the degree of node $i$ in $G$. 
The estimator works in the region $\delta / \beta \le d_{min}$, where $d_{min}$ is the minimum degree of $G$.
When the estimation is reasonably good, that is $| \hat{I}(\GS) - I(\GS) | \le \tau$, we have the following relation:
\begin{equation}
    \begin{aligned}
I(\PerturGS) - I(\GS) > 2\tau & \implies (\hat{I}(\PerturGS) + \tau) - (\hat{I}(\GS) - \tau) > 2\tau \\
                              & \implies \hat{I}(\PerturGS) - \hat{I}(\GS) > 0.
    \end{aligned}
\end{equation}
Thus, in what follows we focus on deriving the necessary condition for $\hat{I}(\PerturGS) - \hat{I}(\GS) > 0$, which directly translates to the necessary condition for $I(\PerturGS) - I(\GS) > 2\tau$.

Suppose there exists an adjacency matrix $\TildeAdj^* \in \mathcal{P}$ such that $I(\PerturGS) - I(\GS) > 2\tau$.
This indicates that the corresponding instance \texttt{TargetDiff}($\mathcal{S}, G, \epsilon$) is successful.
Consequently, it follows that $\hat{I}(\PerturGS) - \hat{I}(\GS) > 0$.
Recall that $\hat{I}(\PerturGS) - \hat{I}(\GS)= \frac{\delta}{\beta} \sum_{i \in \mathcal{S}}^{}{( \frac{1}{d_i} - \frac{1}{\tilde{d}_i} )}$.
Let $\TildeDegDistS \in \mathbb{R}^{|\mathcal{S}|}$ represent the degree sequence of  nodes in $\mathcal{S}$.
Due to Proposition~\ref{prop:deg_bound} we have $\norm{  \TildeDegDistS - \bm{d}_\mathcal{S} }_2^2 \le \norm{\TildeDegDist - \bm{d}}_2^2 \le n \epsilon^2$.
Consider the optimization problem in Eq.~\eqref{eq:opt_robust}, where the objective function is $\hat{I}(\PerturGS) - \hat{I}(\GS)$ (up to a multiplicative factor).
The last constraint follows from the assumption that for a successful instance the degrees of nodes in $\mathcal{S}$ increase.
The fact that $\hat{I}(\PerturGS) - \hat{I}(\GS) > 0$ implies that the optimal solution of Eq.~\eqref{eq:opt_robust} exists and the associated objective value is greater than zero.

\begin{equation}\label{eq:opt_robust}
\begin{aligned}
& \max_{\TildeDegDistS} & &    \sum_{i \in \mathcal{S}}^{}{\left( \frac{1}{d_i} - \frac{1}{\tilde{d}_i} \right)} \\
&s.t.    &    &    \norm{  \TildeDegDistS - \bm{d}_\mathcal{S} }_2^2 \le n \epsilon^2 \\
&         &    &  \TildeDegDistS \succeq \bm{d}_\mathcal{S}.
\end{aligned}
\end{equation}

Denote the feasible region of the above optimization problem as $\mathcal{M}$.
Note that $\mathcal{M}$ is a convex set since it is the intersection of two convex sets.

The objective function is concave in $\TildeDegDistS$, since it is twice differentiable on the feasible region $\mathcal{M}$ and the Hessian matrix is negative definite; the Hessian matrix is a diagonal matrix with the $i$-th diagonal element being $-2 / \tilde{d}_i^3$. 
Thus, Eq.~\eqref{eq:opt_robust} is a convex optimization problem.
Note that the Slater's condition is satisfied (e.g., with $\TildeDegDist_\mathcal{S} = \bm{d}_\mathcal{S}$), which indicates that strong duality holds. 
Thus, the KKT conditions are satisfied at \emph{any}  primal and dual optimal solutions. 

For convenience, in what follows we use $d_i$ (resp. $\tilde{d}_i$) to represent the degree of a node $i \in \mathcal{S}$ before (resp. after) graph modification.
The Lagrange function of Eq.~\eqref{eq:opt_robust} is:
\begin{equation}
\begin{aligned}
\mathcal{L}(\TildeDegDistS, \lambda, \bm{\beta}) & = \sum_{i \in \mathcal{S}}^{}{\left( \frac{1}{d_i} - \frac{1}{\tilde{d}_i} \right)} + \lambda \left( n\epsilon^2 - \norm{  \TildeDegDistS - \bm{d}_\mathcal{S} }_2^2 \right)  \\
& + \bm{\beta}^\top \left( \TildeDegDistS - \bm{d}_\mathcal{S} \right),
\end{aligned}
\end{equation}
where $\lambda \ge 0$ and $\bm{\beta} \succeq \bm{0}$ are Lagrangian multipliers.
Recall that the degrees of nodes in $\mathcal{S}$ are  increased, i.e., $\tilde{d}_i \ge d_i$ for all $i \in \mathcal{S}$.
Note that for a node $i \in \mathcal{S}$ such that $\tilde{d}_i = d_i$, we let the corresponding $\beta_i=0$.
Thus, by complementary slackness, we have $\beta_i=0$ for all $i \in \mathcal{S}$.
The gradient of $\mathcal{L}(\TildeDegDistS, \lambda, \bm{\beta})$ w.r.t. $\tilde{d}_i$ becomes:
\[
\frac{\partial \mathcal{L}}{\partial \tilde{d}_i} = \frac{1}{\tilde{d}^2_i} - 2\lambda \tilde{d}_i + \beta_i = \frac{1}{\tilde{d}^2_i} - 2\lambda \tilde{d}_i.
\]
Setting the gradient to zero leads to:
\[
\tilde{d}_i = \left( \frac{1}{2\lambda} \right)^{1/3}, \forall i \in \mathcal{S}.
\]
Since the optimal solution exists, we have $\lambda \ne 0$. 
By complementary slackness we have $\lambda \left( n\epsilon^2 - \norm{  \TildeDegDistS - \bm{d}_\mathcal{S} }_2^2 \right)=0$, which indicates:
\[
n \epsilon^2 = \norm{  \TildeDegDistS - \bm{d}_\mathcal{S} }_2^2. 
\]
Expand the above equation:
\begin{equation}
\begin{aligned}
   n \epsilon^2 & =  \norm{  \TildeDegDistS - \bm{d}_\mathcal{S} }_2^2 \\
                            & = \sum_{i \in \mathcal{S}}^{}{\left( \tilde{d}_i - d_i \right)^2} \\
                            & = \sum_{i \in \mathcal{S}}^{}{\left( \left(\frac{1}{2\lambda}\right)^{1/3} - d_i \right)^2}.
\end{aligned}
\end{equation}
% Define a constant $K = \sum_{i \in \mathcal{S}_2}^{}{\left( \left(\frac{1}{2\lambda}\right)^{1/3} - d_i \right)^2} \ge 0$, which leads to:
% \[
% n \epsilon^2 + K = \sum_{i \in \mathcal{S}}^{}{\left( \left(\frac{1}{2\lambda}\right)^{1/3} - d_i \right)^2}
% \]
Substitute $\left(\frac{1}{2\lambda}\right)^{1/3}$ with a variable $x$ and re-arrange the above equation:
\[
x^2 - \frac{2\sum_{i \in \mathcal{S}}^{}{d_i}}{|\mathcal{S}|}x + \frac{\sum_{i\in \mathcal{S}}^{}{d^2_i - n\epsilon^2}}{|\mathcal{S}|} = 0.
\]
According to vieta theorem, a necessary condition that we can solve for $x \in \mathbb{R}$ from the above equation is:
\[
\left( \frac{2\sum_{i \in \mathcal{S}}^{}{d_i}}{|\mathcal{S}|}  \right)^2 - 4 \left(\frac{\sum_{i\in \mathcal{S}}^{}{d^2_i - n\epsilon^2}}{|\mathcal{S}|} \right) \ge 0,
\]
which leads to:
\[
\epsilon \ge \sqrt{\frac{|\mathcal{S}|}{n}} \left( \frac{\sum_{i \in \mathcal{S}}^{}{d^2_i}}{|\mathcal{S}|} - \frac{\left(\sum_{i \in \mathcal{S}}^{}{d_i} \right)^2}{|\mathcal{S}|^2} \right)^{1/2}.
\]
\end{proof}

\section{Additional Results on Real Networks}\label{sec:real_exp}
We run three experiments to show the effectiveness of each term in the objective function of \ModelName.
The results are showed in Figure~\ref{fig:control}.
The three rows (from top to bottom) correspond to experimental results on the email, airport and brain networks.

The first column corresponds to the first experiment, which is to show
 that maximizing $\lambda_1(\TildeAdjS)$ leads to higher infected ratios.
The hyper-parameters are set to $\alpha_1=1/3, \alpha_2=0$, and $\alpha_3=1/3$ (the hyper-parameters do not need to sum to one).
The labels of the $y$-axis become $I_{\text{modified}}^\mathcal{S} - I_{\text{original}}^\mathcal{S}$, which highlights that the infected ratios are for the targeted subgraph $\GS$ (the higher the better).
Note that $\alpha_2$ is set to zero in order to avoid the coupling between the eigenvector centrality of $\mathcal{S}$ and $\lambda_1(\TildeAdjS)$.
From the plot it is clear that maximizing $\lambda_1(\TildeAdjS)$ is important to increase the infected ratios within $\GS$ (the blue line).
Note that solely maximizing the normalized cut of $\mathcal{S}$ may backfire (the red line), as a large portion of edges are deleted from $\GS$ when $\gamma$ increases.  

The second column is to show the effectiveness of limiting the impact on $\GSPrime$ by maximizing the eigenvector centrality of $\mathcal{S}$.
The $y$-axis represents $I_{\text{modified}}^\mathcal{S'} - I_{\text{original}}^\mathcal{S'}$, which highlights that the infected ratios are for the non-targeted subgraph $\GSPrime$ (the lower the better). 
The plot shows that the impact on $\GSPrime$ is well limited; the effectiveness is most significant when $\gamma > 40\%$.

The last column is to show the effectiveness of maximizing the normalized cut of $\mathcal{S}$.
We set $\alpha_2=0$ to avoid the effect of maximizing the eigenvector centrality of $\mathcal{S}$.
Observe that maximizing the normalized cut of $\mathcal{S}$ is effective in increasing the infected ratio within $\GS$ only for the email network.
This suggests that on weighted graphs normalized cut might not be a good heuristic to increase the centrality of $\mathcal{S}$.

\begin{figure}[t]
\def\FigSize{1.1in}
\centering
\setlength{\tabcolsep}{0.1pt}
\begin{tabular}{cccc}
\includegraphics[width=\FigSize]{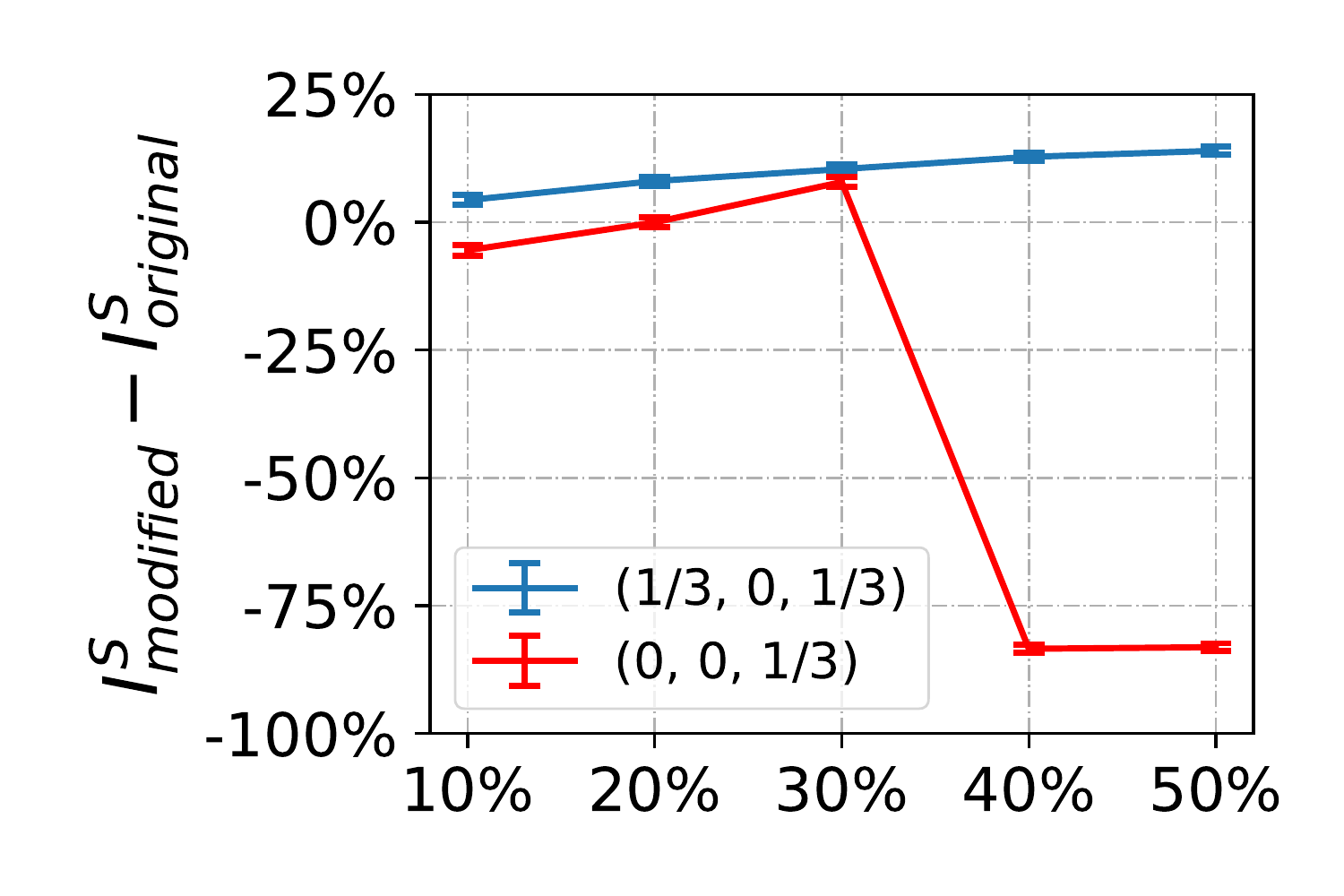} & \includegraphics[width=\FigSize]{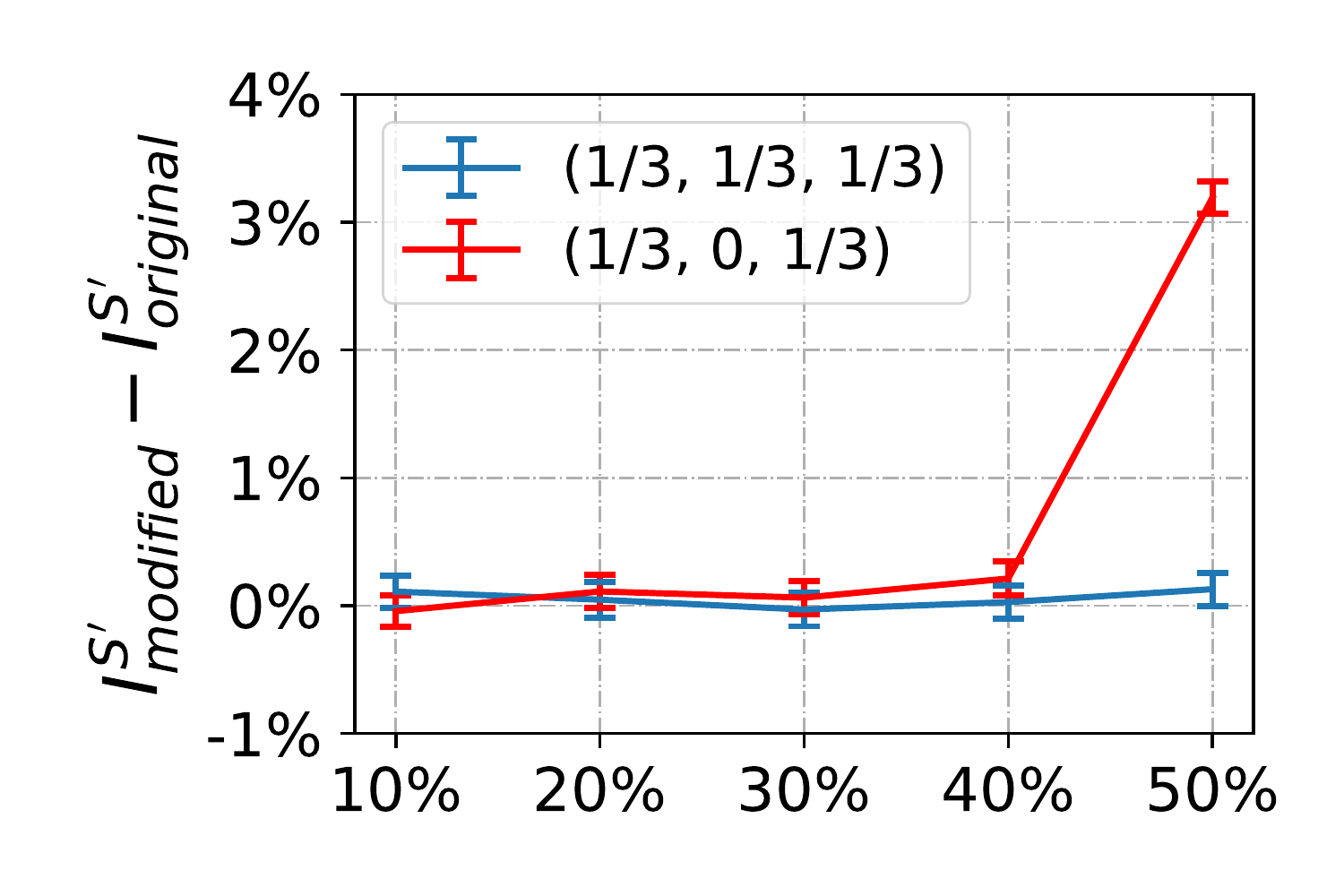} & \includegraphics[width=\FigSize]{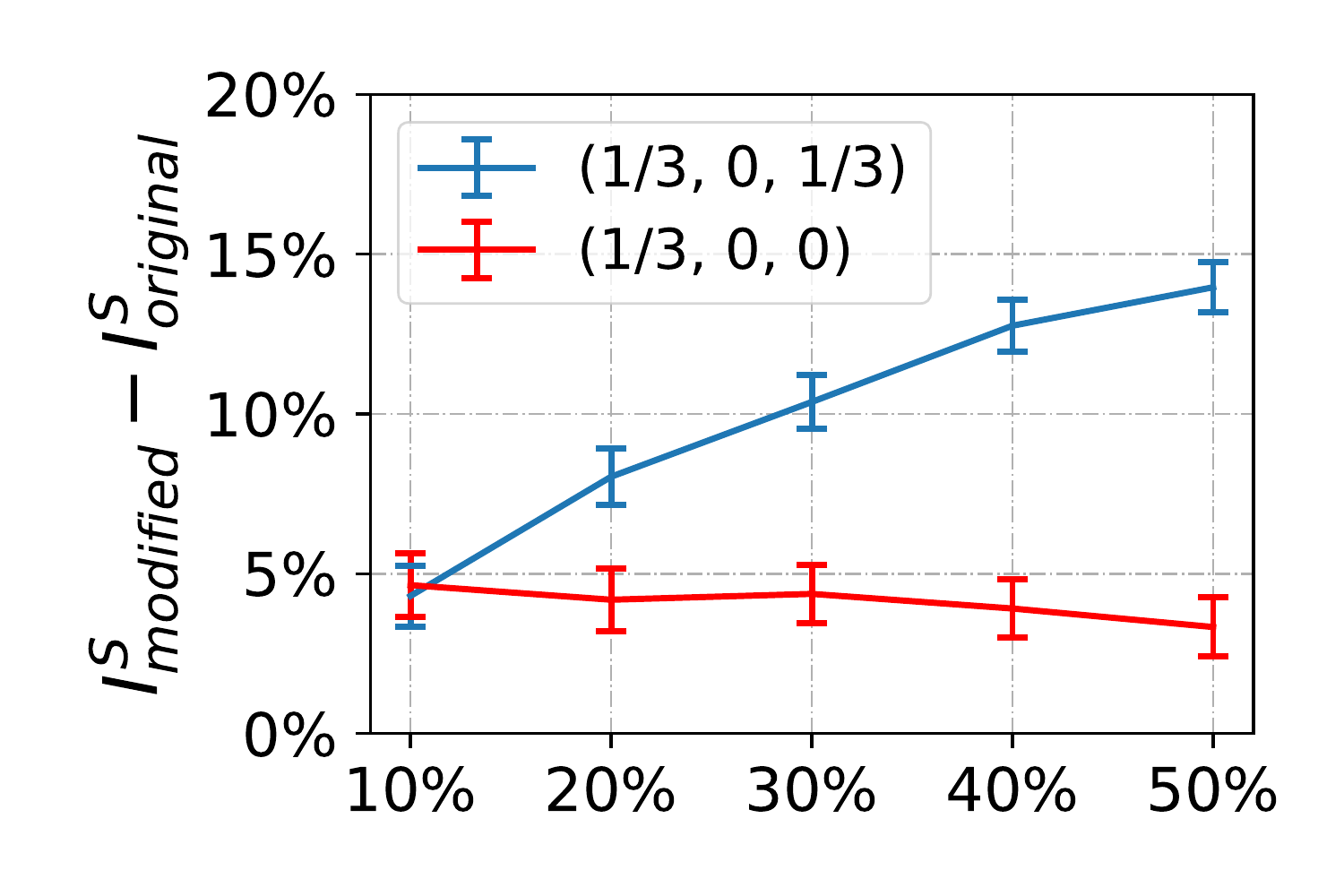} \\
\includegraphics[width=\FigSize]{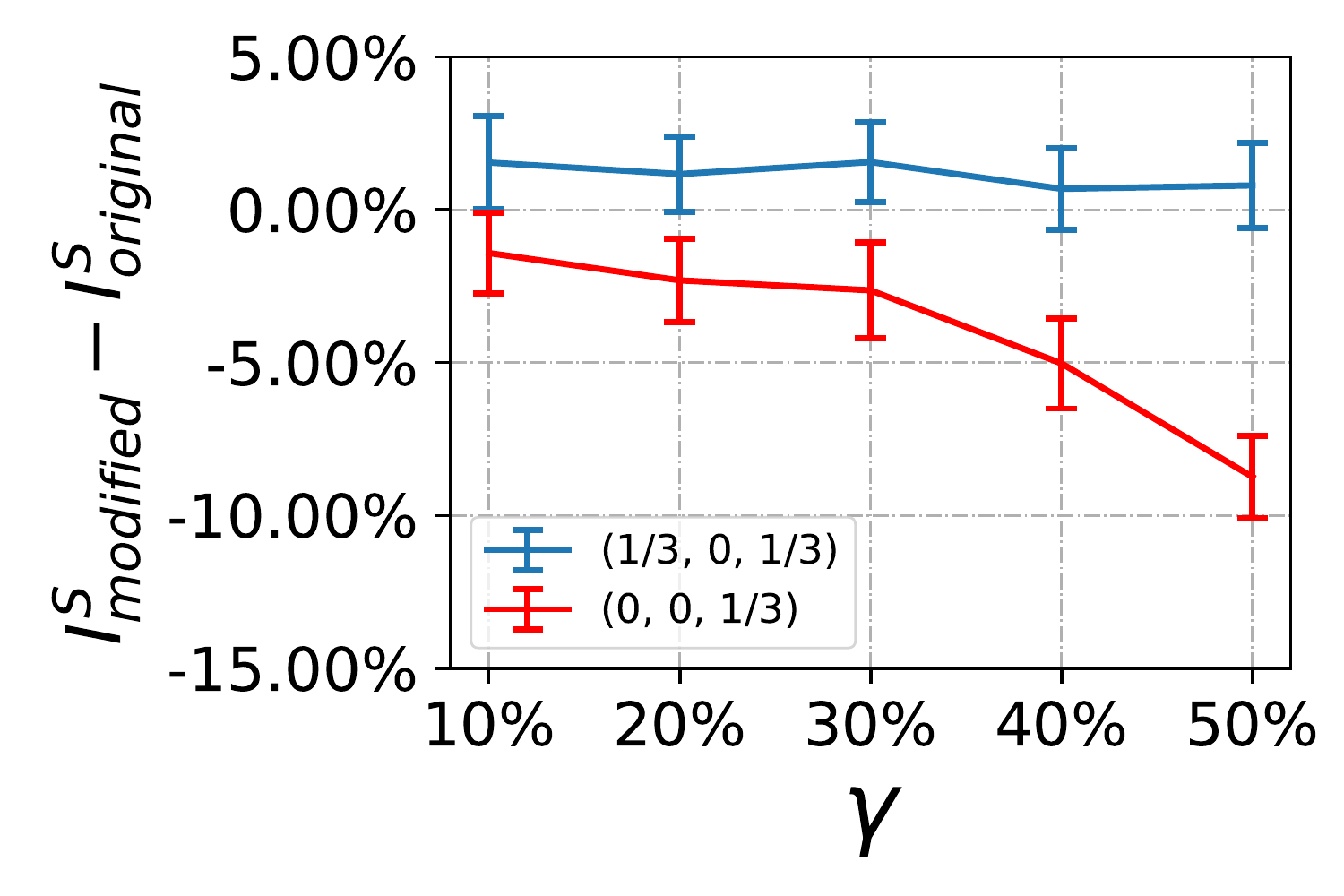} & \includegraphics[width=\FigSize]{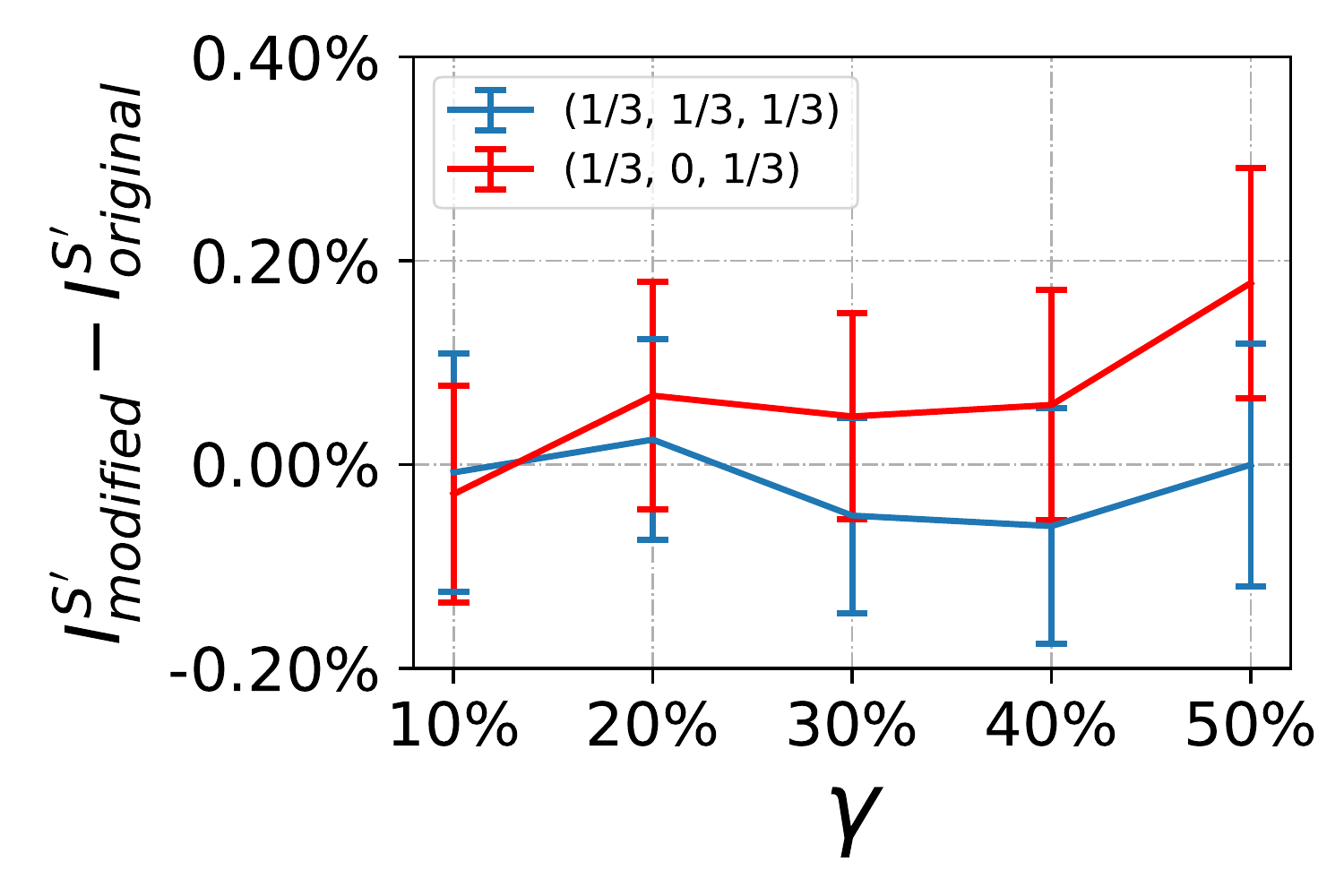} & \includegraphics[width=\FigSize]{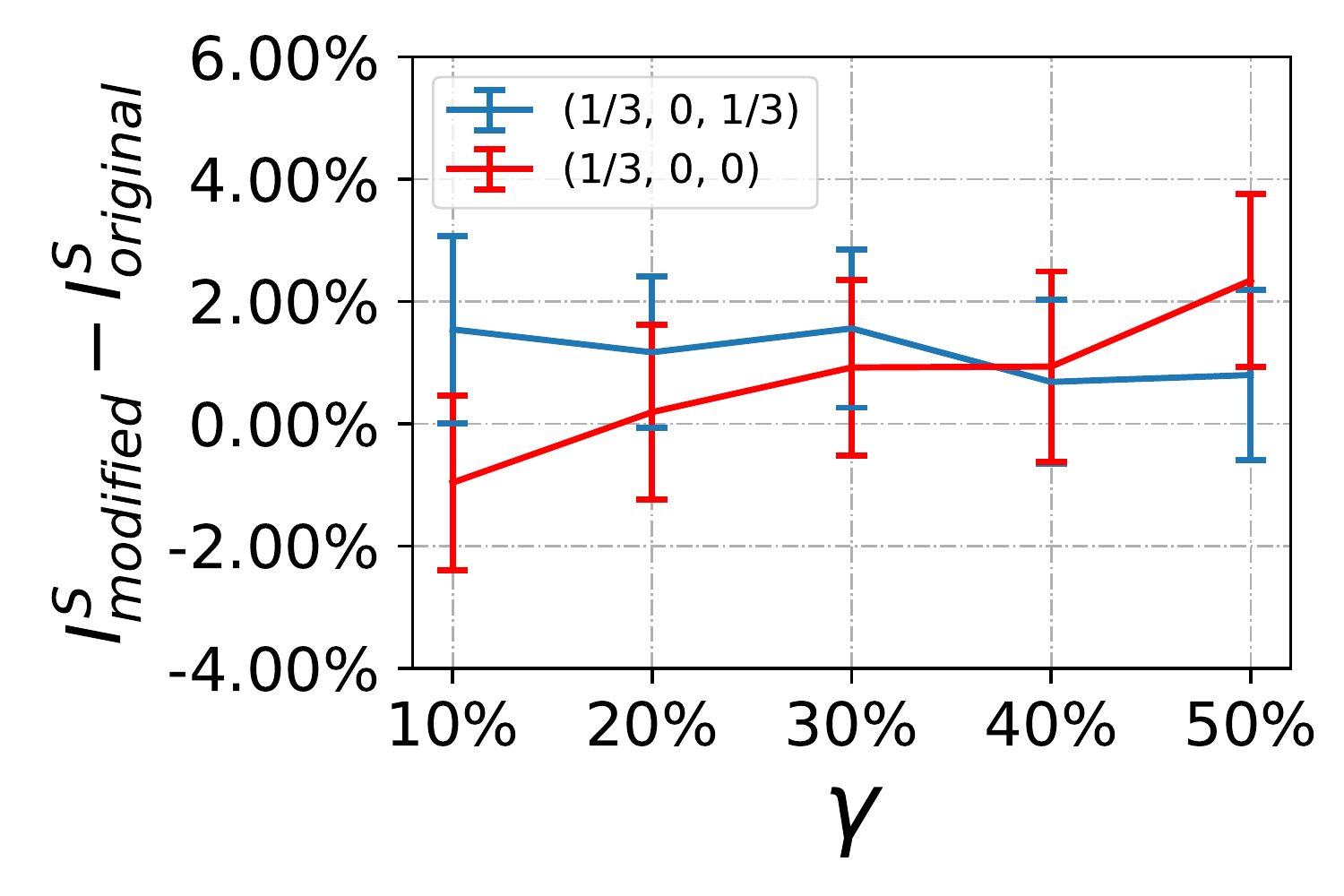} \\
\includegraphics[width=\FigSize]{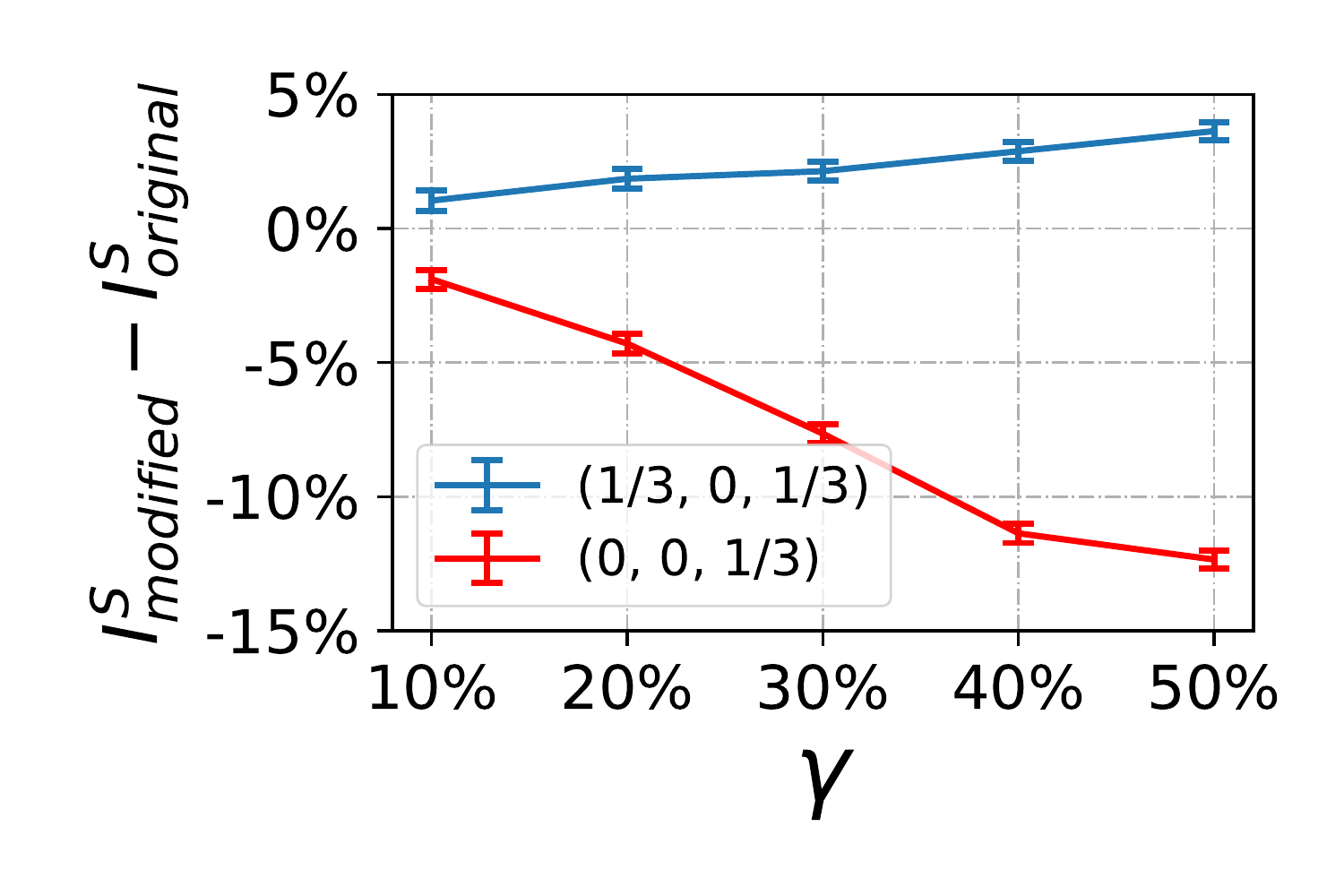} & \includegraphics[width=\FigSize]{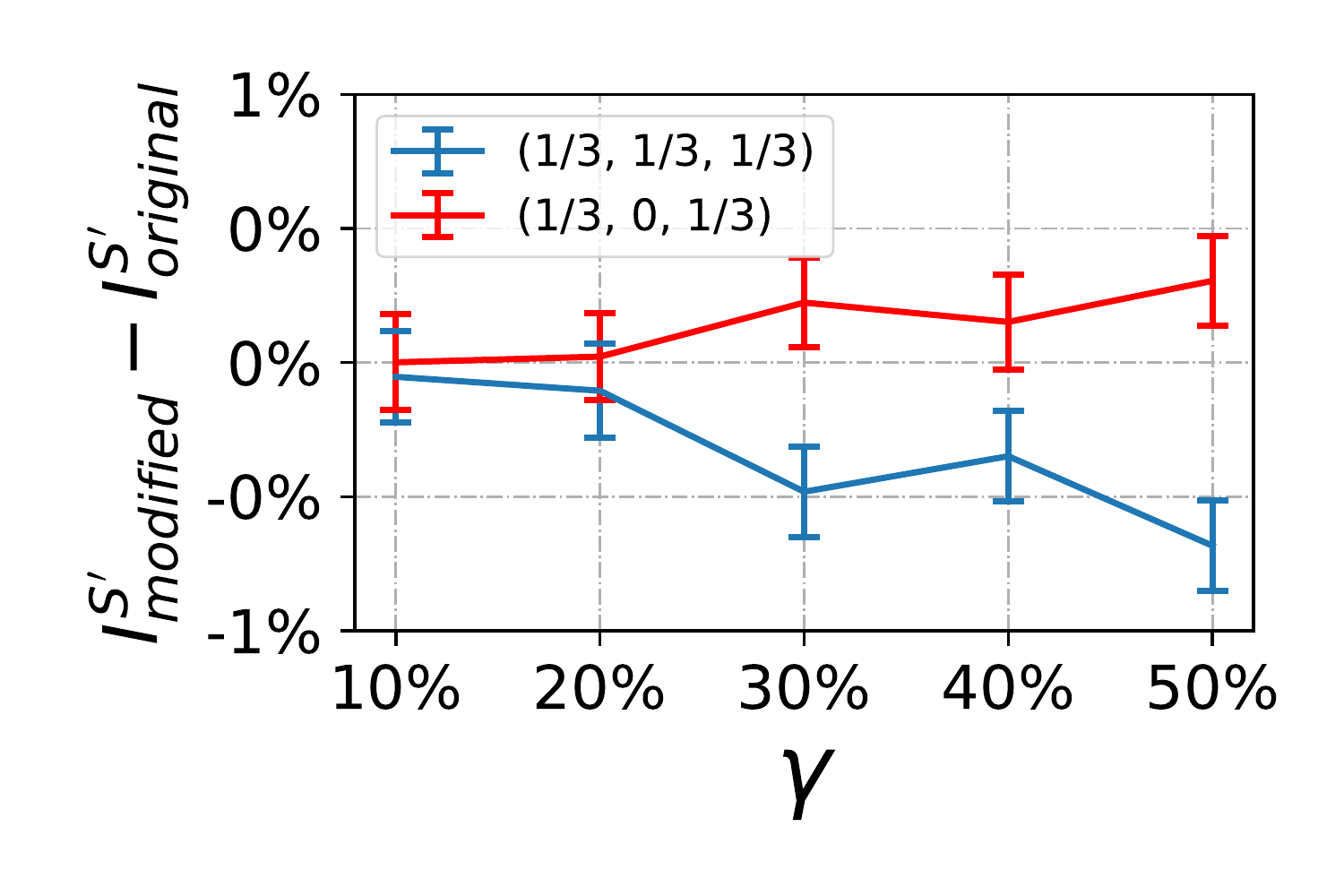} & \includegraphics[width=\FigSize]{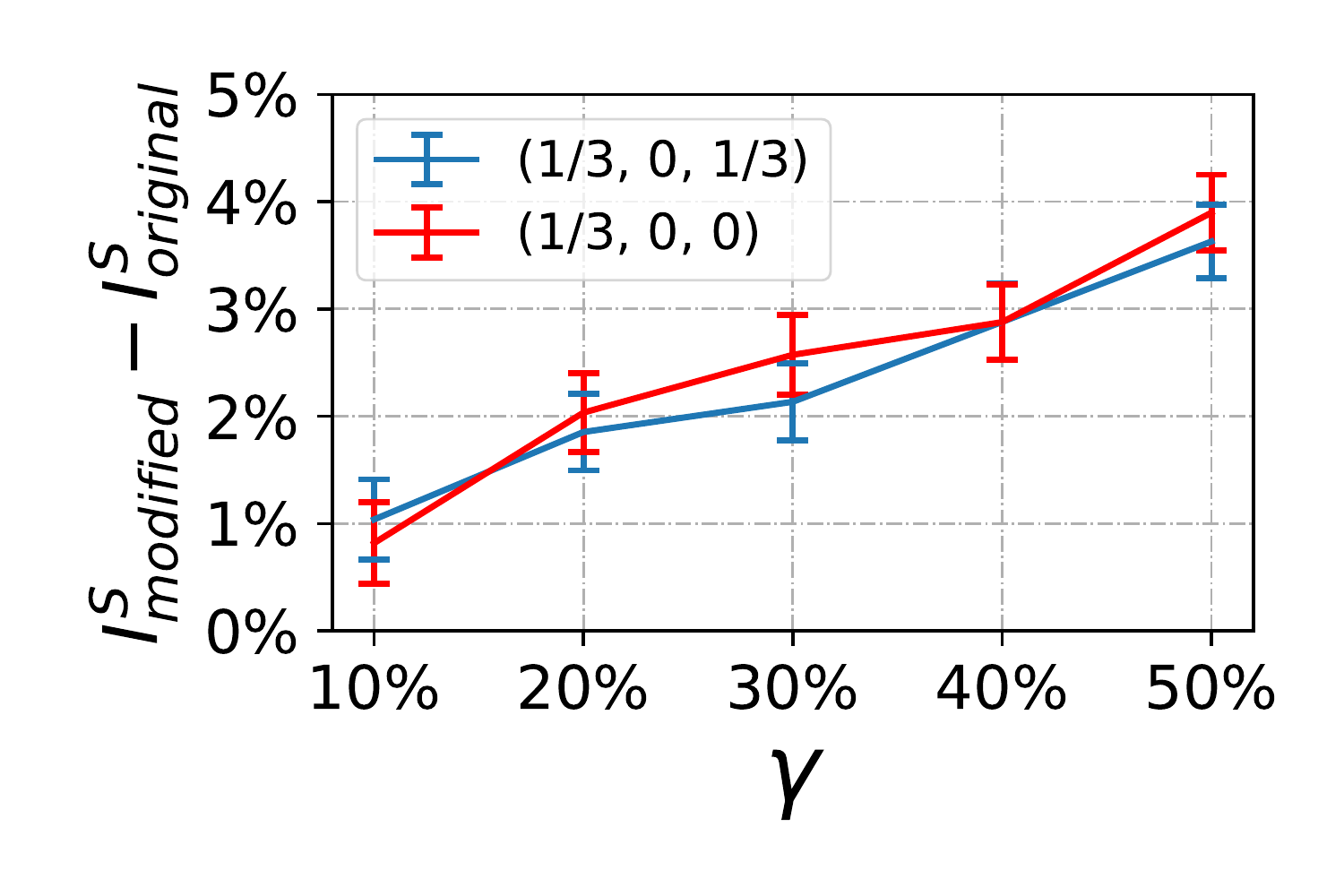}
\end{tabular}
\caption{
\small
Experiments showing the model's effectiveness.
\textbf{Top}: the email network; \textbf{Middle}: the airport network; \textbf{Bottom}: the brain network.
The three columns (from left to right) show the effectiveness of: 1) maximizing $\lambda_1(\TildeAdjS)$; 2) Maximizing the eigenvector centrality of $\mathcal{S}$; 3) maximizing the normalized cut of $\mathcal{S}$. }
\label{fig:control}
\end{figure}

\section{Additional Results on Synthetic Networks}\label{sec:synthetic_exp}
In this section we show experimental results on synthetic unweighted graphs with 375 nodes.
We focus on three classes of networks: Barabasi-Albert (BA), Watts-Strogatz, and BTER~\cite{seshadhri2012community}.
BA is characterized by its power-law degree distribution~\cite{barabasi1999emergence}.
Watts-Strogatz is well-known for its local clustering in a way as to qualitatively resemble real networks~\cite{watts1998collective}.
BTER are generative network models that can be calibrated to match real-world networks, in particular, to reproduce the community structures~\cite{seshadhri2012community}.

The experimental setup is similar to the setup for the email network, except for a few changes.
First, the experimental results for each class of the synthetic networks are averaged over 30 randomly generated network topologies. 
Another difference lies in how the targeted set $\mathcal{S}$ is selected.
For each randomly generated network, the targeted set $\mathcal{S}$ is selected as the node whose degree is the 90 percentile of the degree sequence,  and its neighbors. 
Some statistics of the synthetic networks are summarized in Table~\ref{tab:sync-stat}.
Recall that $\delta=0.24$ and $\beta=0.06$.
The experimental results are showed in Figure~\ref{fig:synthetic}.
The conclusion derived from  Figure~\ref{fig:synthetic} is similar to that of the email network. 
It is worth pointing out that maximizing the normalized cut of $\mathcal{S}$ is  effective on BA networks, while for other network it may  backfire.

\begin{table}[h]
\small
\centering
\begin{tabular}{@{}cccc@{}}
\toprule
                          & BA   & Watts-Strogatz & BTER \\ \midrule
$|\mathcal{S}|$ & 17.5 & 12 & 20.03 \\
$d_{min}$       & 9.86 & 10 & 11.69 \\
density                   & 0.02 & 0.03        & 0.03 \\
average degree            & 9.87 & 10          & 11.5 \\
average clustering coeff. & 0.08 & 0.35        & 0.05 \\
\bottomrule
\end{tabular}
\caption{Statistics of synthetic networks.}
\label{tab:sync-stat}
\end{table}

\begin{figure*}[h]
\def\FigSize{1in}
\centering
\setlength{\tabcolsep}{0.1pt}
\begin{tabular}{cccc}
\includegraphics[width=\FigSize]{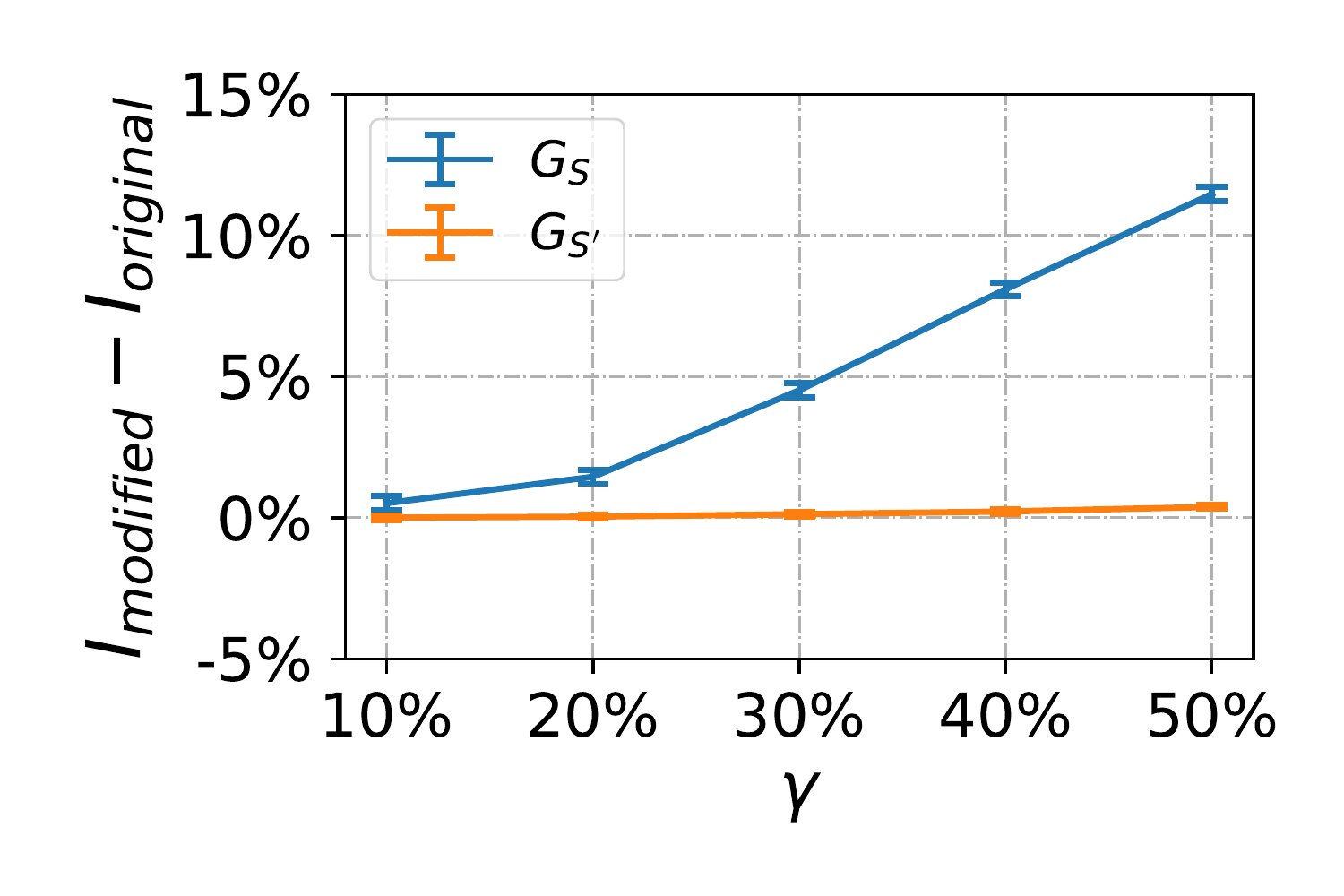} & \includegraphics[width=\FigSize]{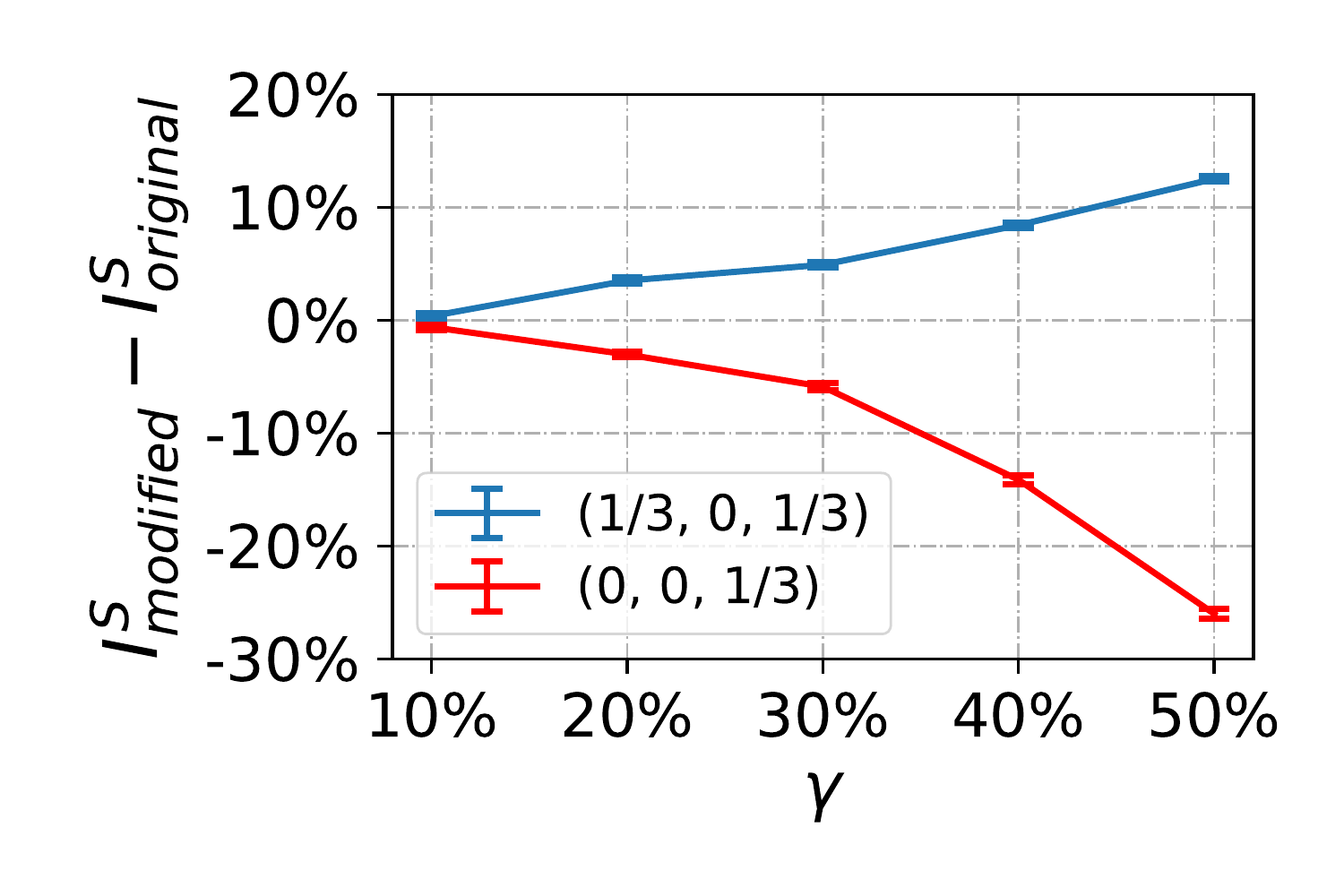} & \includegraphics[width=\FigSize]{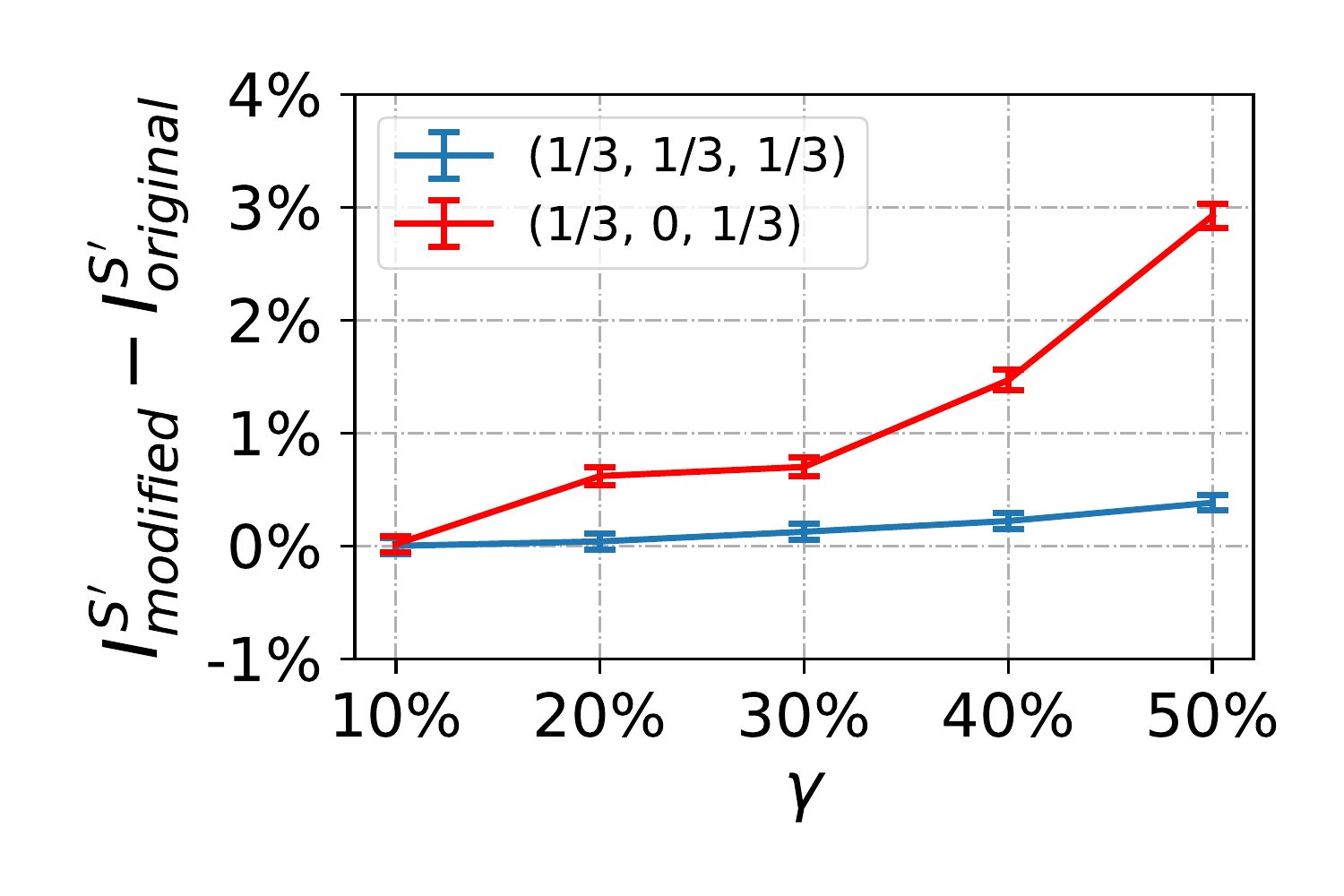} & \includegraphics[width=\FigSize]{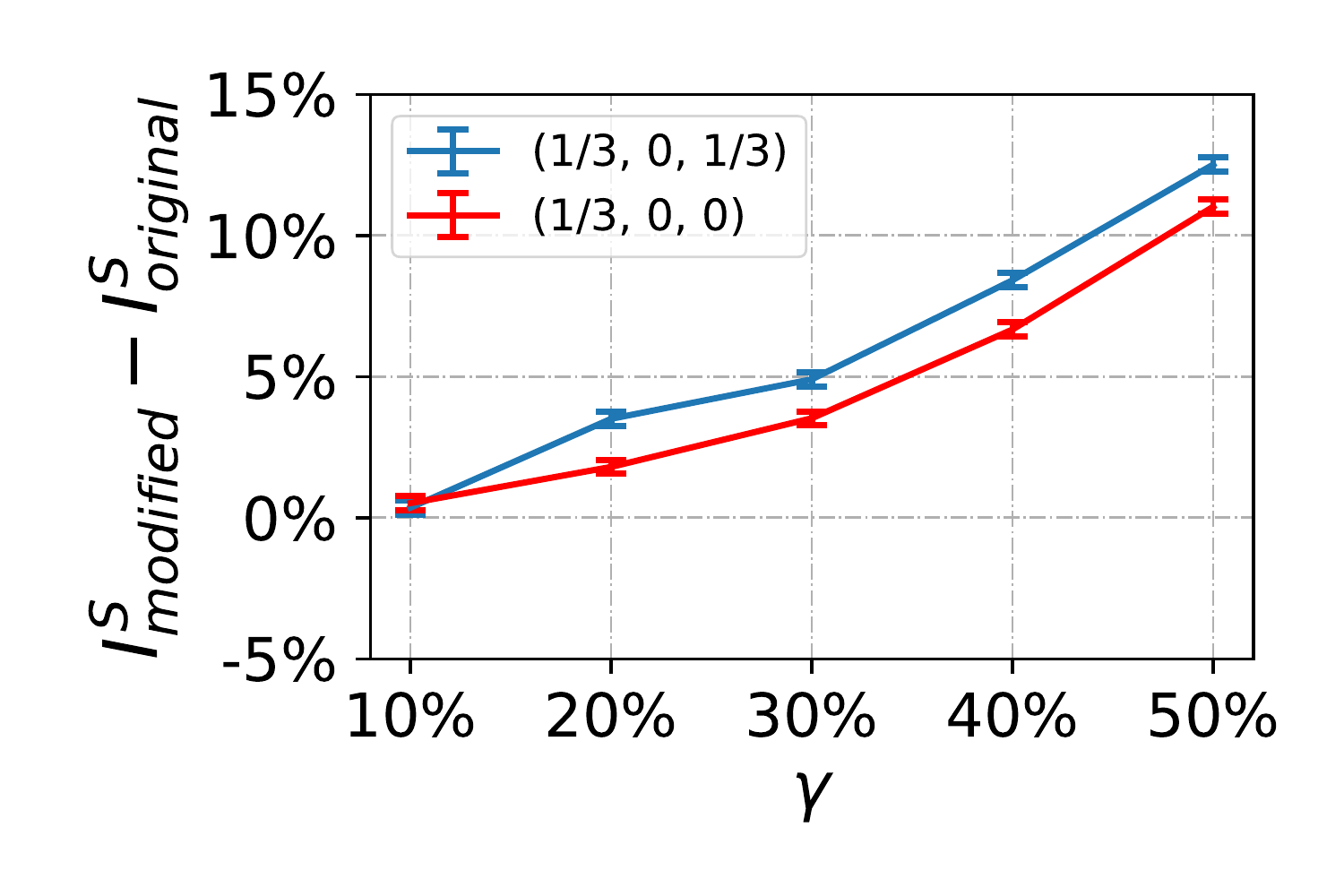} \\
\includegraphics[width=\FigSize]{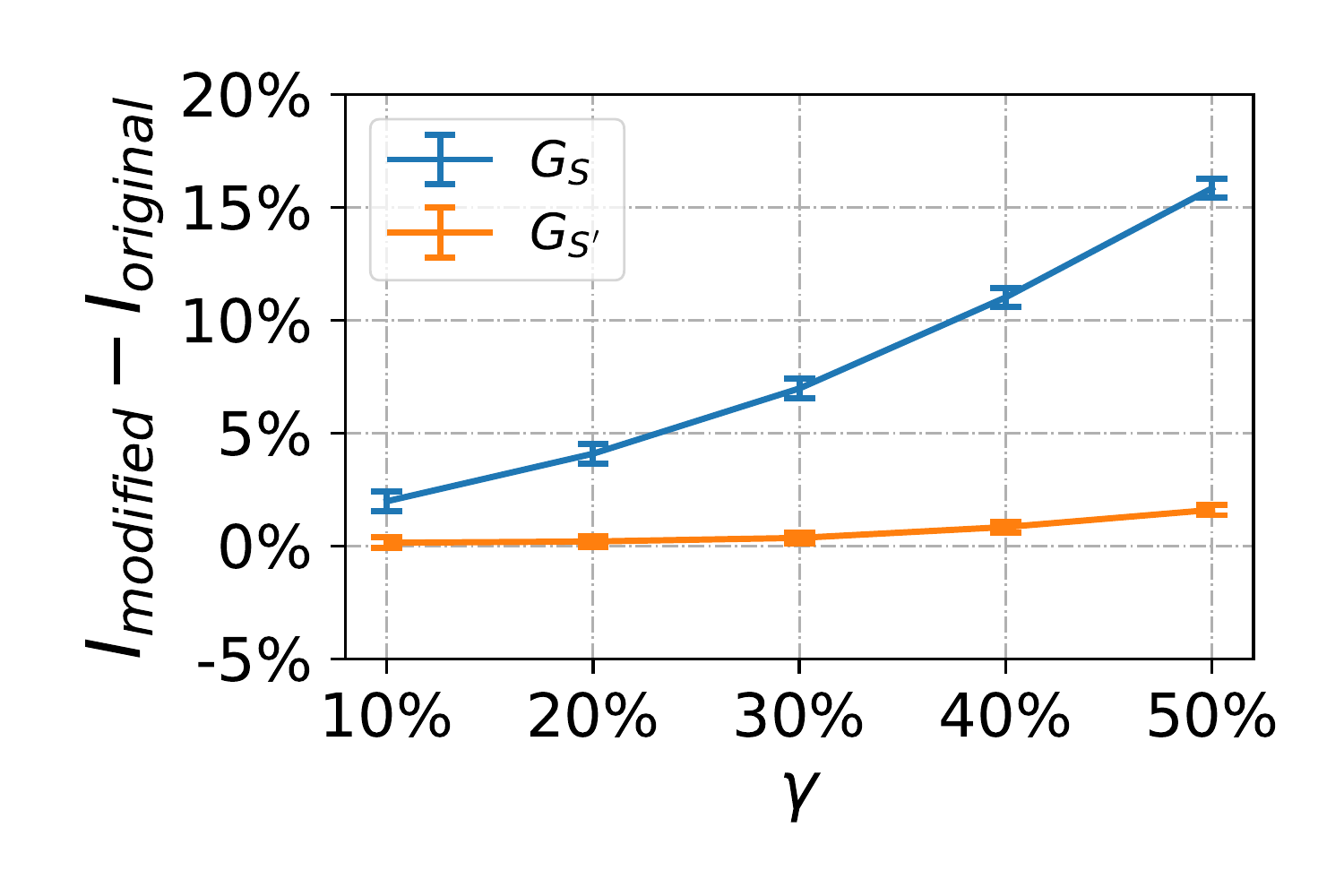} & \includegraphics[width=\FigSize]{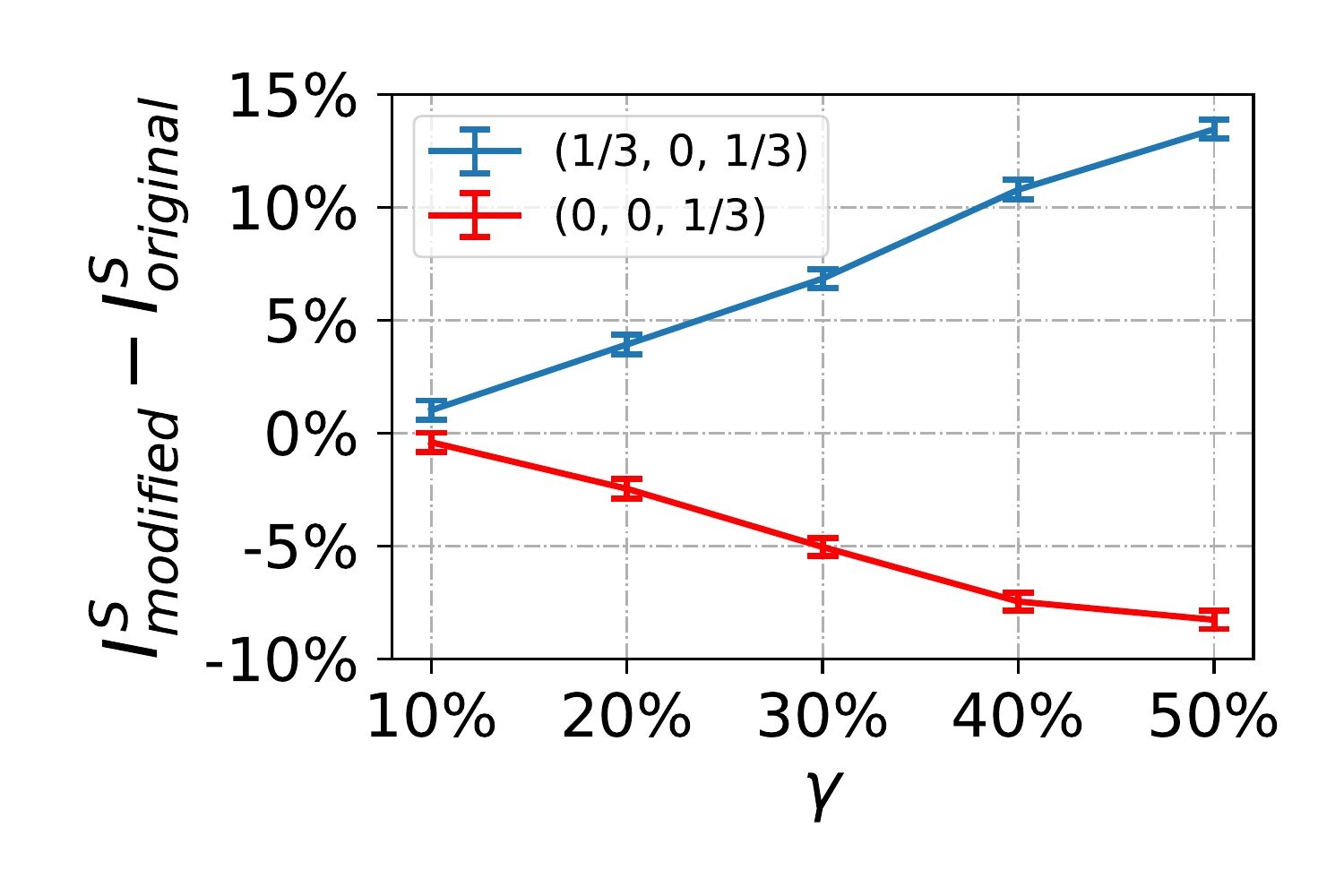} & \includegraphics[width=\FigSize]{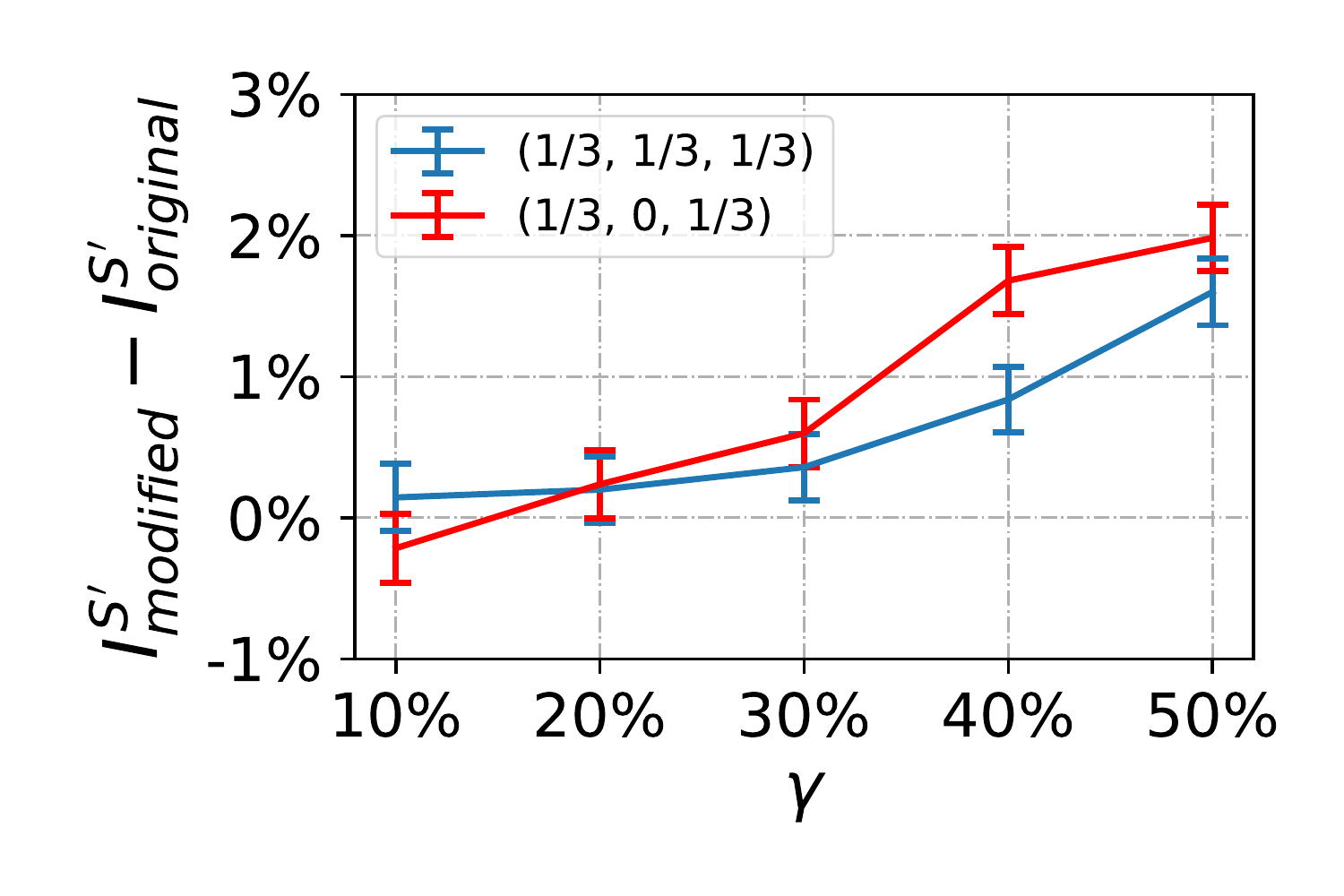} & \includegraphics[width=\FigSize]{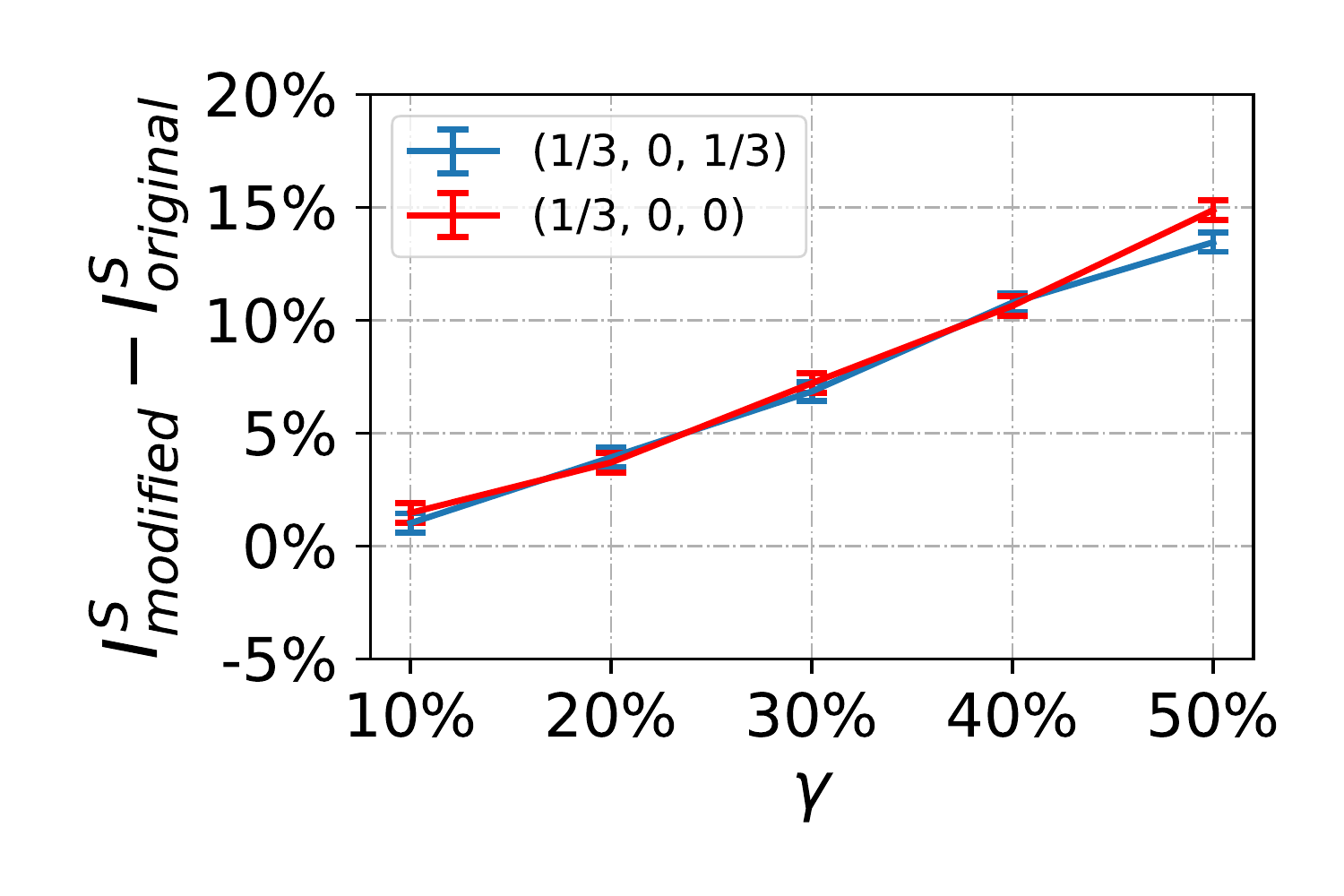} \\
\includegraphics[width=\FigSize]{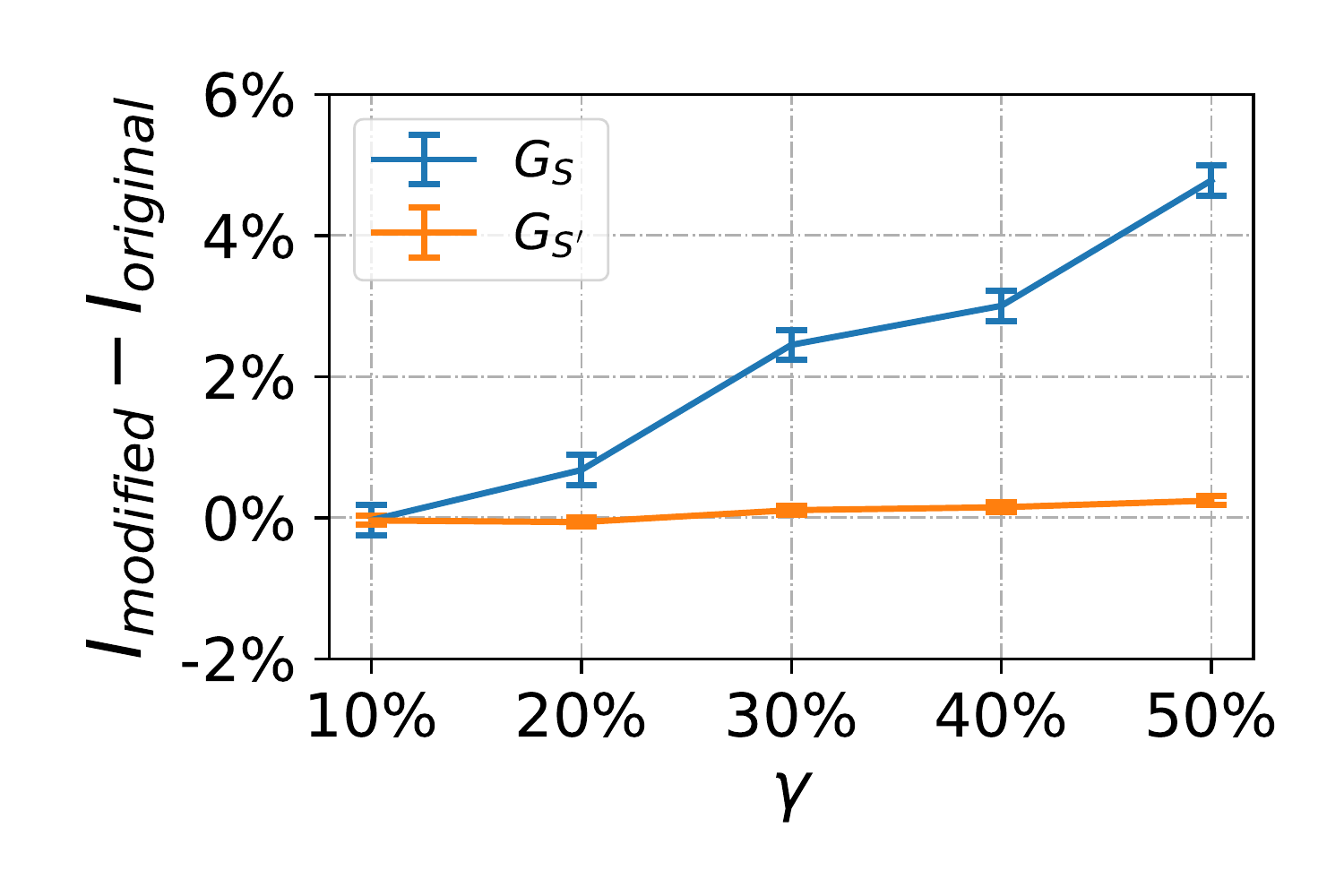} & \includegraphics[width=\FigSize]{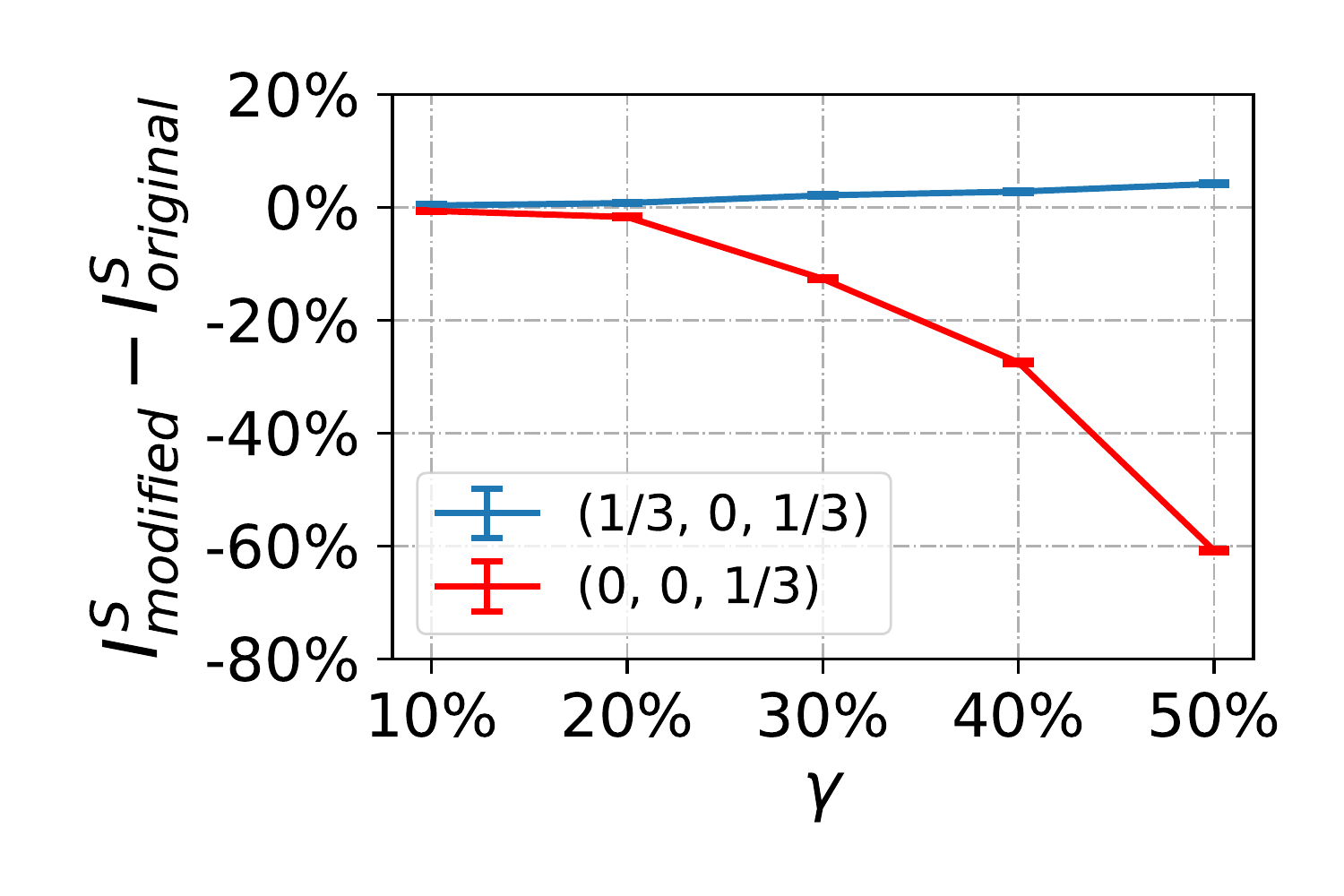} & \includegraphics[width=\FigSize]{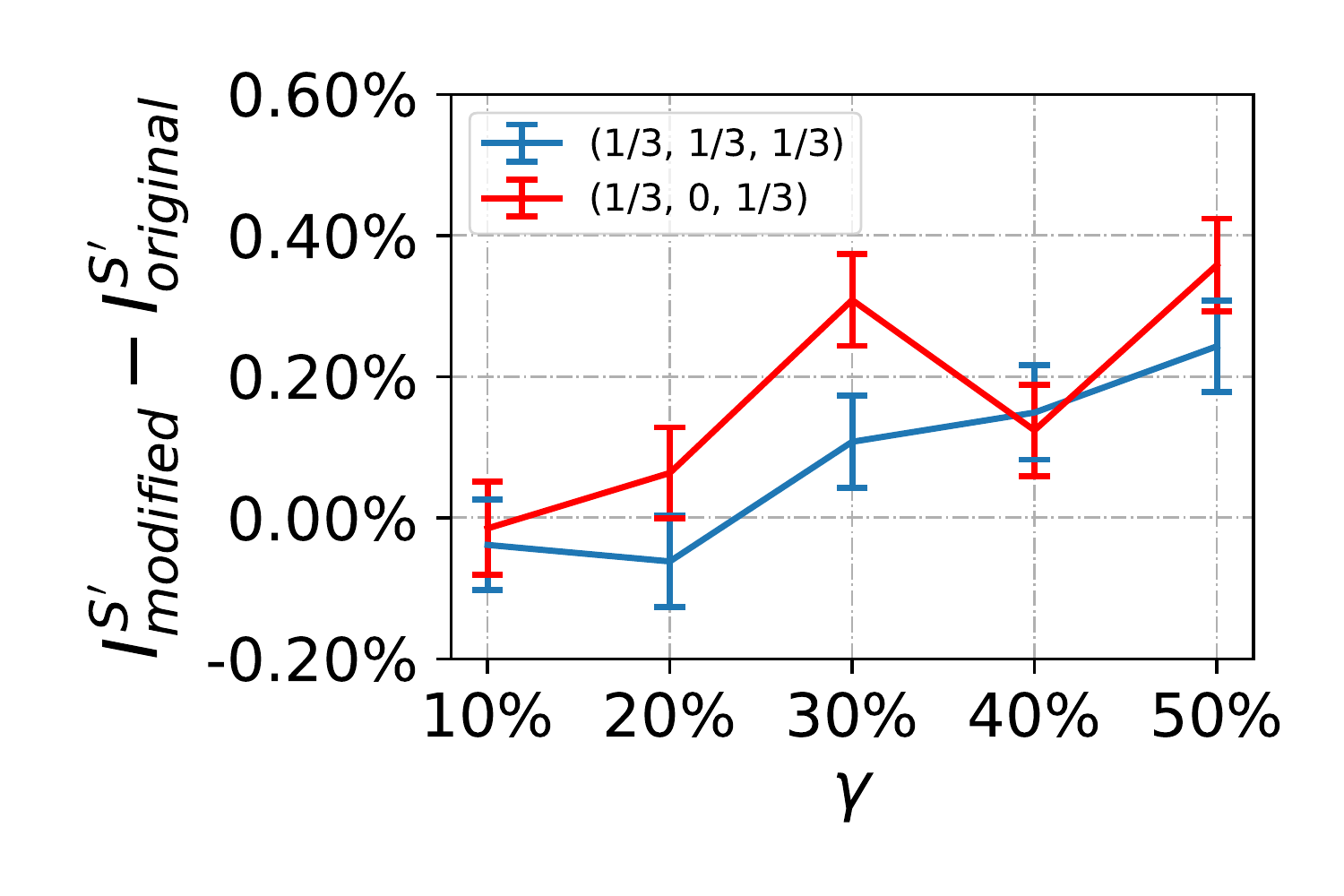} & \includegraphics[width=\FigSize]{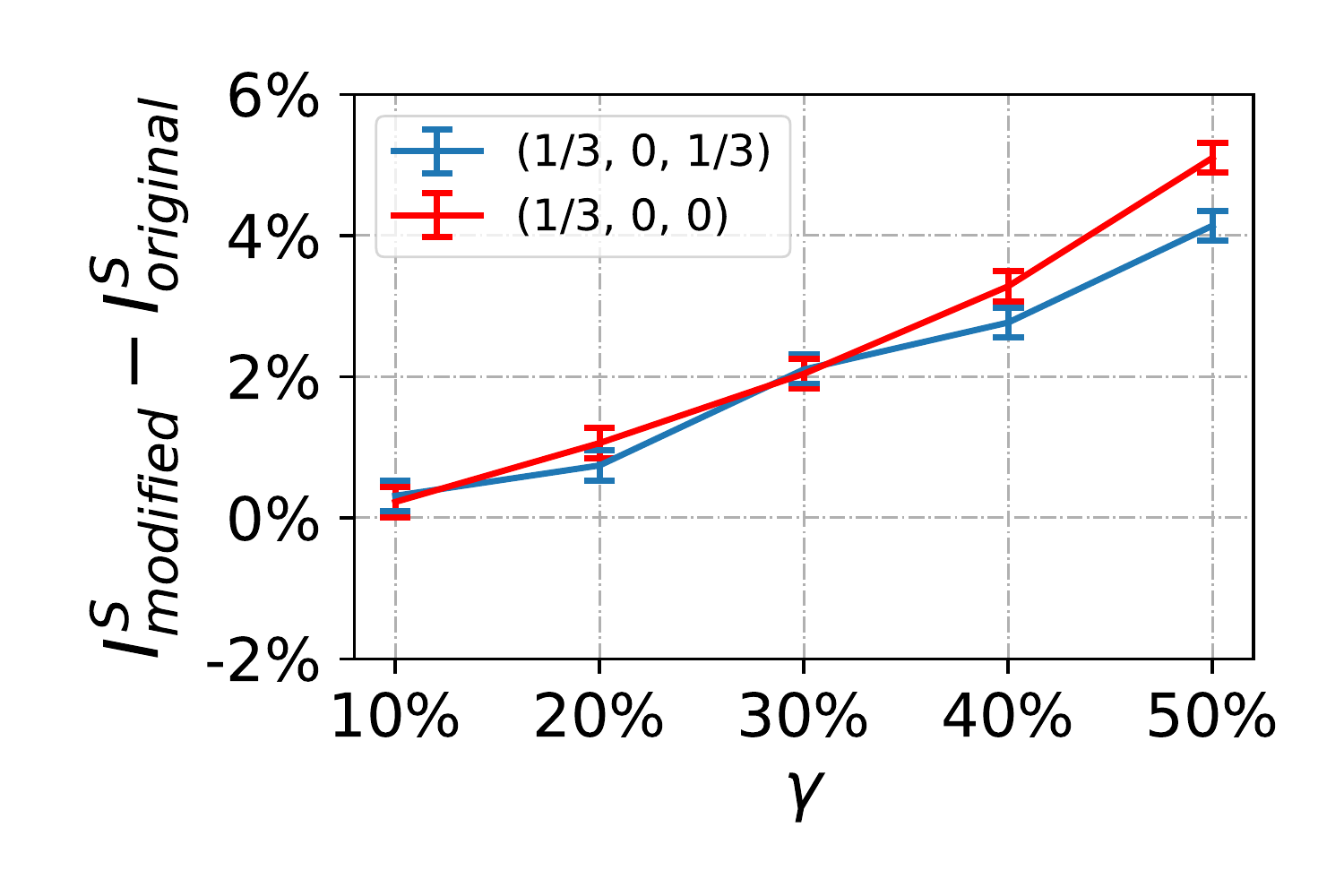} 
\end{tabular}
\caption{Experimental results on synthetic networks.
The first column shows the overall effectiveness of the threat model. 
Remaining columns show the effectiveness of: 1) maximizing $\lambda_1(\TildeAdjS)$; 2) Maximizing the eigenvector centrality of $\mathcal{S}$; 3) maximizing the normalized cut of $\mathcal{S}$. 
\textbf{Top}: BA; \textbf{Middle}: Watts-Strogatz; \textbf{Bottom}: BTER}
\label{fig:synthetic}
\end{figure*}

\section{Additional Results for Different Values of $\delta$ and $\beta$}\label{app:other_delta_beta}
In the main paper, $\delta=0.24$ and  $\beta=0.2$  for the airport and brain networks, while $\delta=0.24$ and $\beta=0.06$ for the email network.  
The ratio  $\delta/ \beta$ is $1.2$ for the former two networks, while $4$ for the latter.
In what follows we explore the effectiveness of our model in different regimes of $\delta / \beta$.
For the airport and brain networks,  we present results for $(\delta=0.5, \beta=0.1)$ and $(\delta=0.3, \beta=0.5)$.
The former (resp. latter) corresponds to the regime above (resp. below) $1.2$.
The results for the airport network are showed in Figure~\ref{fig:airport_diff_threshold}, and the results for the brain network are in Figure~\ref{fig:brain_diff_threshold}.
For the email network we present results for $(\delta=0.5, \beta=0.1)$ and $(\delta=0.3, \beta=0.5)$, also corresponds to the regime above and below the original ratio respectively. 
The results are showed in Figure~\ref{fig:email_diff_threshold}.
The conclusions are consistent with that presented in the main paper.

\section{Results for Random Walk Based Spreading Dynamics}
We simulate random walk based spreading dynamics on the original and the modified networks. 
Although \ModelName is motivated from the analysis of SIS spreading dynamics, the simulation results show that it is capable of achieving targeted diffusion when the underlying spreading dynamics is based on random walk.
Random walk  has extensive use in machine learning, data mining, security, ranking, etc.~\citep{perozzi2014deepwalk,backstrom2011supervised,tong2006fast,sun2005neighborhood,page1999pagerank}.
We focus on two variants of random walks: random walk with restart (RWR, a.k.a.~personalized PageRank) and PageRank. The former has been widely used in data mining and security applications~\citep{tong2006fast,sun2005neighborhood}. The latter is a powerful tool to measure the ``importance'' of nodes in a network~\citep{page1999pagerank}.
We run \ModelName on the Airport, Brain, and Email networks, with the same targeted subgraphs as in previous experiments.
The trade-off parameters are set to $(\alpha_1=1/3,\alpha_2=1/3,\alpha_3=1/3)$.

For the RWR dynamics, the starting node of a random walk is picked uniformly at random from the non-targeted subgraph \GSPrime. 
The restart probability $c$ is set to $0.05$ -- i.e., at each time step the RWR restarts from the starting node with probability $0.05$.
The RWR dynamics is simulated until convergence,\footnote{We are guaranteed convergence since the Markov transition matrix of the network is stochastic, irreducible, and aperiodic.} which gives us a rank vector $\bm{r} \in \R^n_+$ over the nodes for the given starting node. The sum of the sub-vector $\bm{r}[\mathcal{S}]$ (resp. $\bm{r}[\SPrime]$) is the probability that a random walk lands in the targeted subgraph $\GS$ (resp. non-targeted subgraph $\GSPrime$), which quantifies the impact on \GS (resp. \GSPrime).
 Figure~\ref{fig:RWR} shows the experimental results here.
The left column represents the landing probability on the targeted subgraph \GS. It is clear that the probability is higher when the underlying graph is modified by \ModelName (although the difference is only statistically significant on the Email network).
The right column is the landing probability on the non-targeted subgraph \GSPrime. The probability does not increase, which is desired as we would like to limit the impact on \GSPrime.

For the PageRank dynamics, the starting node is picked from the node set $\mathcal{V}$ uniformly at random. 
The restart probability $c$ is set to $0.1$ -- i.e., at each time step the PageRank dynamics restarts with probability $0.1$ from a node (not necessarily the starting node) picked from $\mathcal{V}$ uniformly at random. 
When the simulation is finished the PageRank gives a vector $\bm{r} \in \R^n_+$ indicating how ``important'' each node is.
Intuitively, $\bm{r}$ specifies a ranking of the nodes in $\mathcal{V}$ -- i.e., a node $i \in \mathcal{V}$ is ranked higher when $r[i]$ is larger.
We use the sum of the sub-vector $\bm{r}[\mathcal{S}]$ (resp. $\bm{r}[\SPrime]$) to quantify the impact on \GS (resp. \GSPrime).
Other experimental setup is the same as the setup for the RWR dynamics. Figure~\ref{fig:PR} shows these results.
The left column indicates the ranking of the nodes in \GS. It is clear that the ranking is boosted and the increase is statistically significant.
The right column shows that the ranking of the nodes in \GSPrime is not increased, as desired.

\begin{figure}[h]
\def\FigSize{3in}
\centering
\setlength{\tabcolsep}{0.1pt}
\begin{tabular}{c}
\includegraphics[width=\FigSize]{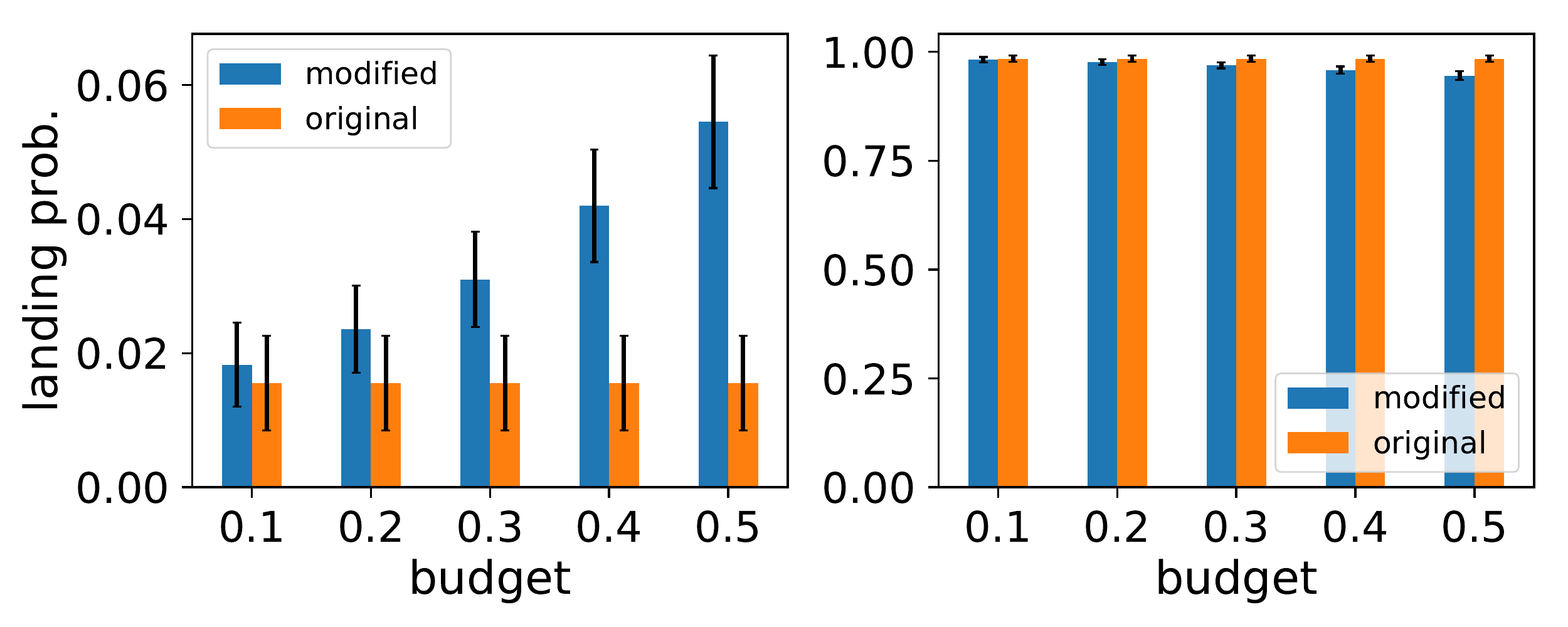} \\
\includegraphics[width=\FigSize]{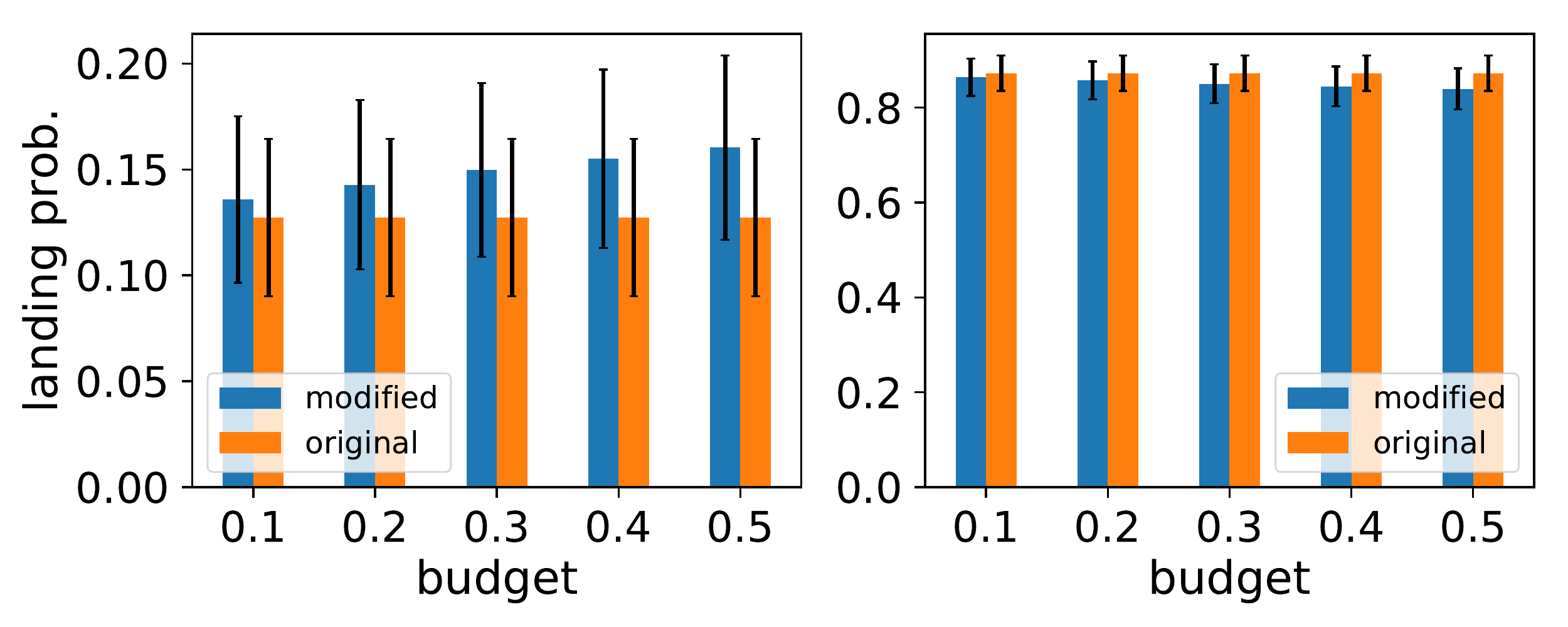} \\
\includegraphics[width=\FigSize]{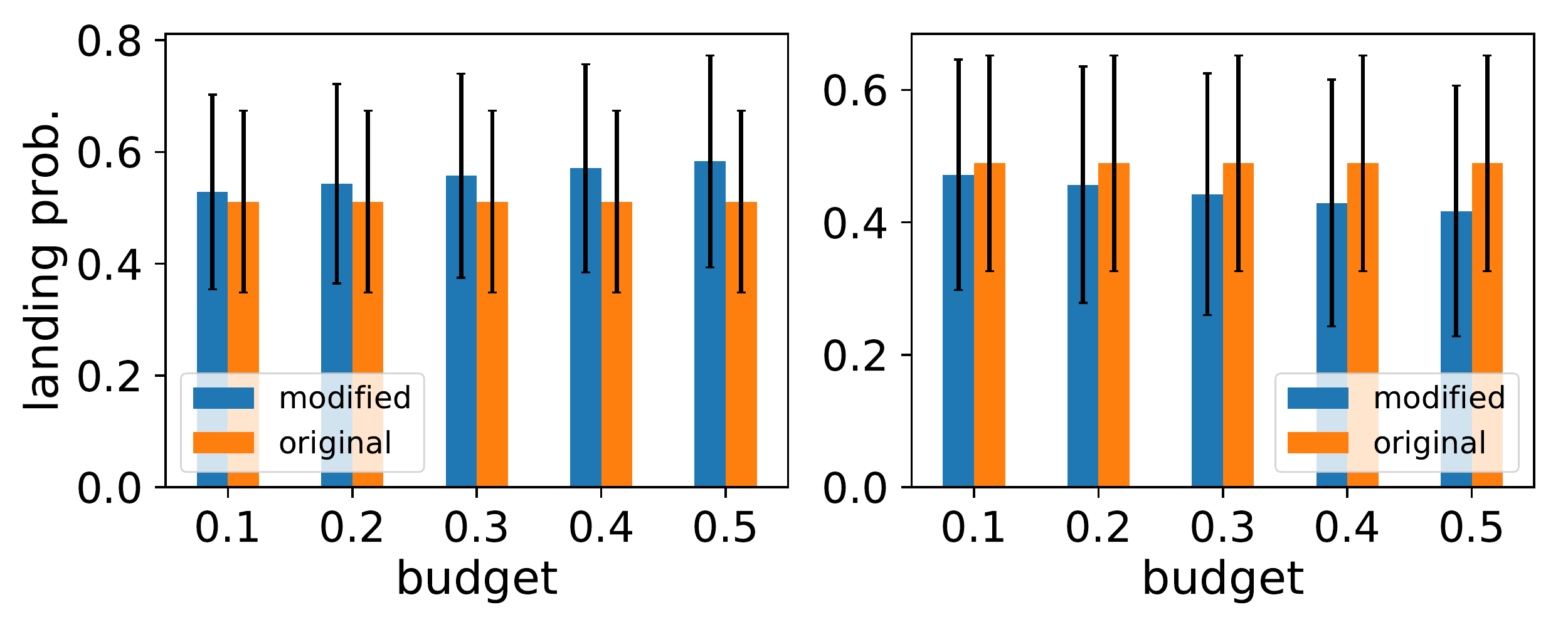}
\end{tabular}
\caption{Landing probability on (\textbf{left}) the targeted subgraph \GS and (\textbf{right}) the non-targeted subgraph \GSPrime. Modifications by \ModelName increase the probabilities of landing in the targeted subgraph \GS (left column) and do not increase the probabilities of landing in the non-targeted subgraph \GSPrime (right column), as desired.
 \textbf{Top}: Email; \textbf{Middle}: Brain; \textbf{Bottom}: Airport.}
\label{fig:RWR}
\end{figure}

\begin{figure}[h]
\def\FigSize{3in}
\centering
\setlength{\tabcolsep}{0.1pt}
\begin{tabular}{c}
\includegraphics[width=\FigSize]{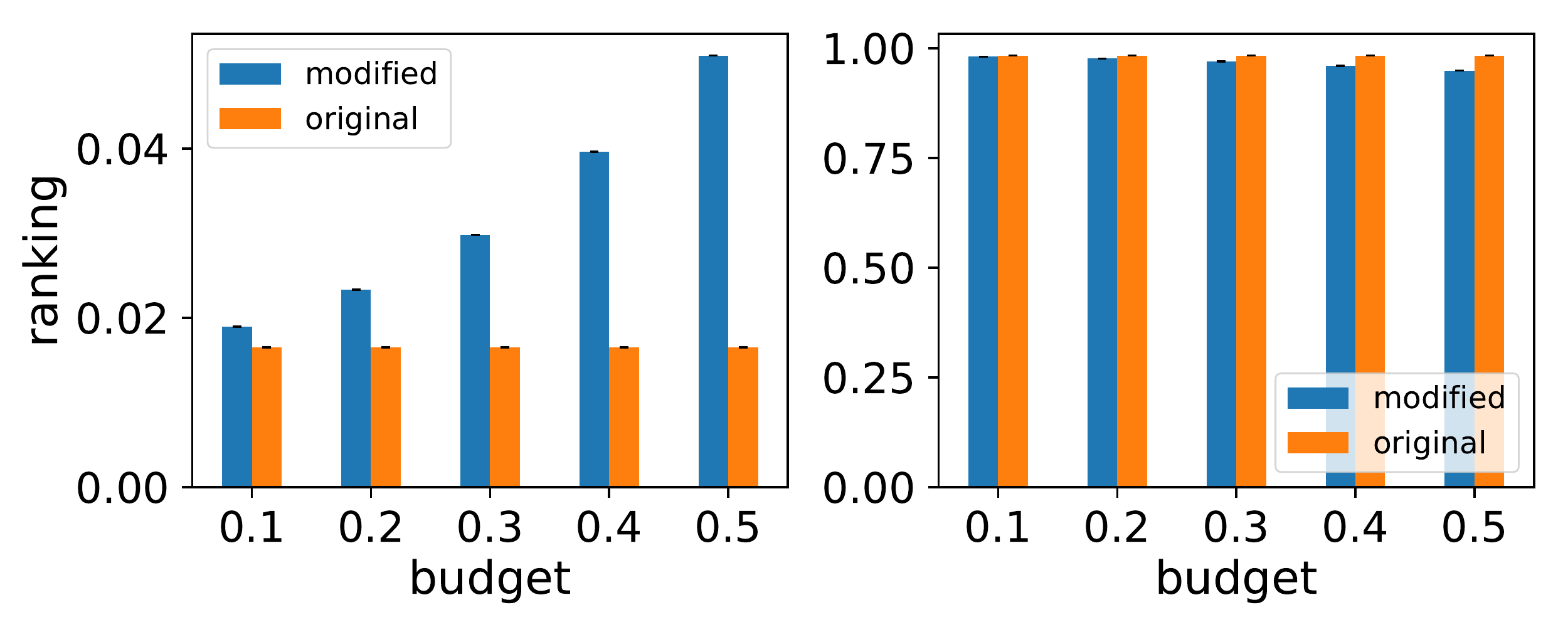} \\
\includegraphics[width=\FigSize]{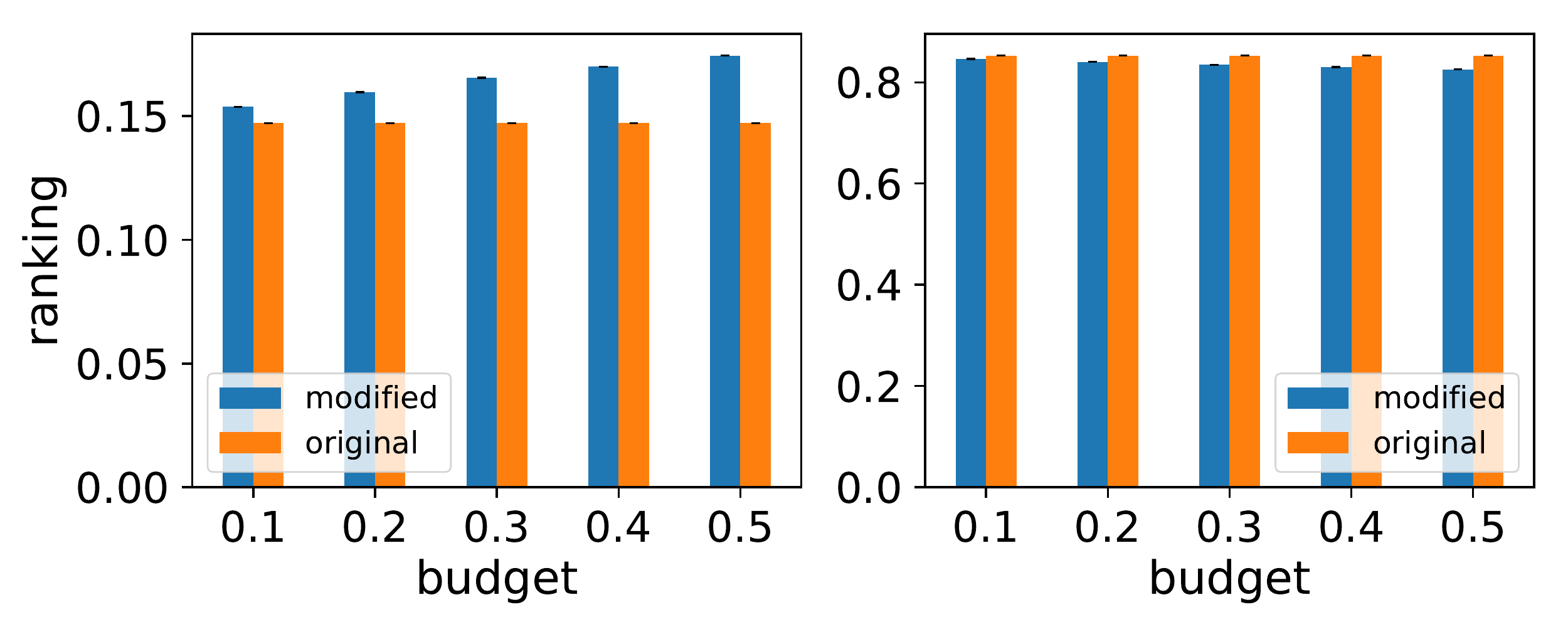} \\
\includegraphics[width=\FigSize]{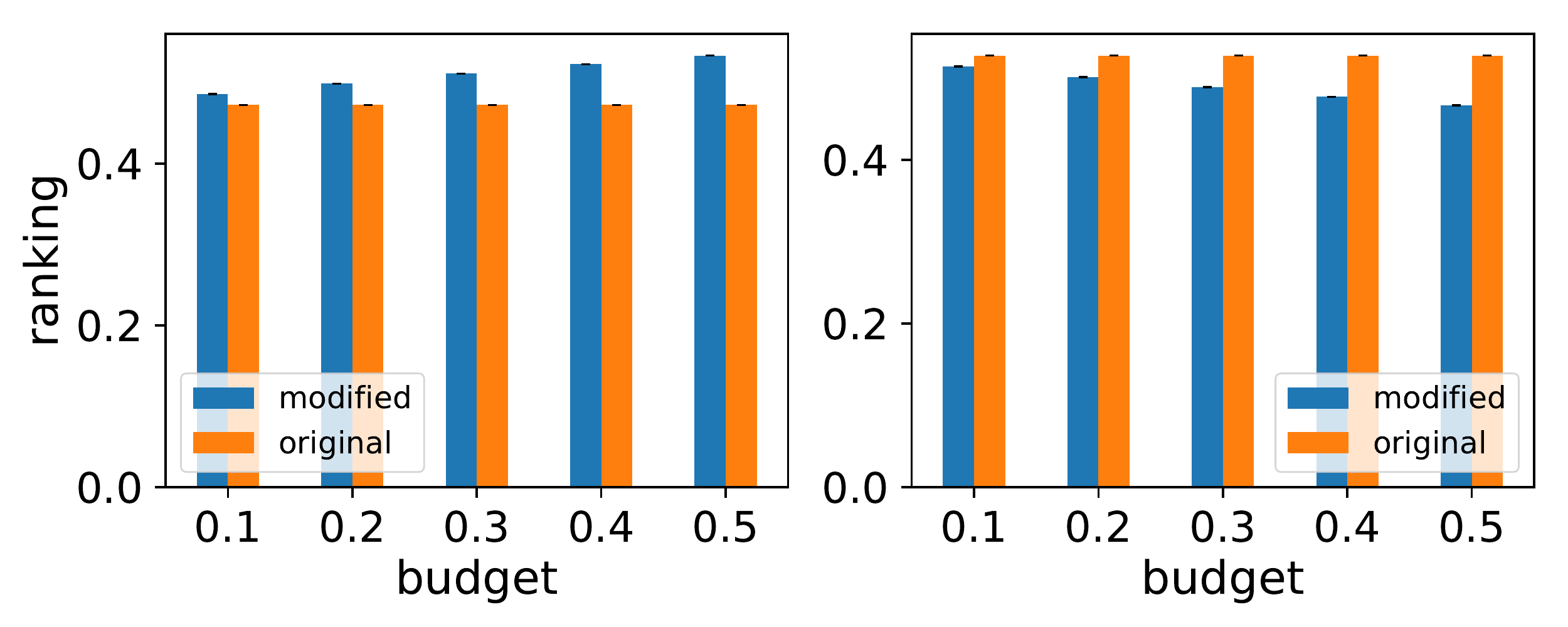}
\end{tabular}
\caption{Ranking of the nodes in (\textbf{left}) the targeted subgraph \GS and (\textbf{right}) the non-targeted subgraph \GSPrime. Modifications by \ModelName increase the probabilities of landing in the targeted subgraph \GS (left column) and do not increase the probabilities of landing in the non-targeted subgraph \GSPrime (rightt column), as desired.
 \textbf{Top}: Email; \textbf{Middle}: Brain; \textbf{Bottom}: Airport.}
\label{fig:PR}
\end{figure}

\begin{figure*}[h]
\def\FigSize{1in}
\centering
\setlength{\tabcolsep}{0.1pt}
\begin{tabular}{cccc}
\includegraphics[width=\FigSize]{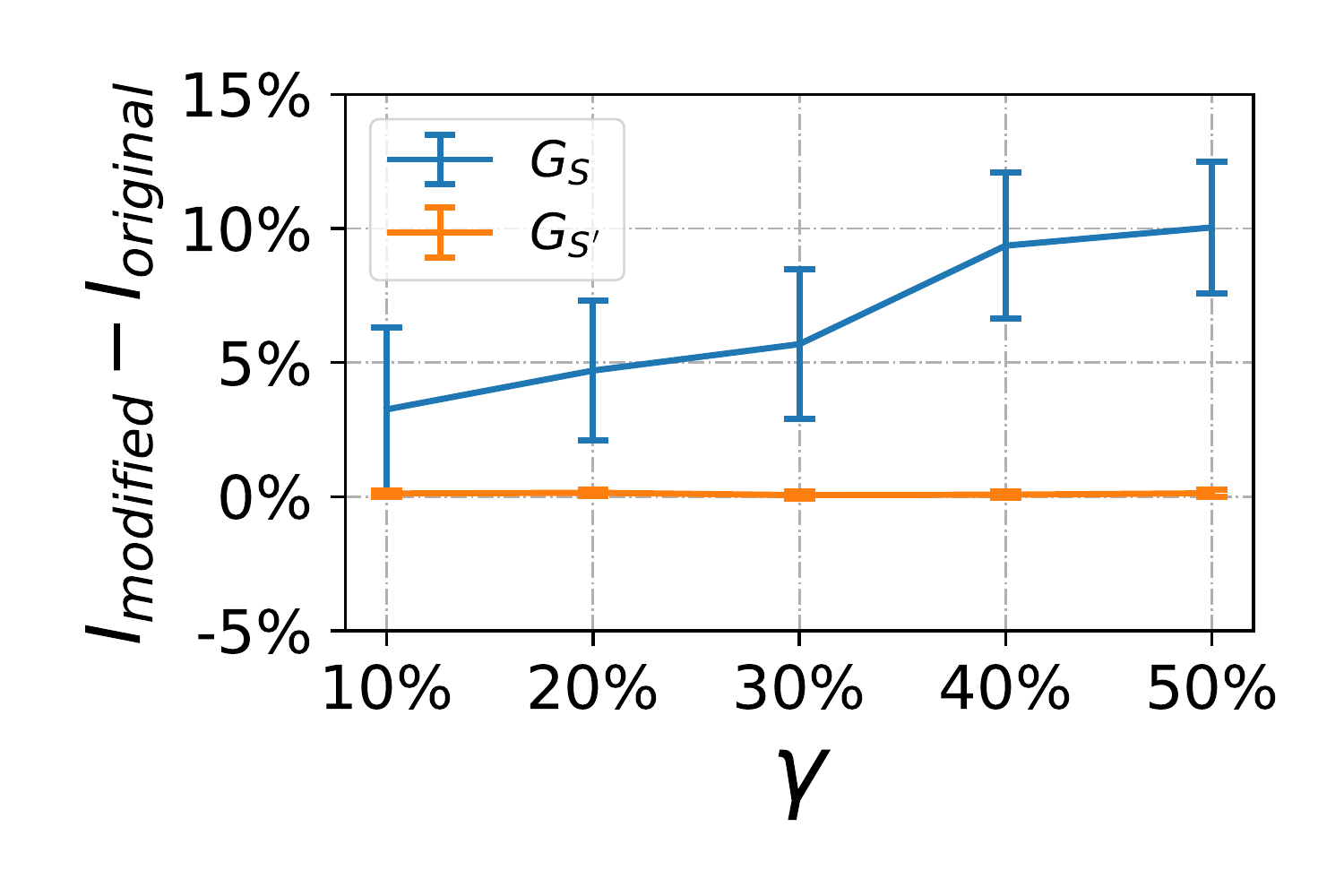} & \includegraphics[width=\FigSize]{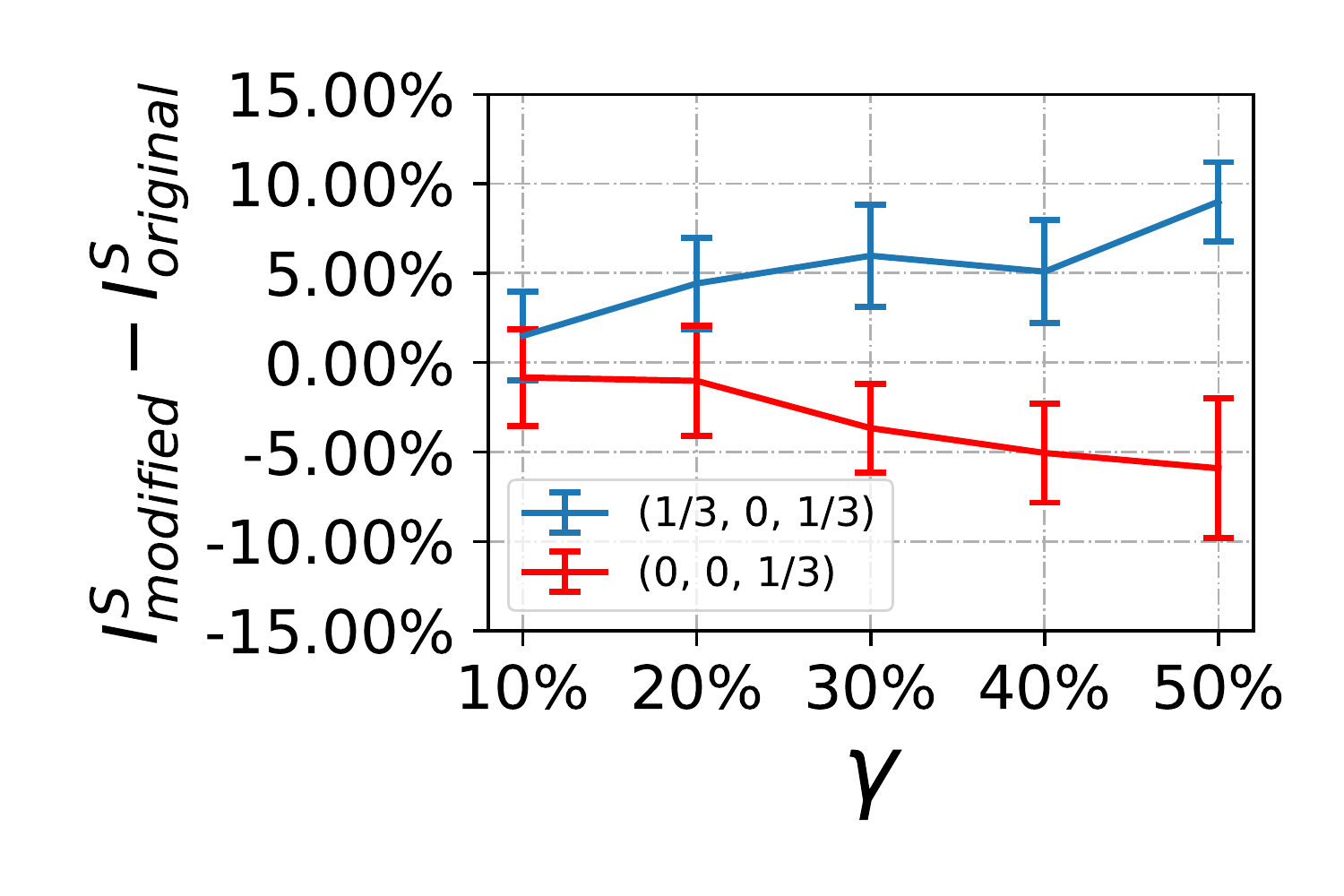} & \includegraphics[width=\FigSize]{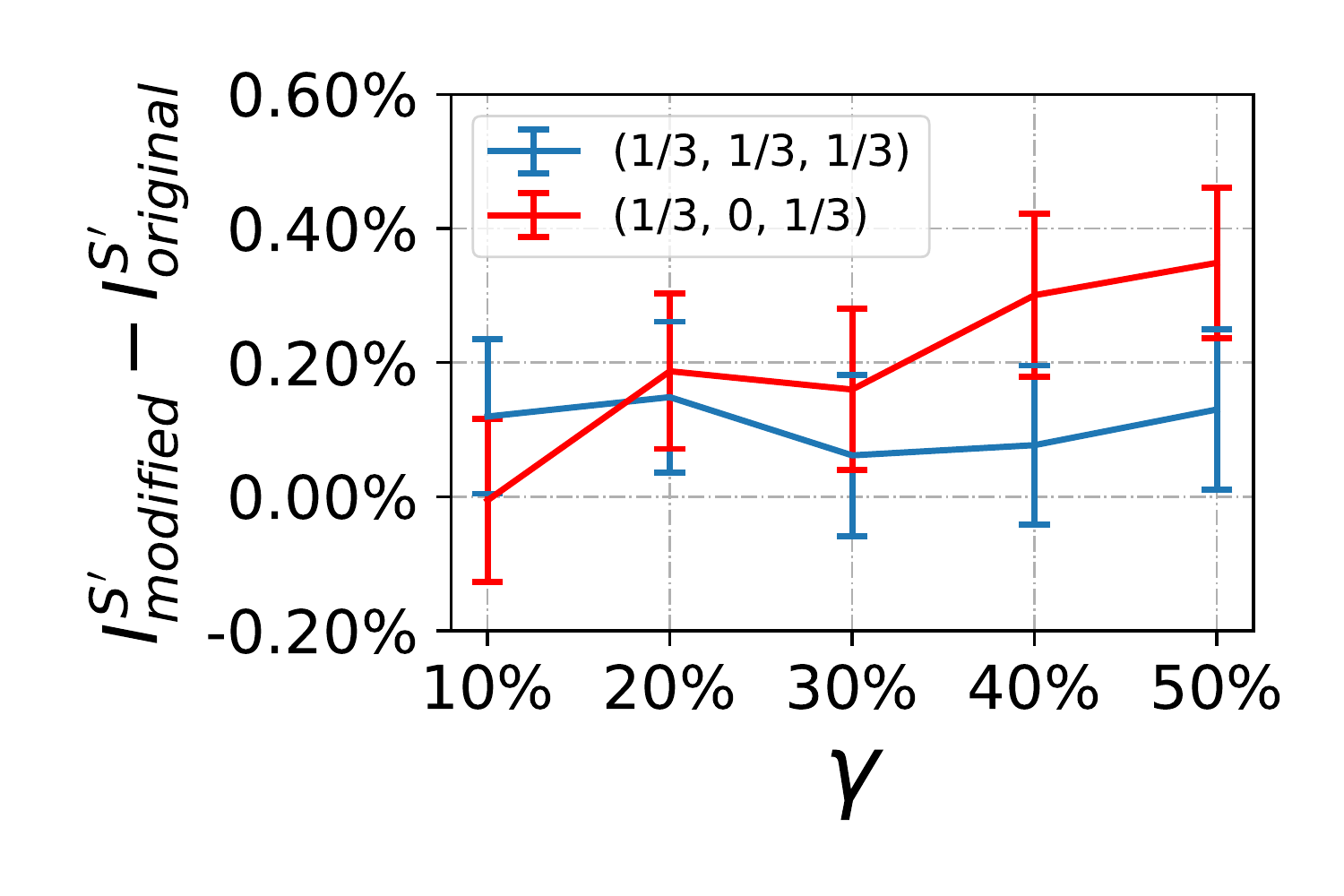} & \includegraphics[width=\FigSize]{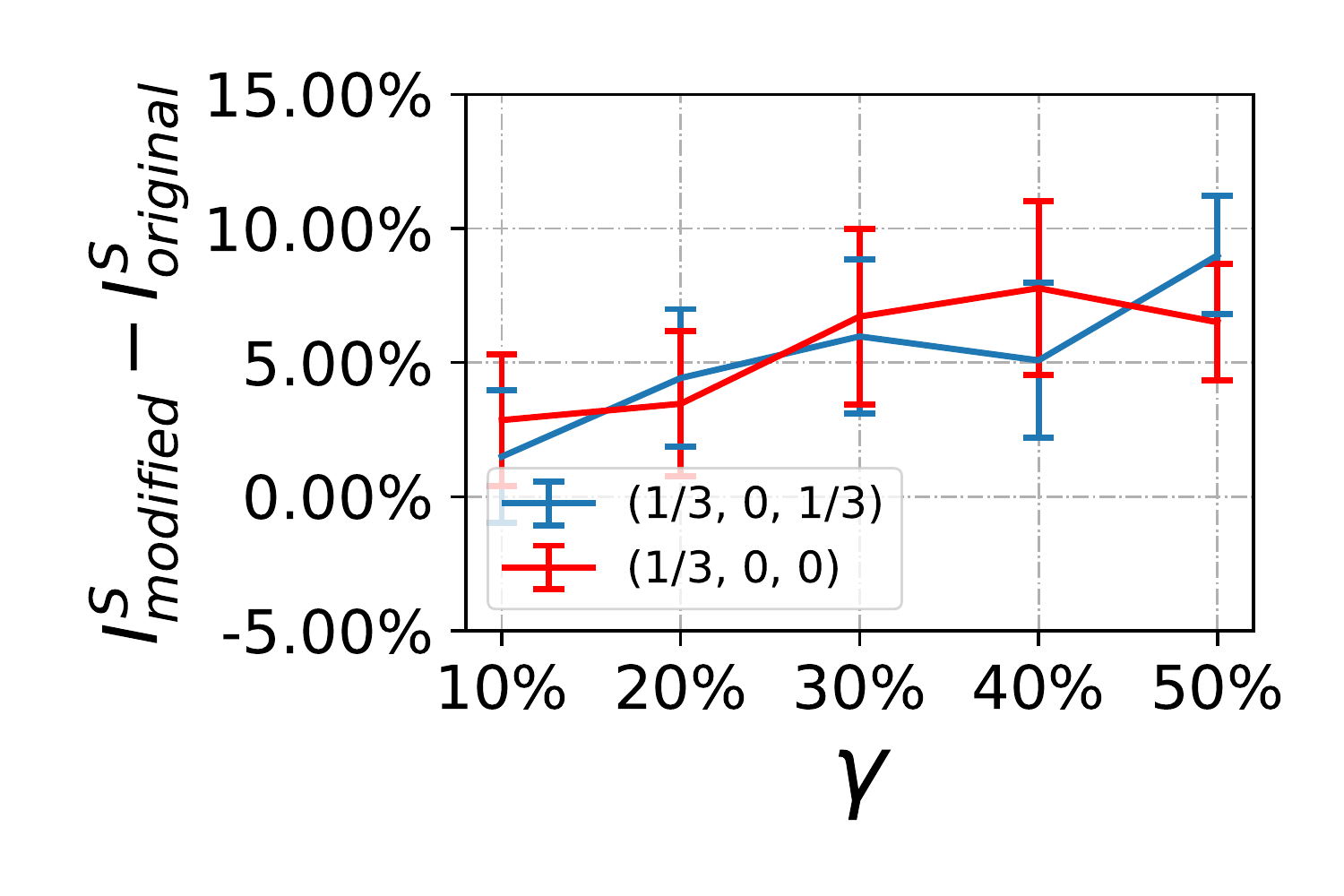} \\
\includegraphics[width=\FigSize]{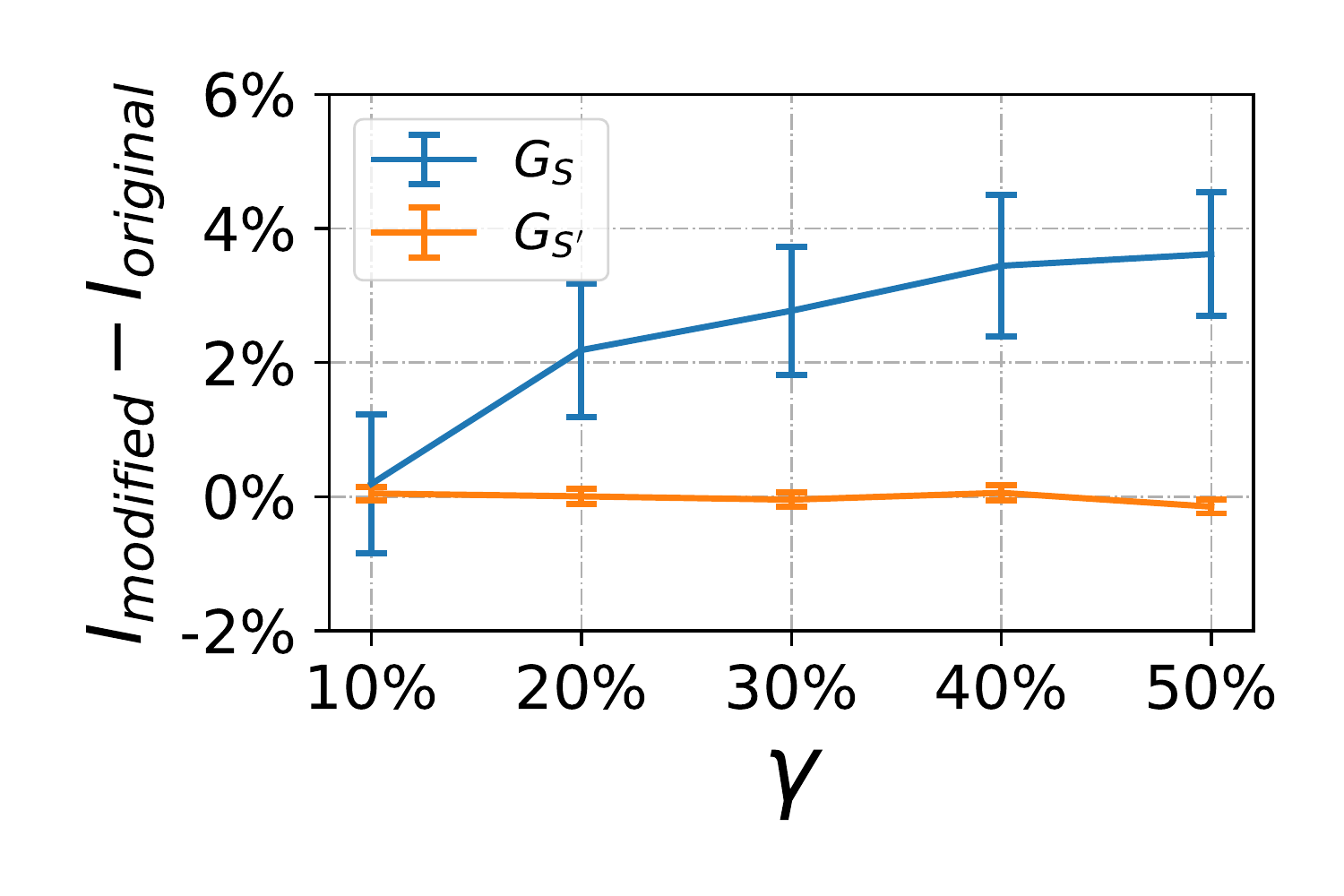} & \includegraphics[width=\FigSize]{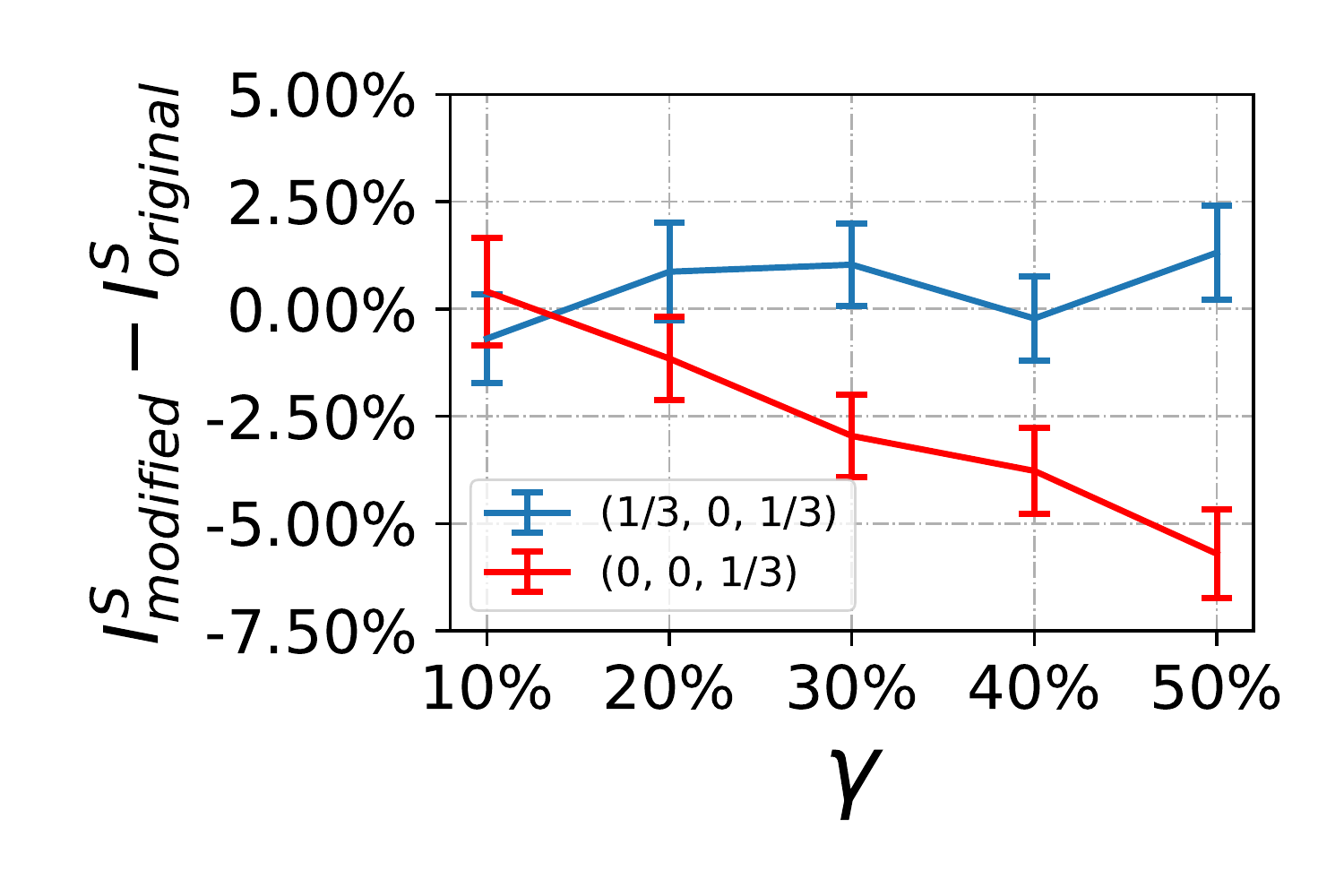} & \includegraphics[width=\FigSize]{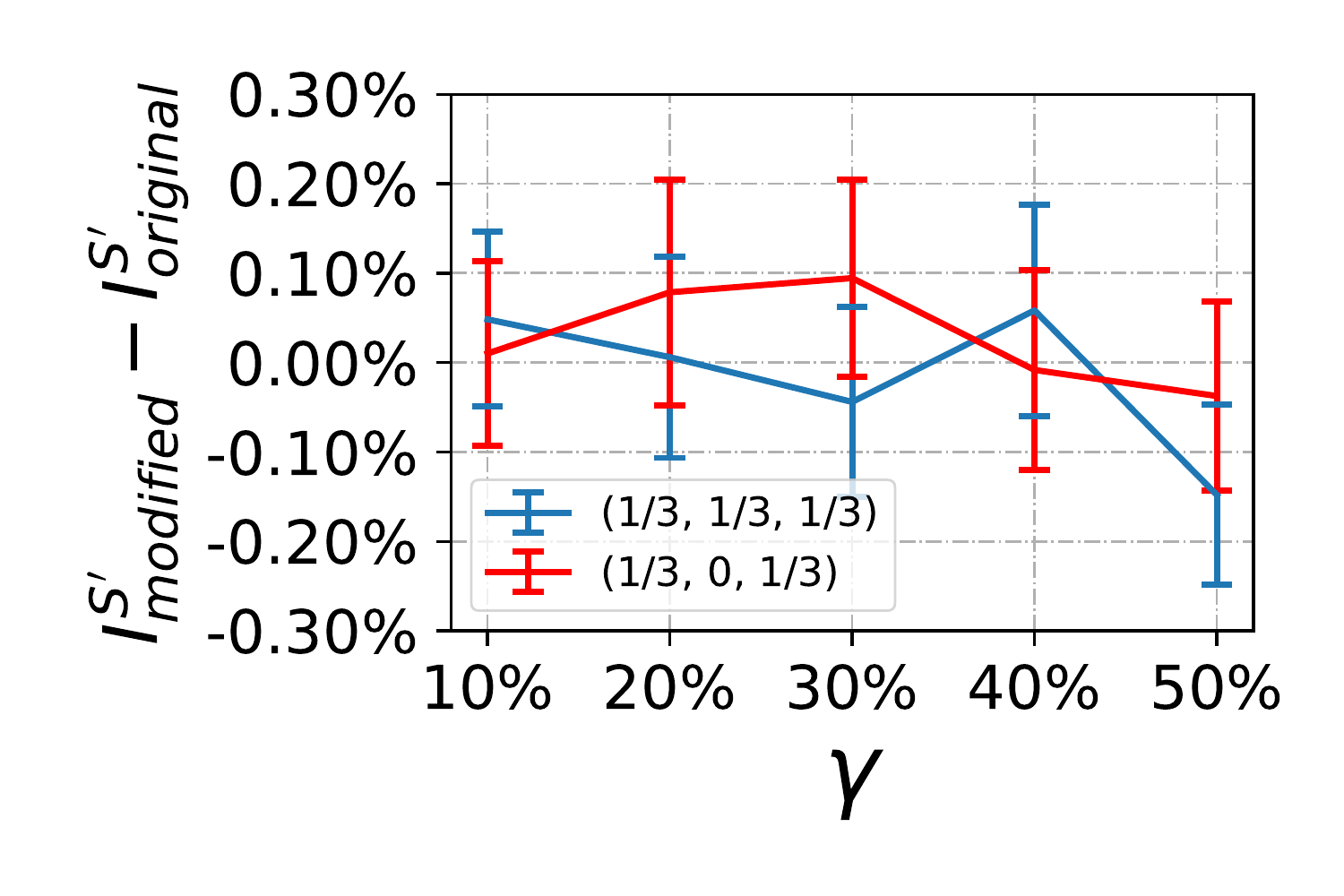} & \includegraphics[width=\FigSize]{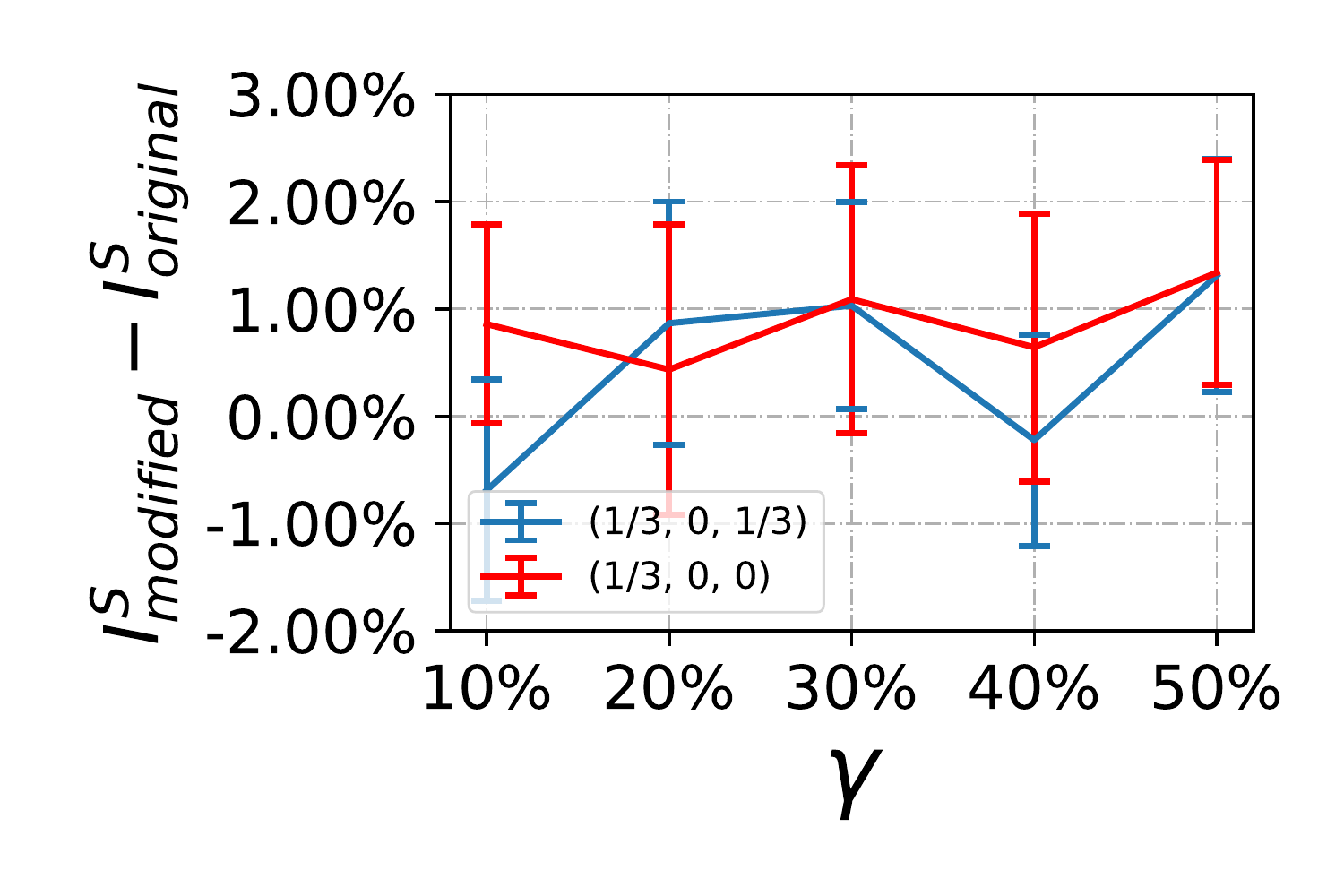} \\
\end{tabular}
\caption{Experimental results for different $\delta$ and $\beta$ values on the airport network. \textbf{Top}: $\delta=0.5, \beta=0.1$; \textbf{Bottom}: $\delta=0.3, \beta=0.5$.}
\label{fig:airport_diff_threshold}
\end{figure*}

\begin{figure*}[h]
\def\FigSize{1in}
\centering
\setlength{\tabcolsep}{0.1pt}
\begin{tabular}{cccc}
\includegraphics[width=\FigSize]{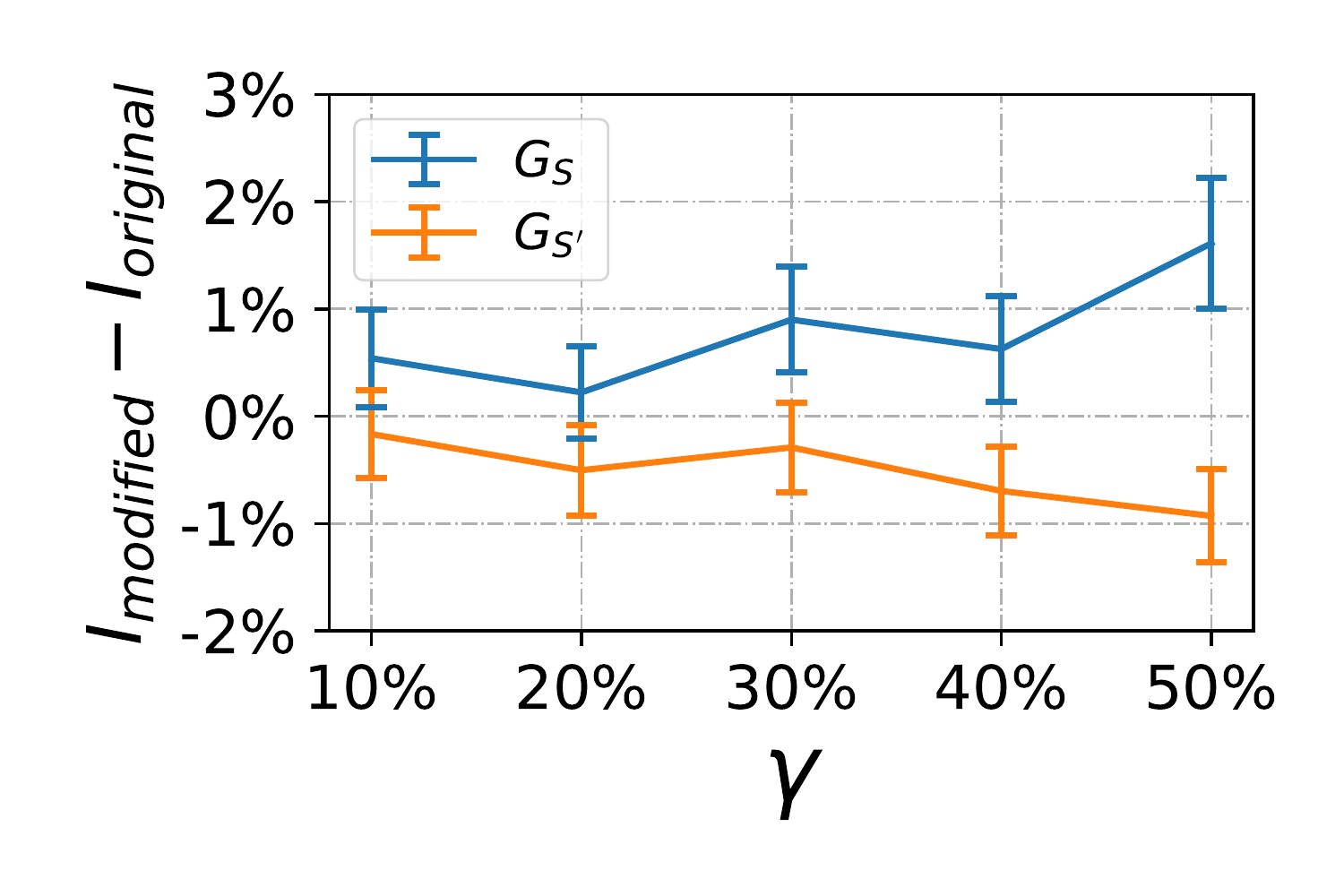} & \includegraphics[width=\FigSize]{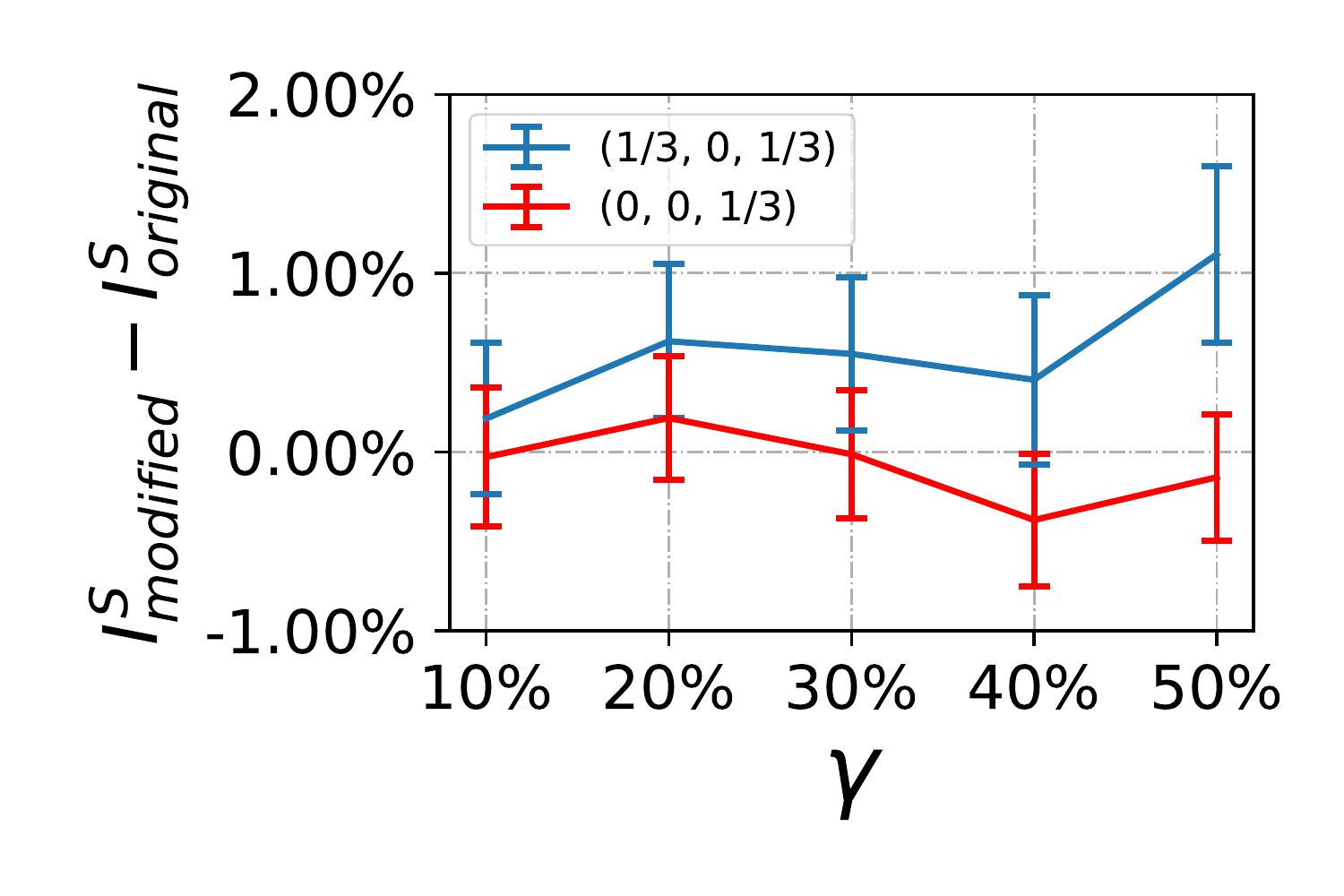} & \includegraphics[width=\FigSize]{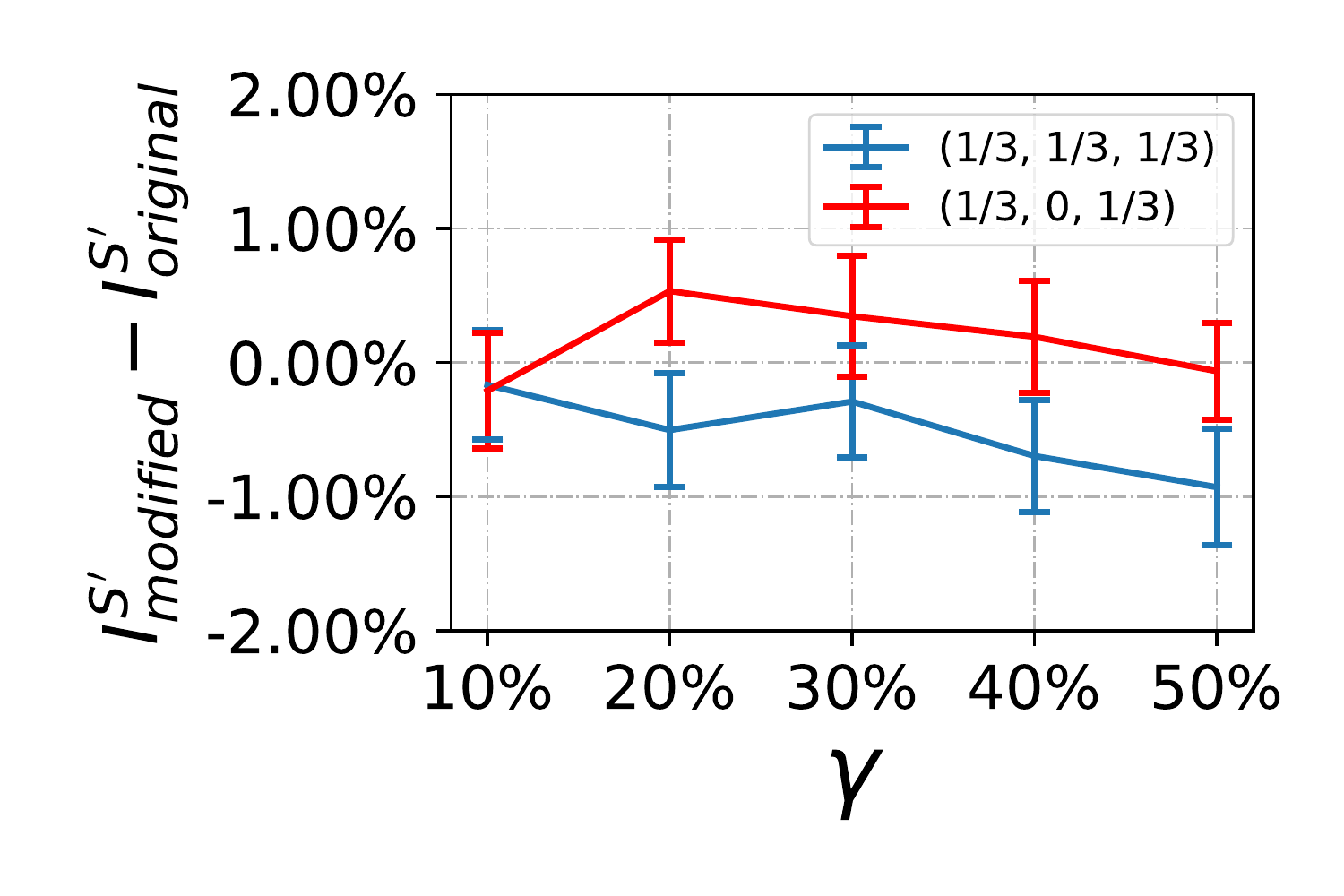} & \includegraphics[width=\FigSize]{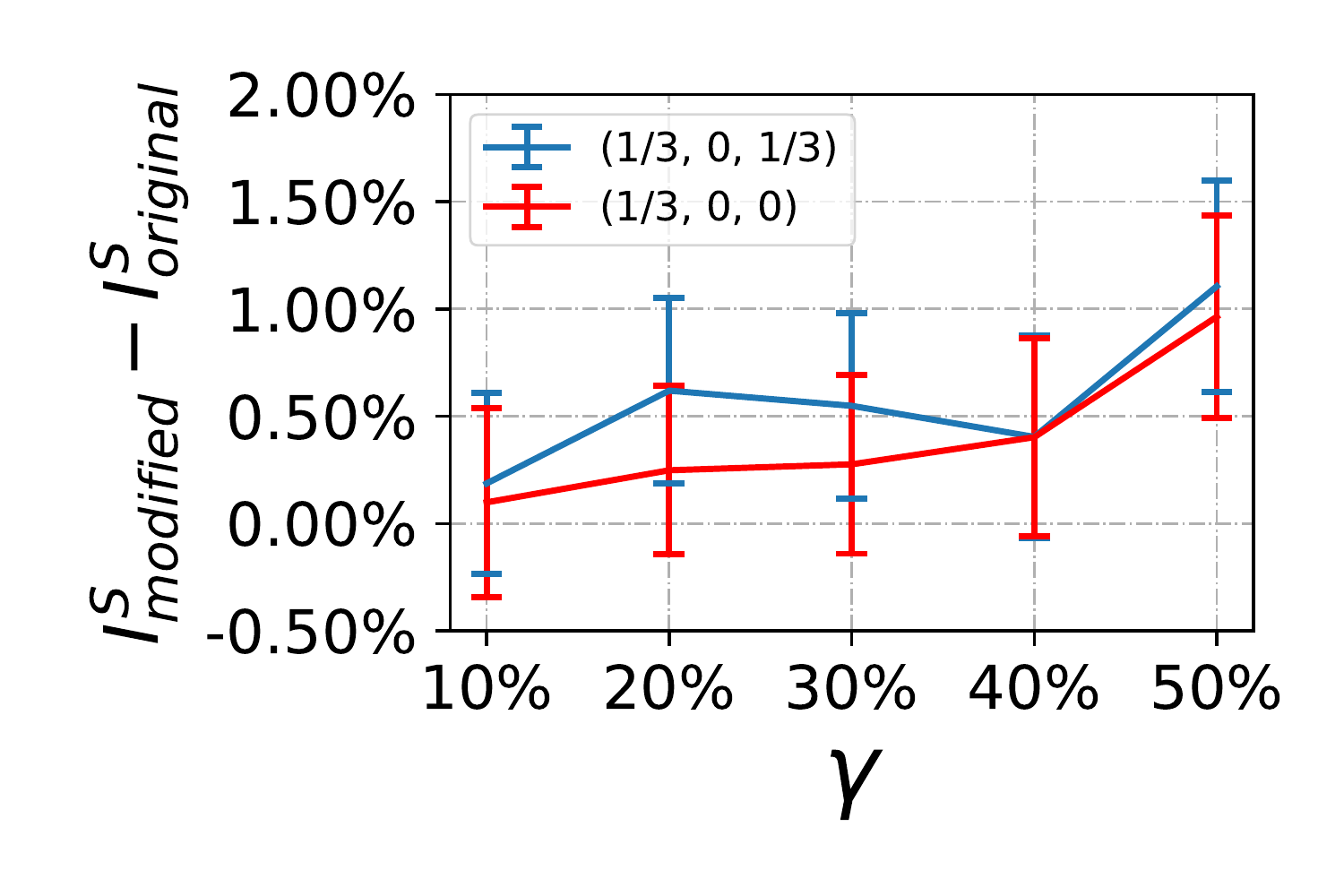} \\
\includegraphics[width=\FigSize]{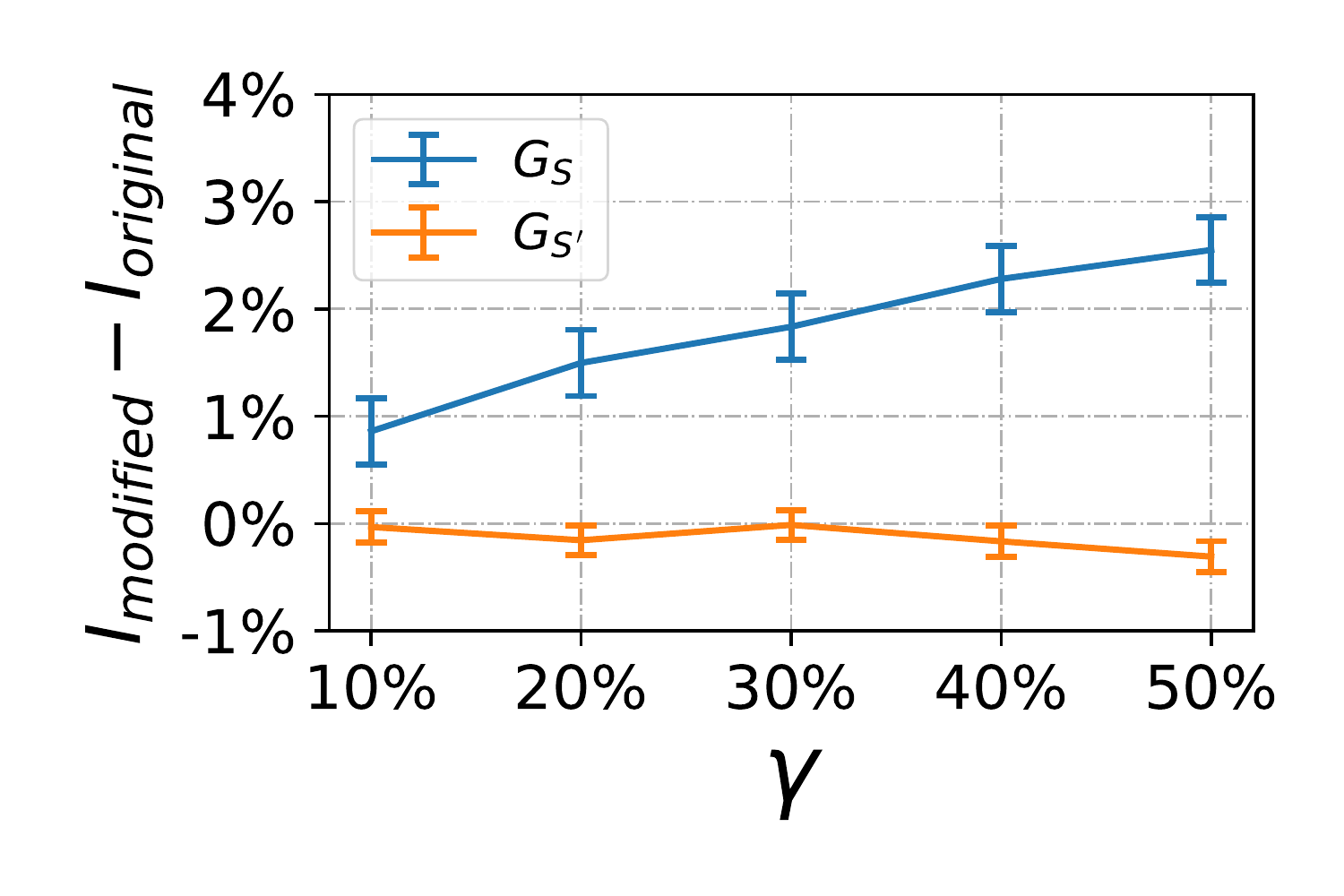} & \includegraphics[width=\FigSize]{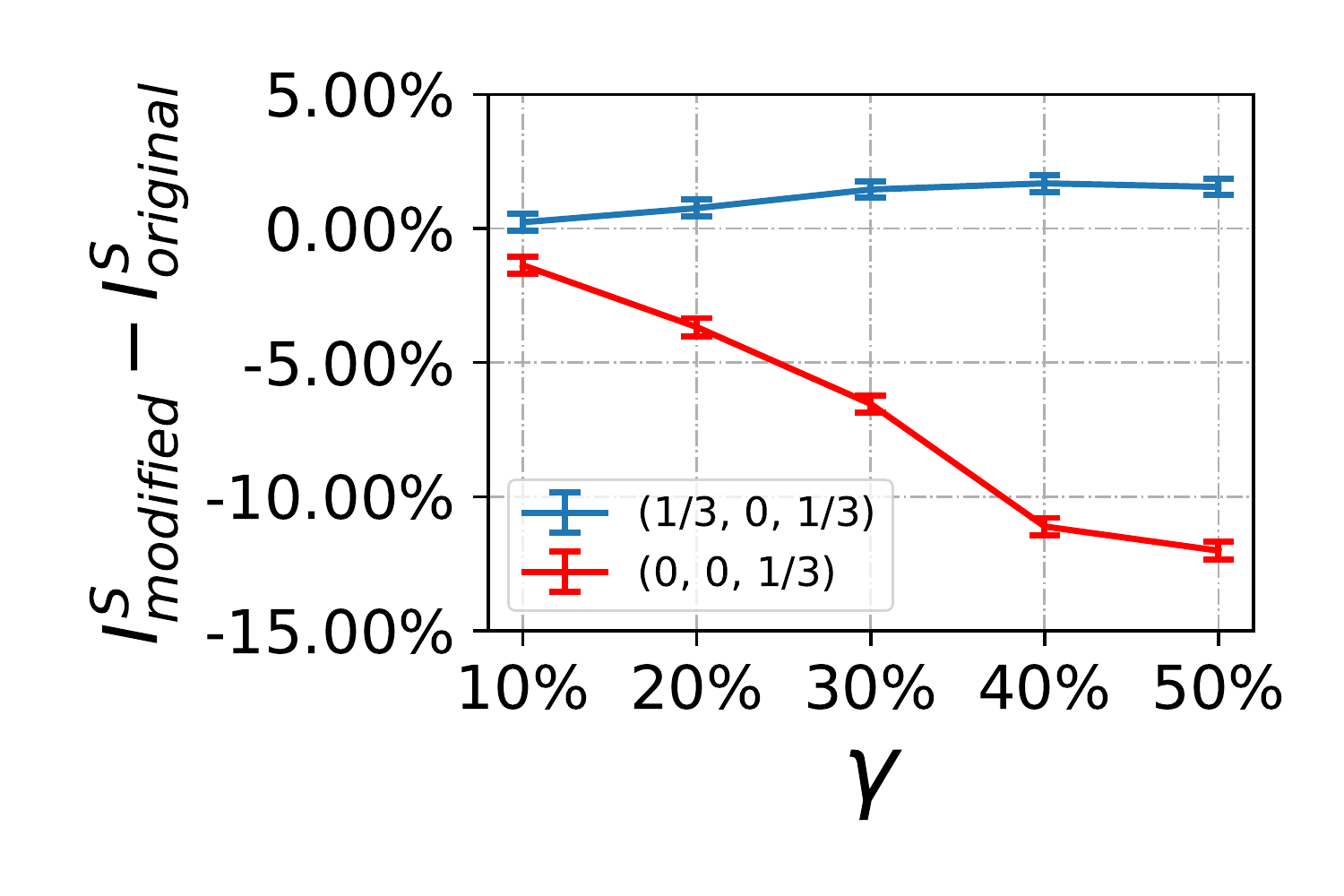} & \includegraphics[width=\FigSize]{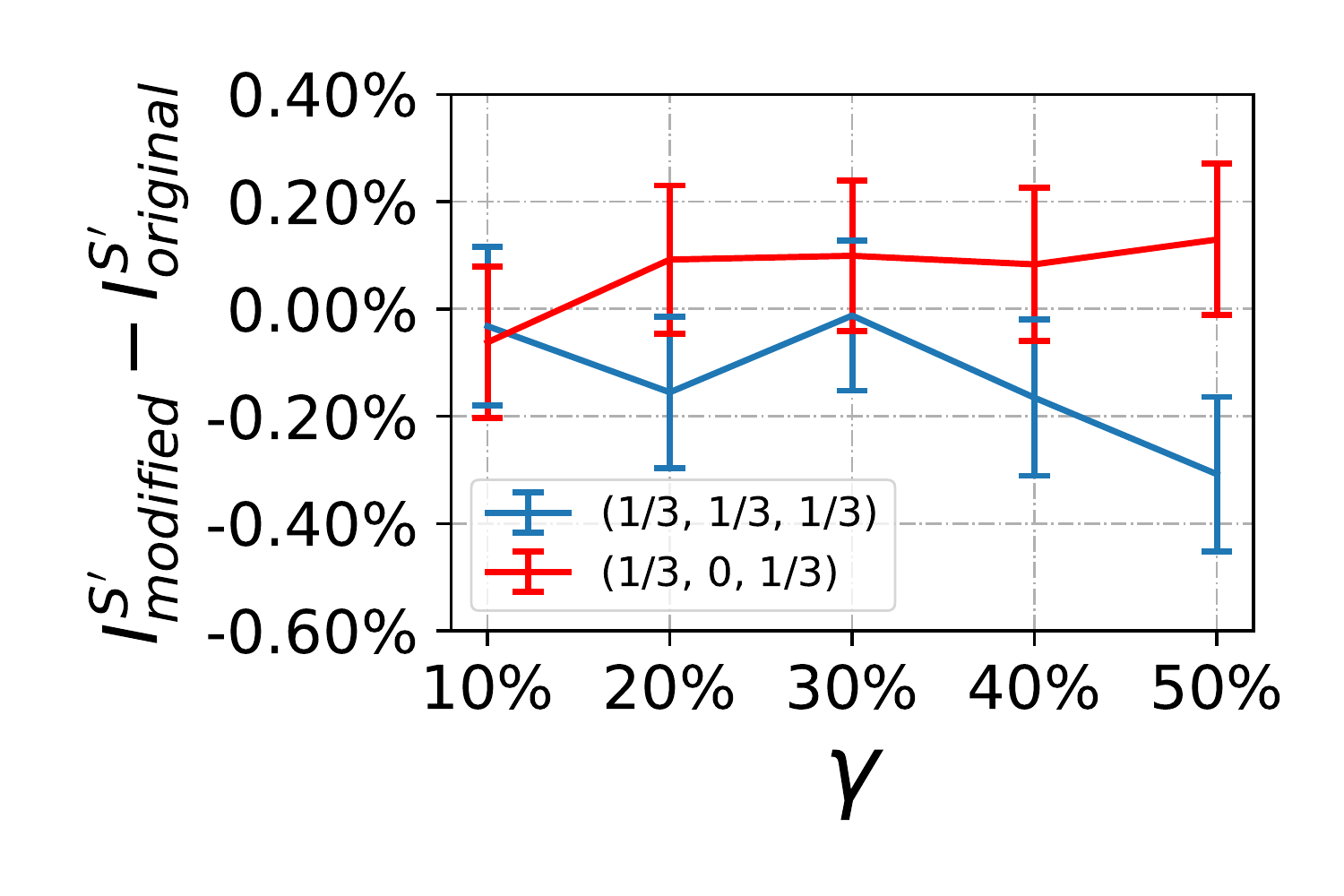} & \includegraphics[width=\FigSize]{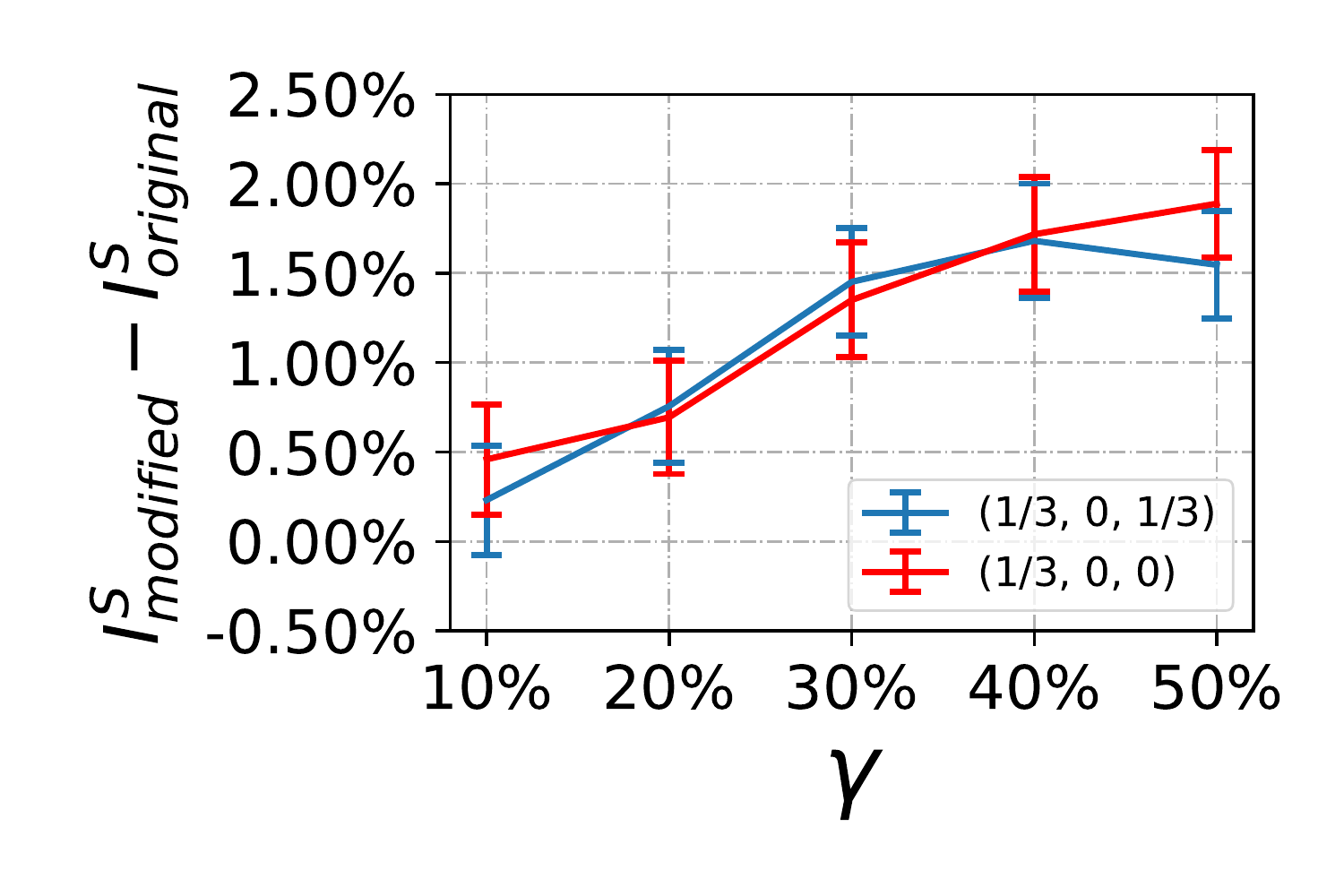} \\
\end{tabular}
\caption{Experimental results for different $\delta$ and $\beta$ values on the brain network. \textbf{Top}: $\delta=0.5, \beta=0.1$; \textbf{Bottom}: $\delta=0.3, \beta=0.5$.}
\label{fig:brain_diff_threshold}
\end{figure*}

\begin{figure*}[h]
\def\FigSize{1in}
\centering
\setlength{\tabcolsep}{0.1pt}
\begin{tabular}{cccc}
\includegraphics[width=\FigSize]{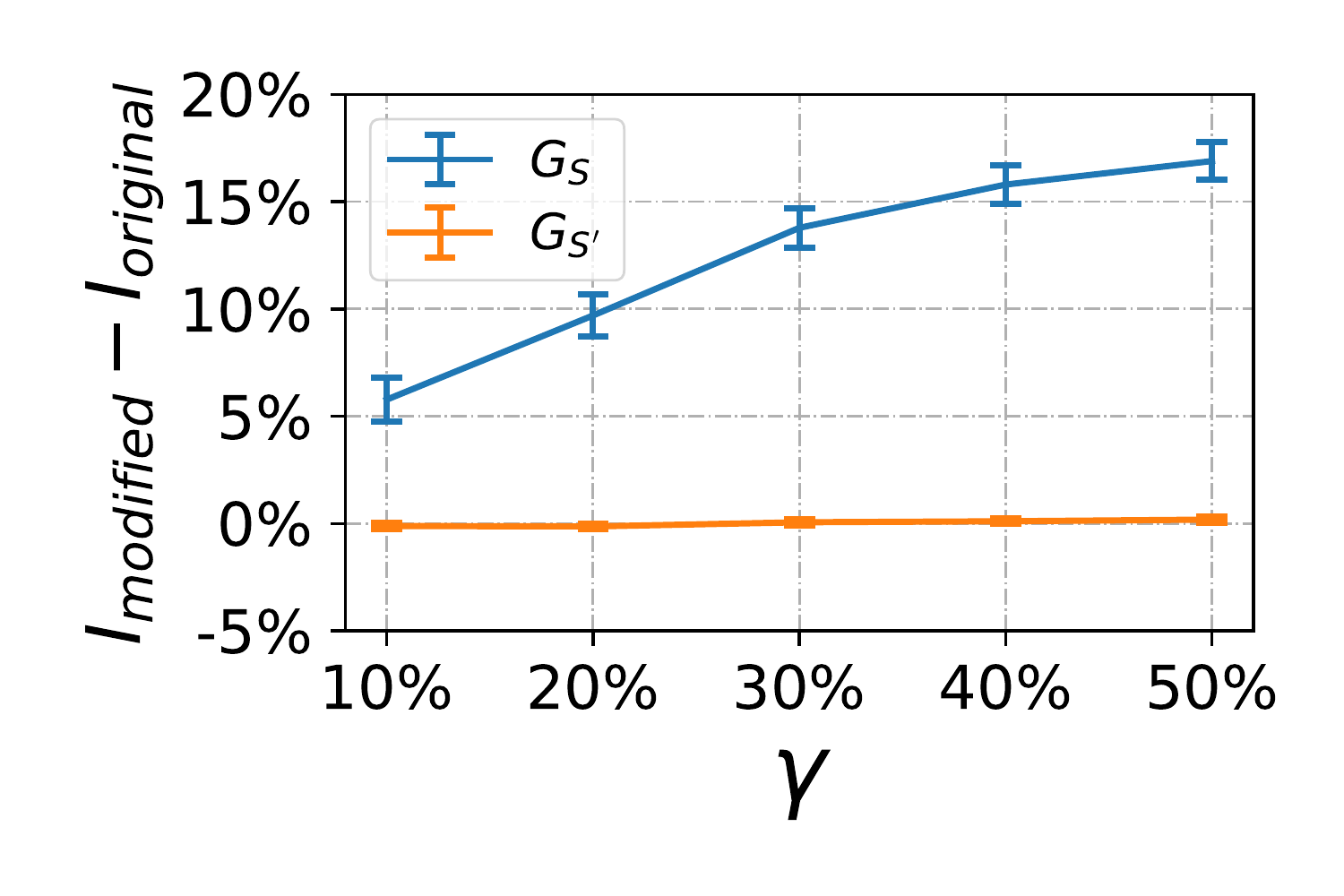} & \includegraphics[width=\FigSize]{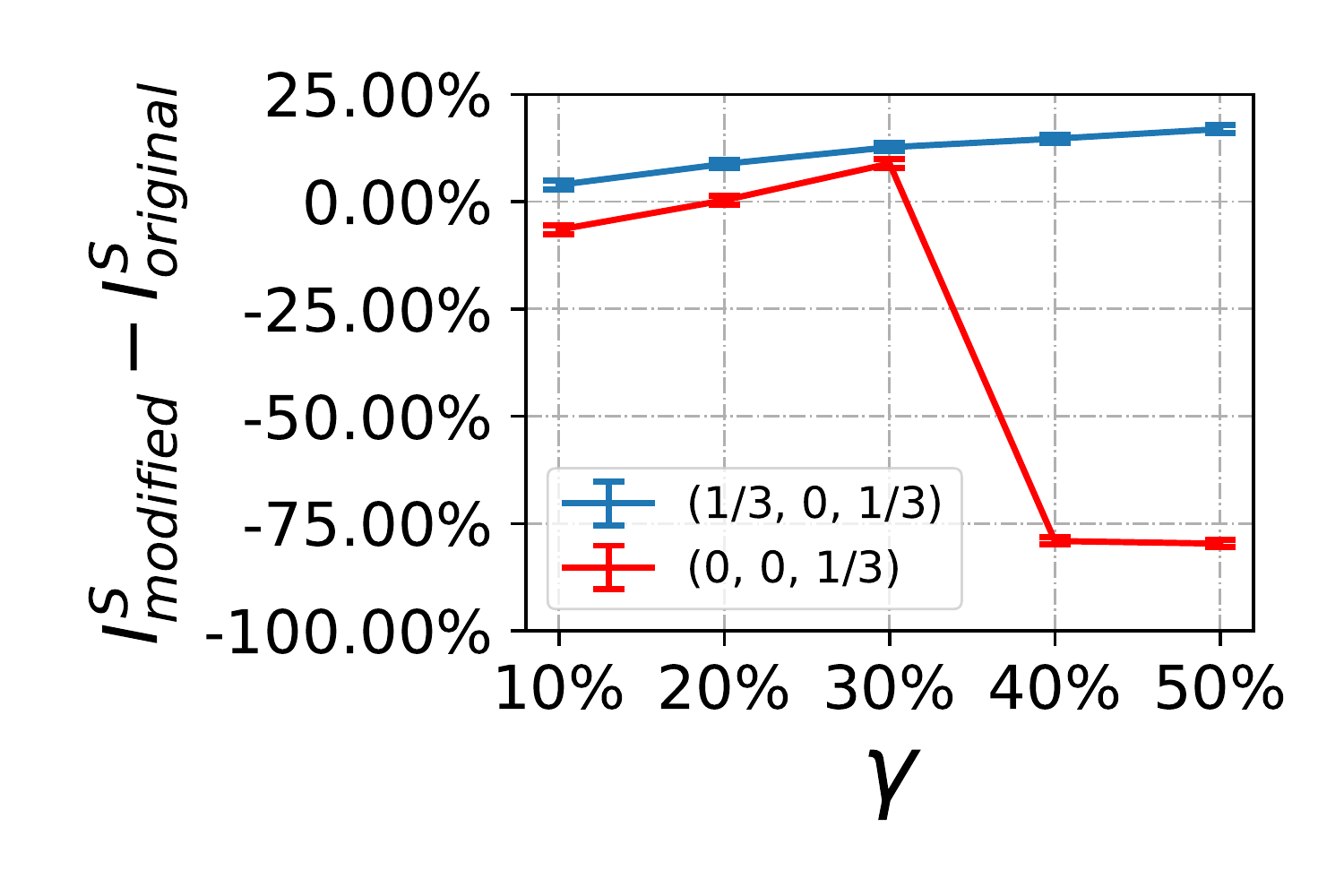} & \includegraphics[width=\FigSize]{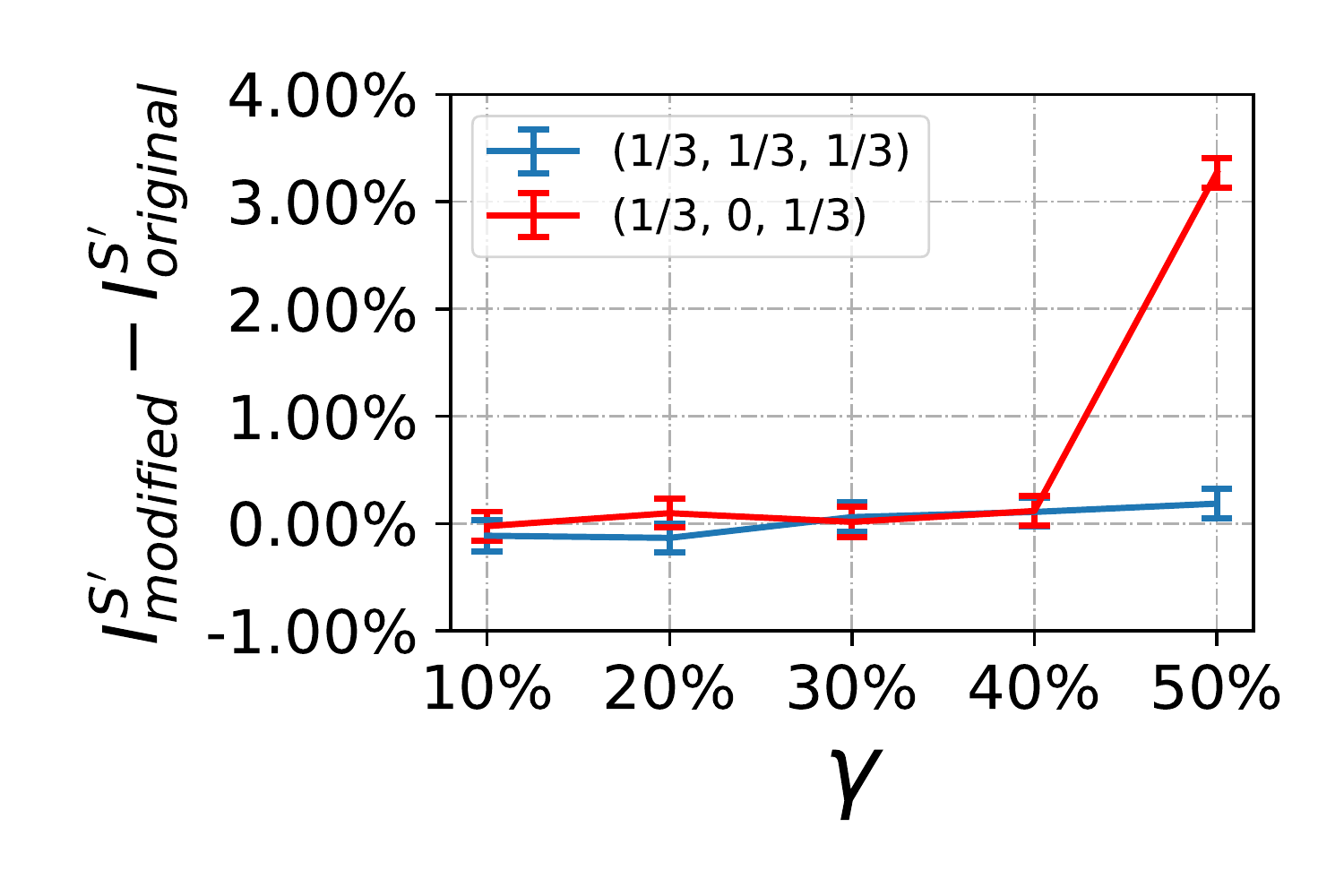} & \includegraphics[width=\FigSize]{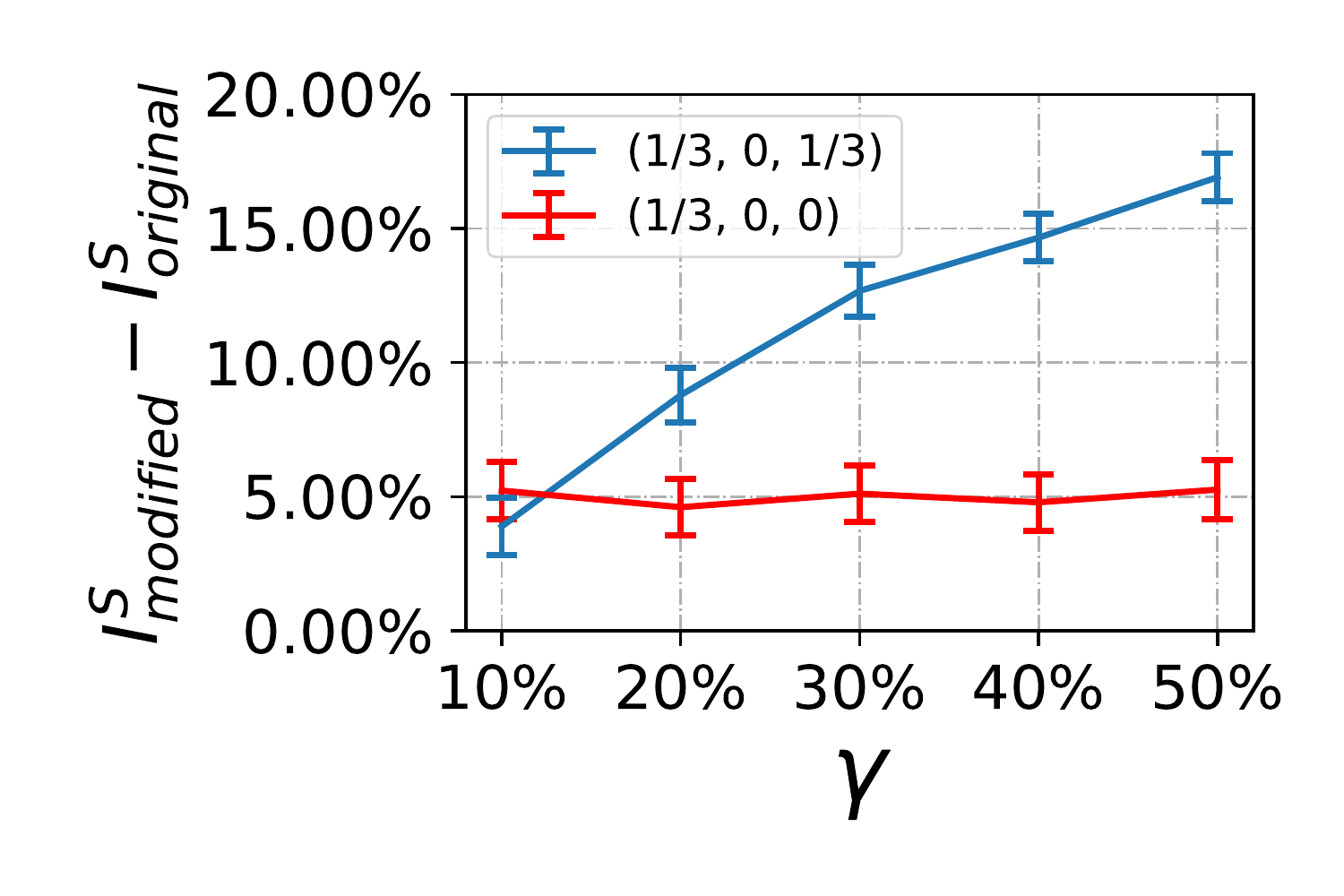} \\
\includegraphics[width=\FigSize]{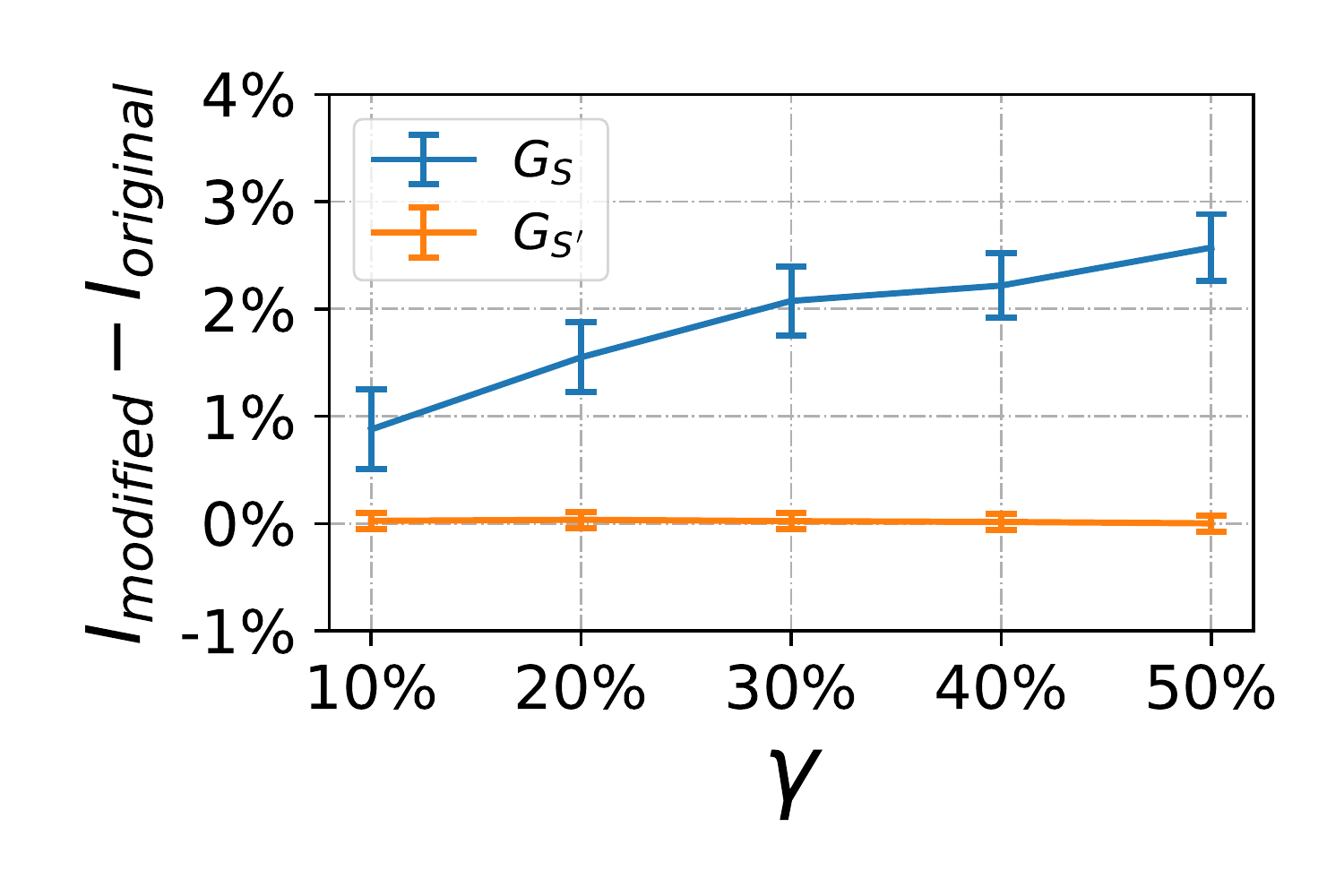} & \includegraphics[width=\FigSize]{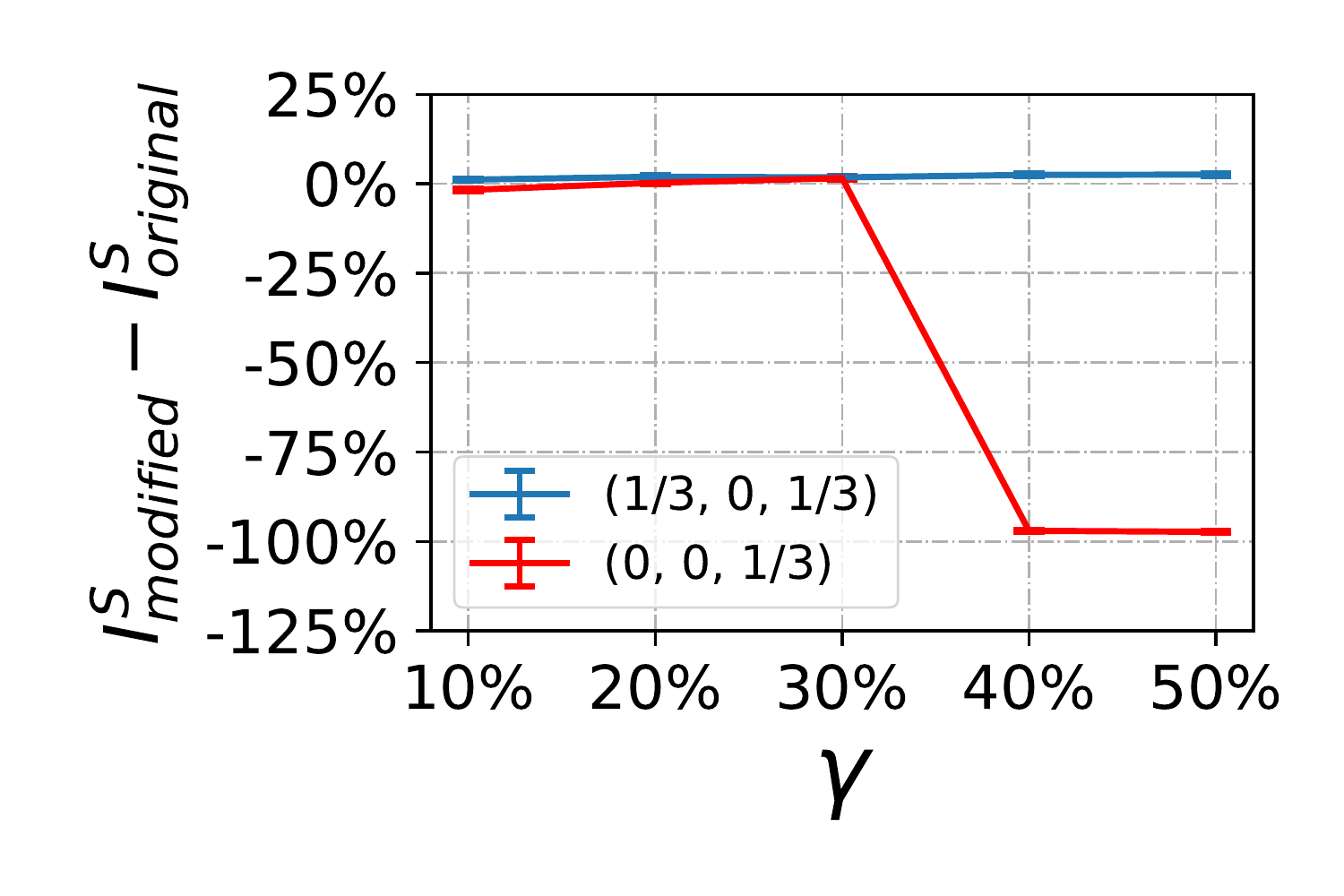} & \includegraphics[width=\FigSize]{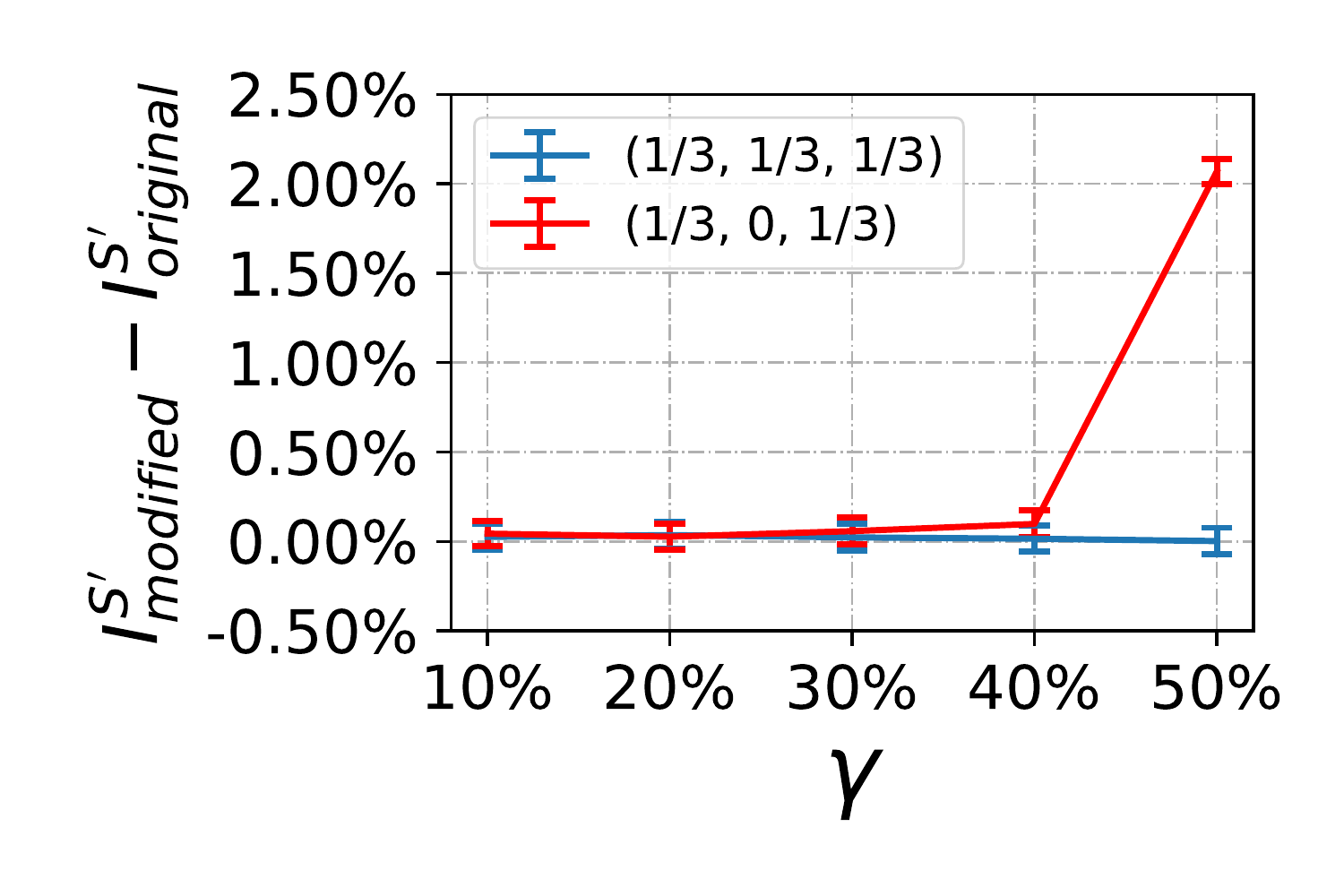} & \includegraphics[width=\FigSize]{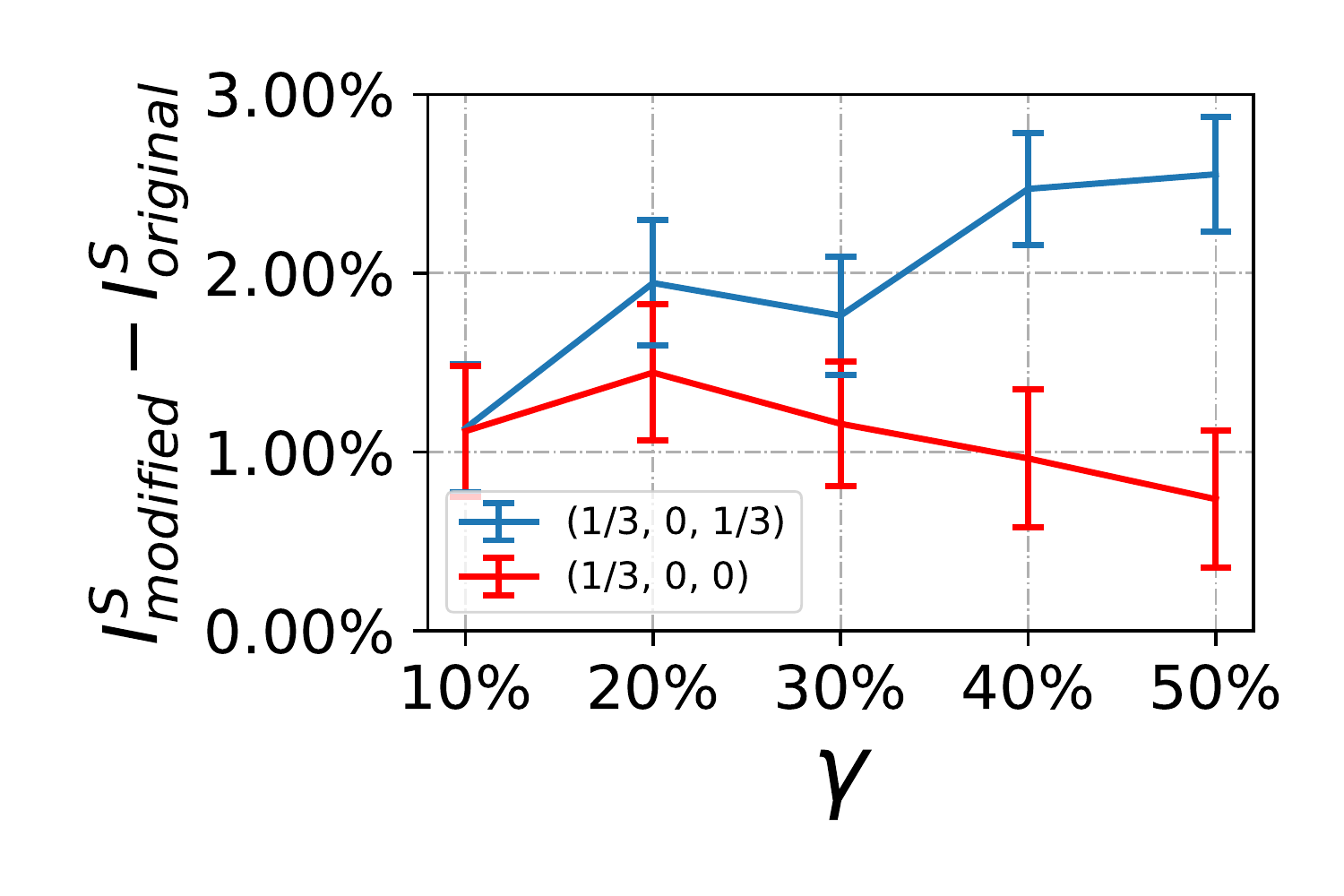} \\
\end{tabular}
\caption{Experimental results for different $\delta$ and $\beta$ values on the email network. \textbf{Top}: $\delta=0.5, \beta=0.1$; \textbf{Bottom}: $\delta=0.3, \beta=0.5$.}
\label{fig:email_diff_threshold}
\end{figure*}

\end{document}